\documentclass{lmcs} %%% last changed 2014-08-20
\pdfoutput=1

\usepackage{lastpage}
\lmcsdoi{20}{4}{15}
\lmcsheading{}{\pageref{LastPage}}{}{}%
{Feb.~01,~2023}{Nov.~20,~2024}{}


\usepackage{hyperref}
\theoremstyle{plain} %\crefname{satz}{Satz}{S\"atze}
\usepackage{ifthen}

\newboolean{talk}
\setboolean{talk}{false}
\newboolean{paper}
\setboolean{paper}{false}

\newboolean{LMCSstyle}
\setboolean{LMCSstyle}{false}
\newboolean{IEEEstyle}
\setboolean{IEEEstyle}{false}
\newboolean{lipicsstyle}
\setboolean{lipicsstyle}{false}
\newboolean{eptcsstyle}
\setboolean{eptcsstyle}{false}
\newboolean{sigplanstyle}
\setboolean{sigplanstyle}{false}
\newboolean{PACMPL}
\setboolean{PACMPL}{false}
\newboolean{acmartstyle}
\setboolean{acmartstyle}{false}

\newboolean{needstheorems}
\setboolean{needstheorems}{false}
\newboolean{withimages}
\setboolean{withimages}{false}
\newboolean{withproofs}
\setboolean{withproofs}{true}

\newboolean{french}
\setboolean{french}{false}

\setboolean{LMCSstyle}{true}
\ifthenelse{\boolean{talk}}{}{\usepackage[utf8]{inputenc}}
\ifthenelse{\boolean{PACMPL}}{}{
	\ifthenelse{\boolean{french}}{\usepackage[french]{babel}}{
	  \ifthenelse{\boolean{lipicsstyle}}{}{\usepackage[english]{babel}}}
	  }
\ifthenelse{\boolean{PACMPL}}{}{
	\ifthenelse{\boolean{acmartstyle}}{}{
		\usepackage{amssymb} 
	}
}
\usepackage{amsmath}
\usepackage{graphicx}
\usepackage{bussproofs}
\usepackage{fancybox}
\usepackage{cmll}
\usepackage{stmaryrd}
\usepackage{multirow}
\ifthenelse{\boolean{IEEEstyle}}{
	\usepackage[bookmarks={false}]{hyperref}
	}{
	\usepackage{hyperref}
	}
\usepackage{proof}
\usepackage{hhline}
\usepackage{xspace}
\usepackage{booktabs}
\usepackage{microtype}
\usepackage{wrapfig}
\usepackage{array}
\usepackage{tabularx}
\usepackage{arydshln}
\usepackage{commath}
\ifthenelse{\boolean{IEEEstyle}}{
\setboolean{needstheorems}{true}}{}

\ifthenelse{\boolean{eptcsstyle}}{
\setboolean{needstheorems}{true}}{}

\ifthenelse{\boolean{sigplanstyle}}{
\setboolean{needstheorems}{true}}{}

\ifthenelse{\boolean{needstheorems}}{
\input{\macrospath/thm-environments}}{}

\newcommand{\myproof}[1]{
\ifthenelse{\boolean{withproofs}}{#1}{}}

\newcommand{\withproofs}[1]{
\ifthenelse{\boolean{withproofs}}{#1}{}}

\newcommand{\withoutproofs}[1]{
\ifthenelse{\boolean{withproofs}}{}{#1}}

\newcommand{\tm}{t}
\newcommand{\tmtwo}{u}
\newcommand{\tmthree}{r}
\newcommand{\tmfour}{w}

\newcommand{\var}{x}
\newcommand{\vartwo}{y}
\newcommand{\varthree}{z}
\newcommand{\varfour}{w}

\newcommand{\Rew}[1]{\rightarrow_{#1}}

\renewcommand{\to}{\Rew{}}

\newcommand{\tob}{\Rew{\beta}}

\newcommand{\towh}{\Rew{wh}}

\newcommand{\symfont}[1]{\mathsf{#1}}

\newcommand{\db}{{\symfont{dB}}}

\newcommand{\vsym}{\symfont{v}}

\newcommand{\val}{v}

\newcommand{\toval}{\rightarrow_\vsym}

\newcommand{\ctxholep}[1]{[#1]}
\newcommand{\ctxhole}{\ctxholep{\cdot}}

\newcommand{\ctx}{C}

\newcommand{\ctxp}[1]{\ctx\ctxholep{#1}}

\newcommand{\evctx}{E}

\newcommand{\nbvctxtwo}[1]{\nbvctxtwo{#1}}

\newcommand{\defeq}{:=}

\newcommand{\grameq}{::=}

\newcommand{\isub}[2]{\{#1/#2\}}
\newcommand{\esub}[2]{[#1/#2]}
\newcommand{\llbrace}{\{ \kern -0.27em \vert}
\newcommand{\rrbrace}{\vert \kern -0.27em \}}

\renewcommand{\l}{\lambda}
\newcommand{\ie}{i.e.\xspace}
\ifthenelse{\boolean{LMCSstyle}}{}{
	
	}
\newcommand{\ih}{\textit{i.h.}\xspace}
\newcommand{\fv}[1]{\symfont{fv}(#1)}

\newcommand{\ignore}[1]{}

\newcommand{\myinput}[1]{\ifthenelse{\boolean{withimages}}{\input{#1}}{}}

\newcommand{\nat}{\mathbb{N}}

\newcommand{\size}[1]{|#1|}
\newcommand{\sizeparam}[2]{|#1|_{#2}}

\newcommand{\clos}{c}

\newcommand{\env}{E}
\newcommand{\envtwo}{\env'}

\newcommand{\lenv}{e}
\newcommand{\lenvtwo}{\lenv'}
\newcommand{\lenvthree}{\lenv''}

\newcommand{\stack}{S}

\ifthenelse{\boolean{PACMPL}}{\renewcommand{\state}{s}}{
	\ifthenelse{\boolean{acmartstyle}}{\renewcommand{\state}{s}}{
		\newcommand{\state}{s}
	}
}
\newcommand{\statetwo}{{\state'}}

\newcommand{\tospkam}{\mathrm{\Rew{Sp\KAM}}}
\newcommand{\tosplam}{\mathrm{\Rew{Sp\LAM}}}
\newcommand{\tonakam}{\mathrm{\Rew{Na\KAM}}}
\newcommand{\tooutkam}{\mathrm{\Rew{Out\KAM}}}

\newcommand{\exec}{\rho}

\newcounter{numberone}

\newcounter{numbertwo}

\renewcommand{\ctxholep}[1]{\langle #1\rangle}

\newcommand{\dom}[1]{\mathsf{dom}(#1)}

\newcommand{\unfsym}{\rotatebox[origin=c]{-90}{$\rightarrow$}}
\newcommand{\unf}[1]{#1\unfsym}

\newcommand{\deriv}{\rho}

\newcommand{\refprop}[1]{Prop.~\ref{prop:#1}}
\newcommand{\refpropp}[2]{Prop.~\ref{prop:#1}.\ref{p:#1-#2}}
\newcommand{\refsect}[1]{Sect.~\ref{sect:#1}}

\newcommand{\refdef}[1]{Def.~\ref{def:#1}}
\newcommand{\refthm}[1]{Theorem~\ref{thm:#1}}

\newcommand{\reffig}[1]{Fig.~\ref{fig:#1}}

\ifthenelse{\boolean{LMCSstyle}}{	
	\newcommand{\refequ}[1]{(\ref{eq:#1})}
	}{
	
	}

\renewcommand{\esub}[2]{[#1{\shortleftarrow}#2]}
\renewcommand{\isub}[2]{\{#1{\shortleftarrow}#2\}}

\newcommand{\run}{\rho}
\newcommand{\runtwo}{\sigma}

\newcommand{\kstatetab}[3]{#1 & #2 & #3 }
\newcommand{\kstate}[3]{(#1,#2,#3)}

\newcommand{\cons}{{\cdot}}

\newcommand{\mach}{\mathrm{M}}

\newcommand{\KAM}{KAM\xspace}
\newcommand{\NaKAM}{Naive KAM\xspace}
\newcommand{\STKAM}{Sub-Term KAM\xspace}
\newcommand{\OutKAM}{Outlined KAM\xspace}
\newcommand{\LinkKAM}{Linked KAM\xspace}
\newcommand{\SpKAM}{Space KAM\xspace}
\newcommand{\CoKAM}{Collecting KAM\xspace}
\newcommand{\LAM}{LAM\xspace}
\newcommand{\SpLAM}{Space LAM\xspace}

\newcommand{\tomachhole}[1]{\rightarrow_{#1}}
\newcommand{\tomach}{\tomachhole{}}

\newcommand{\stempty}{\epsilon}

\newcommand{\la}[1]{\lambda #1.}

\newcommand{\midd}{\; \; \mbox{\Large{$\mid$}}\;\;}

\newcommand{\rel}{\mathcal{R}}

\newcommand{\bigo}[1]{\mathcal{O}(#1)}

\newcommand{\ccallbn}{Closed Call-by-Name\xspace}
\newcommand{\cbn}{CbN\xspace}
\newcommand{\ccbn}{Closed \cbn}
\newcommand{\cbv}{CbV\xspace}
\newcommand{\ccbv}{Closed \cbv}
\newcommand{\cbneed}{CbNeed\xspace}

\newcommand{\Id}{\symfont{I}}

\newcommand\mydots{\hbox to .6em{.\hss.}}

\newcommand{\lm}[1]{\lambda \text{Meas}(#1)}
\renewcommand{\lm}[1]{\sizeparam{#1}{\mathsf{sp}}}

\newcommand{\str}{s}

\newcommand{\spkamstate}[3]{(#1,#2,#3)}
\newcommand{\splamstate}[4]{(#1,#2,#3,#4)}

\newcommand{\enc}[1]{\lceil#1\rceil}
\newcommand{\ems}{\varepsilon}

\newcommand{\detLam}{\Lambda_{\tt det}}
\newcommand{\tobdet}{\rightarrow_{det}}

\newcommand{\alpone}{\Sigma}
\newcommand{\cods}[1]{\overline{#1}}

\renewcommand{\dump}{d}
\newcommand{\dentry}[2]{#1{\diamond}#2}

\newcommand{\sizeb}[1]{\sizeparam{#1}{\beta}}

\newcommand{\seasym}{\symfont{sea}}
\newcommand{\seavsym}{\symfont{sea}_{\symfont{v}}}
\newcommand{\seanvsym}{\symfont{sea}_{\neg \symfont{v}}}
\newcommand{\subsym}{\symfont{sub}}
\newcommand{\bwsym}{\beta_{\symfont{w}}}
\newcommand{\bnwsym}{\beta_{\neg\symfont{w}}}
\newcommand{\retsym}{\symfont{ret}}

\newcommand{\tokamsea}{\tomachhole{\seasym}}
\newcommand{\tokamb}{\tomachhole{\beta}}
\newcommand{\tokamsub}{\tomachhole{\subsym}}
\newcommand{\tokamseav}{\tomachhole{\seavsym}}
\newcommand{\tokamseanv}{\tomachhole{\seanvsym}}
\newcommand{\tokambw}{\tomachhole{\bwsym}}
\newcommand{\tokambnw}{\tomachhole{\bnwsym}}
\newcommand{\tokamret}{\tomachhole{\retsym}}

\newcommand{\sizetime}[1]{\sizeparam{#1}{T}}
\newcommand{\sizesub}[1]{\sizeparam{#1}{\subsym}}
\newcommand{\sizesea}[1]{\sizeparam{#1}{\seasym}}

\newcommand{\heap}{h}
\newcommand{\heaptwo}{\heap'}

\newcommand{\compil}[1]{\symfont{init}(#1)}

\newcommand{\fixnospace}{\symfont{fix}}
\newcommand{\fix}{\fixnospace\,}
\newcommand{\M}{\mathcal M}
\newcommand{\cont}{k}

\newcommand{\inits}{{\tt init}}

\newcommand{\elemblank}{\Box}
\renewcommand{\state}{q}
\newcommand{\States}{Q}
\newcommand{\statein}{\state_{\mathit{in}}}
\newcommand{\statefint}{\state_{\mathit{T}}}
\newcommand{\statefinf}{\state_{\mathit{F}}}

\newcommand{\config}{C}
\newcommand{\configtwo}{D}

\newcommand{\finals}{{\tt final}}

\newcommand{\transaux}{{\tt transaux}}

\newcommand{\transs}{{\tt trans}}

\newcommand{\strone}{s}
\newcommand{\strtwo}{r}

\newcommand{\append}{{\tt{append}}}

\newcommand{\appendchar}[1]{\append^{#1}}

\newcommand{\elone}{a}
\newcommand{\elem}{a}
\renewcommand{\enc}[1]{\overline{#1}}
\renewcommand{\cods}[1]{\lceil#1\rceil}
\newcommand{\cod}[2]{\cods{#1}^{#2}}

\newcommand{\initconfigs}{\config_{\tt in}(\inputstr)}

\newcommand{\tomachtur}{\tomachhole{\M}}

\newcommand{\inputstr}{i}
\newcommand{\counter}{n}
\newcommand{\Statesfin}{\States_{\mathit{fin}}}
\newcommand{\statefin}{\state_{\mathit{fin}}}

\newcommand{\tuple}[1]{\langle #1 \rangle}

\newcommand{\succl}{\mathtt{succ}}

\newcommand{\lookupl}{\mathtt{lookup}}

\newcommand{\too}{\rightarrow}
\newcommand{\ntostr}[1]{\lfloor #1 \rfloor}

\newcommand{\kop}{\textsc{k}}

\newcommand{\Bool}{\mathbb{B}}
\newcommand{\Boolb}{\Bool_{\mathsf{W}}}
\newcommand{\Boolin}{\Bool_{\mathsf{I}}}
\newcommand{\csep}{\,|\,}
\newcommand{\encp}[2]{\enc{#1}^{#2}}
\newcommand{\bool}{b}

\newcommand{\wstr}{w}
\newcommand{\wstrl}{\wstr_{l}}
\newcommand{\wstrr}{\wstr_{r}}

\newcommand{\istr}{i}

\newcommand{\varf}{f}

\renewcommand{\L}{\mathsf{L}}

\newcommand{\PSPACE}{\mathsf{PSPACE}}

\renewcommand{\P}{\mathsf{P}}

\newcommand{\EXP}{\mathsf{EXP}}

\newcommand{\toy}{\mathtt{toy}}
\newcommand{\toyaux}{\mathtt{toyaux}}

\newcommand{\localcopy}{\mathtt{loCpy}}

\newcommand{\globalcopy}{\mathtt{glCpy}}

\newcommand{\addr}{a}
\newcommand{\addrtwo}{\addr'}

\newcommand{\codesize}[1]{\norm{#1}}

\newcommand{\stsym}{\mathsf{L}}
\newcommand{\ensym}{\mathsf{R}}

\newcommand{\tox}{\Rew{X}}
\renewcommand{\tox}{\Rew{\symfont{str}}}

\newcommand{\sizecl}[1]{\sizeparam{#1}{{\mathsf{cl}}}}

\newcommand{\staddr}[2]{#1|_{#2}}

\newcommand{\sizespace}[1]{\sizeparam{#1}{\mathsf{sp}}}
\newcommand{\sizeaspace}[1]{\sizeparam{#1}{\mathsf{clsp}}}

\newcommand{\pointer}[1]{p_{#1}}

\usepackage{microtype}
\usepackage{version}

\renewcommand{\stack}{\pi}

\renewcommand{\sizetime}[1]{\sizeparam{#1}{\mathsf{tm}}}
\renewcommand{\refsect}[1]{Section~\ref{sect:#1}}

\includeversion{SHORT}
\excludeversion{LONG}

\begin{document}

\title[Reasonable Space for the $\lambda$-Calculus, Logarithmically]{Reasonable Space for the $\lambda$-Calculus, Logarithmically$^*$}
\titlecomment{$^*$This is the extended version of the paper appeared with the same title at LICS2022~\cite{DBLP:conf/lics/AccattoliLV22}}
\author[B.~Accattoli]{Beniamino Accattoli\lmcsorcid{0000-0003-4944-9944}}[a]
\author[U.~Dal Lago]{Ugo Dal Lago\lmcsorcid{0000-0001-9200-070X}}[b]
\author[G.~Vanoni]{Gabriele Vanoni\lmcsorcid{0000-0001-8762-8674}}[c]

% affiliation 1 (automatically numbered a)
\address{Inria \& LIX, \'Ecole Polytechnique, UMR 7161, France}	%optional
% write emails for all authors having that affiliation
\email{beniamino.accattoli@inria.fr}  %optional

% affiliation 2 (automatically numbered b)
\address{Universit\`a di Bologna, Italy \& Inria, France}	%optional
\email{ugo.dallago@unibo.it}  %optional

% affiliation 2 (automatically numbered b)
\address{IRIF, CNRS, Université Paris Cité, France}	%optional
\email{gabriele.vanoni@irif.fr}  %optional

%% etc.

%% required for running head on odd and even pages, use suitable
%% abbreviations in case of long titles and many authors:

%%%%%%%%%%%%%%%%%%%%%%%%%%%%%%%%%%%%%%%%%%%%%%%%%%%%%%%%%%%%%%%%%%%%%%%%%%%

%% the abstract has to PRECEDE the command \maketitle:
%% be sure not to issue the \maketitle command twice!

\begin{abstract}
  \noindent Can the $\lambda$-calculus be considered a reasonable computational 
model? Can we use it for measuring the time \emph{and} space consumption of 
algorithms? While the literature contains positive answers about time, much 
less is known about space. This paper presents a new reasonable space cost 
model for the $\lambda$-calculus, based on a variant over the Krivine abstract 
machine. For the first time, this cost model is able to accommodate logarithmic 
space.
Moreover, we study the time behavior of our machine, which is unreasonable but it can be turned into a reasonable one using known techniques. Finally, we show how to transport 
our results to the call-by-value $\lambda$-calculus.
\end{abstract}

\maketitle

% !TeX spellcheck = en_US
% !TEX root = main.tex
%%%%%%%%%%%%%%%%%%%%%%%%%%%%%%%%%%%%%%%%%%%%%%%%%%%%%%%%%%%%%%%%%%%%%
\section{Introduction}
Bounding the amount of resources needed by algorithms and 
 programs is a fundamental problem in computer science. Here we are concerned 
with space. In many applications, say, stream processing or web 
crawling, \emph{linear} bounds on computing space are not satisfactory, given 
the enormous amount of data processed. Therefore \emph{logarithmic} bounds become the standard of reference. Theoretically, complexity classes 
such as the class $\L$ of logarithmic space, although apparently small, are very 
interesting, and it is not known whether they are distinct from $\P$, see for instance Hopcroft and Ullman \cite{DBLP:books/aw/HopcroftU79}.

\paragraph{Space and the $\l$-Calculus.} Dealing with space bounds in the $\l$-calculus, or in functional 
programming languages, has always been considered a challenge.

The first reason is related to the special role of garbage collection in functional languages and in the $\l$-calculus. Space usage that is linearly related to time is the worst possible usage in sequential models, since space is always bounded by time, given that---intuitively---using a unit of space requires a unit of time. In a purely functional setting without garbage collection, space is indeed linearly related to time. Therefore, to properly studying space requires explicitly taking into account garbage collection, which instead is exactly one of the aspects that functional programming aims at \emph{hiding}, leaving it to the meta-level. The case of the $\l$-calculus is slightly different. The $\beta$-reduction rule can erase sub-terms, but there is not much control over this form of erasure, because it is not  asynchronous as in functional languages, since---technically---erasing $\beta$-steps cannot be postponed. Consider indeed the following sequence:
\begin{center}
$\begin{array}{lllllll}
(\la\var\la\vartwo\vartwo)\tm\tmtwo &\tob& (\la\vartwo\vartwo)\tmtwo& \tob& \tmtwo
\end{array}$
\end{center}
The first step is erasing but it \emph{cannot} be postponed after the second one, because the second step is \emph{created} by the first one.

The second reason behind the challenge in studying space is that the abstract notions of time and space in the $\l$-calculus have some 
puzzling properties, as we shall discuss at length. In particular, there are families of terms where the abstract notion of space seems to be \emph{exponential} in the abstract notion of time, a phenomenon known as \emph{size explosion}. This puzzling fact, which roughly means that \emph{one does not need a unit of time to use a unit of space}, eclipses the issues related to garbage collection, and historically was the main reason why $\l$-calculus was not considered a good setting for computational complexity.

\paragraph{Logarithmic Space and the $\l$-Calculus.} But logarithmic space is special, because it adds a further difficulty to an already challenging topic: it requires \emph{log-sensitivity}, that is, to distinguish between \emph{input space} and \emph{work space}---without such a distinction one cannot even measure sub-linear space. Since the $\l$-calculus does not distinguish between \emph{programs} and \emph{data}, it is \emph{log-insensitive} and, apparently, at odds with logarithmic space.

%Finally, sub-linear space requires re-using space, that is, having some \emph{garbage 
%collection}, which is not abstractly modeled by the $\l$-calculus. 

The literature about the $\l$-calculus does nonetheless offer 
results about space complexity, but they are all \emph{partial}, as they either 
concern logarithmic space for
\emph{variants} of the $\l$-calculus, as for Dal Lago and Sch\"opp \cite{bllspace,dal_lago_computation_2016}, 
Mazza \cite{DBLP:conf/csl/Mazza15} and Ghica~\cite{ghica_geometry_2007}, or they do deal with the $\l$-calculus but apply only to linear space and above, as it is the case of Forster et al. \cite{cbv_reasonable}.

\paragraph{Contribution} The main result of this paper is the first fully fledged space 
reasonability result for the pure, untyped $\l$-calculus, accounting for logarithmic space. Precisely, we retrieve log-sensitivity by 
representing the \emph{input space} as $\l$-terms, and the \emph{work space} as 
the space used by a new variant of the well-known Krivine abstract machine (KAM) \cite{krivine_call-by-name_2007} that we dub \emph{\SpKAM}. We then prove that such a notion of work space is \emph{reasonable} and it accounts for logarithmic space. Reasonable, roughly, means equivalent to the notion of space of Turing machines (shortened to TMs). More accurately, we show that there are encodings of TMs into the $\l$-calculus and vice-versa inducing simulations with a linear space overhead\footnote{The notion of \emph{linear} space overhead with respect to sub-linear space might be confusing: if a TM uses work space $\bigo{\log n}$, where $n$ is the size of the input, then the simulation in the \SpKAM must use $k\cdot\log(n)\in\bigo{\log n}$ work space, and not $\bigo{n}$ space. In other words, what is linear is the overhead, not the function describing the space consumption.}. For space, the tricky simulation is the one of TMs into the $\l$-calculus (for time, it is the opposite one), which is studied in great detail in this paper. The other simulation is only outlined, as it does not present any difficulty. \refsect{theory-cost} contains an original perspective on the theory of reasonable cost models for the $\l$-calculus and of how our result fits in.

\paragraph{Key Ingredients} Our result follows from a careful dissection and refinement of the KAM and of the simulation of TMs into the $\l$-calculus. A peculiar aspect is that the result does not rest on a \emph{single} innovation or idea. It rests instead on the simultaneous addressing of \emph{numerous} critical points of the encoding of TMs and of the KAM. None of them is in itself difficult or striking---apart perhaps from the disabling of environment sharing discussed below---but all of them are mandatory for the result to hold. We identify six critical points, of which we provide an overview in \refsect{critical-points}. Four of them are of a high-level nature:
\begin{enumerate}
\item \emph{Eager garbage collection}: environment-based abstract machines such as the KAM are usually presented without garbage collection, assuming that the meta-level shall take care of it, as it is customary in functional programming. For a parsimonious use of space, it is instead essential to re-use space as much as possible, thus having a first-class treatment of garbage collection. From the point of view of cost models, this change disentangles space from time (more precisely, from time considered as the number of $\beta$-steps). Concretely, we introduce a new variant of the KAM with eager garbage collection (plus another optimization) dubbed \CoKAM.
\item \emph{Disabling data structure sharing}: there are two different forms of sharing at work in the KAM that are usually not distinguished, namely the \emph{sharing of sub-terms} provided by 
environments, and the \emph{sharing of environments} themselves. These aspects, actually, are not explicit in the  specification of the KAM, they are left to the concrete implementation of the KAM. The \SpKAM---which is a specific concrete implementation of the \CoKAM---
adopts the former but forbids the latter. Abstractly, this is needed to turn the data structures of the KAM into sort of \emph{tapes} of TMs. Tapes are special in that they are \emph{flat}, that is, cells are juxtaposed without using space, rather than linked via pointers, which would add a space overhead. Similarly, then, the data structures of the \SpKAM are flat, which has the consequence of forbidding the sharing of environments. This is probably the most surprising point of our work, and---to our knowledge---the first time that such an approach is adopted in the literature on abstract machines for the $\l$-calculus.
\item \emph{Encoding and moving over tapes}: designing the \SpKAM is only half of the story. The other half is the 
refinement of the encoding of TMs into the $\l$-calculus. 
Our reference is the encoding by Dal Lago and Accattoli~\cite{DBLP:journals/corr/abs-1711-10078}, which uses a linear amount of extra space to simulate the moving over TMs tapes. This is particularly bad for the input tape, for which a logarithmic overhead is required. Therefore, we change the encoding of input tapes, exploiting their read-only nature, to achieve the required overhead.
\item \emph{Low-level complexity analysis and left addresses}: the complexity analysis of the space used by the \SpKAM to execute the encoding of TMs is not \emph{abstract}, that is, it is not simply given by the maximum number of pointers to the code times the logarithm of the code. It is \emph{low-level} in the sense that one has to inspect the size of pointers, and the reasonable bound holds only because some of them turn out to have constant size. This is obtained via a specific addressing scheme for pointers to the code, what we call \emph{left addresses}.
\end{enumerate}

\paragraph{Looking Back} Understanding the whole picture, at both the high-level of the theory of cost models and the low-level of machines simulations, is far from obvious. The following sections shall strive to provide such a picture. The long-standing riddle of logarithmic reasonable space for the $\l$-calculus, however, turns out to have a relatively simple solution. Roughly, it is enough to add two simple space optimizations (eager garbage collection plus \emph{environment unchaining}) to the KAM, obtaining the \CoKAM, and to further disable its sharing of environments, obtaining the \SpKAM, if the new encoding of TMs is taken for granted. Our \SpKAM, indeed, is not much more complex than the KAM itself.

The difficulty behind the quest for a logarithmic reasonable work space is better understood by considering that it required dismissing two widespread intuitions about the problem, as we shall now explain. Additionally, the literature about abstract machines of the $\l$-calculus is mainly concerned with time, for which the issues for space are irrelevant, and are therefore often treated in an ambiguous and inaccurate way. 

Somewhat unusually, the main obstacle behind the solution of the long-standing problem turned out to be reasoning without the preconceptions and the inaccuracies of the established knowledge about space for the $\l$-calculus and the theory of abstract machines.

\paragraph{A Wrong Positive Belief: the Geometry of Interaction} 
For 15 years, logarithmic reasonable space was believed to be connected to the alternative execution schema offered by Girard's \emph{geometry of interaction} \cite{girard_geometry_1989}. Mackie's and Danos \& Regnier's \emph{interaction abstract machine} (shortened 
 to IAM) \cite{mackie_geometry_1995,danos_regnier_1995}, recently reformulated by Accattoli et al. in~\cite{DBLP:conf/ppdp/AccattoliLV20}, is a machine rooted in the geometry of interaction and in Abramsky 
 et al.'s game semantics \cite{DBLP:journals/iandc/AbramskyJM00}. It is based on a log-sensitive approach, 
 and---apparently---it is parsimonious with respect to space. 
 Sch\"opp  \cite{DBLP:conf/csl/Schopp06,bllspace} (with later developments with Dal Lago \cite{DBLP:conf/esop/LagoS10}) was the first one to show how IAM-like mechanisms can be used for dealing with 
 logarithmic space. It was since then conjectured that the space of the IAM were a 
 reasonable cost model. The belief in the conjecture was reinforced by further 
 uses of IAM-like mechanisms for space parsimony related to circuits, by 
 Ghica~\cite{ghica_geometry_2007}, and for characterizing $\L$, by Mazza 
 \cite{DBLP:conf/csl/Mazza15}. In 2021, however, Accattoli et al.  
 essentially \emph{refuted} 
 the conjecture: the space used by the IAM to evaluate the reference encoding 
 of TMs is unreasonable \cite{DBLP:conf/lics/AccattoliLV21} (as well as time 
 inefficient \cite{ADLVPOPL21}). While one might look for different encodings, 
 the 
 unreasonable behavior of the IAM concerns the modeling of recursion via 
 fix-point combinators (or self application) which is a cornerstone of the $\l$-calculus, hardly avoidable by any alternative encoding.

\paragraph{A Wrong Negative Belief: The Space Cost of Environments} Another misleading 
belief was that environment-based abstract machines could not provide reasonable notions of  space. Environments are data structures used to achieve time 
reasonability. According to Fernandez and Siafakas \cite{DBLP:journals/entcs/FernandezS09}, there are two main styles of environments, 
\emph{local} and \emph{global}, studied in-depth by Accattoli and Barras 
\cite{DBLP:conf/ppdp/AccattoliB17}. Global environments (as in the Milner Abstract Machine  
\cite{DBLP:conf/icfp/AccattoliBM14}) are log-insensitive because they work 
over the input space. Local environments (as in the KAM) are log-sensitive.
There are two reasons why their usual presentation is space unreasonable.

 Firstly, garbage collection is not usually accounted for, which leads to ever-increasing space usage, while reasonable space should be re-usable. This issue is easily solved by adding garbage collection. One actually needs an eager form, in order to maximize space re-usability, and to implement it in a naive, time-ineffiecient way, as time efficient techniques such as reference counters would add an unreasonable space overhead.

Secondly, and more subtly, environments are usually space unreasonable because of the use of pointers for sharing. To be precise, local environments use \emph{two} types of 
pointers, handling two forms of sharing: sub-term pointers, which serve to 
avoid
copying 
sub-terms, and environment pointers, which both realize their linked list 
structure 
and the sharing of environments. Sub-terms pointers are a key aspect of logarithmic space computations, and are thus crucial. Environment pointers are instead what 
makes environments space unreasonable (despite being, according to Douence and Fradet \cite{DBLP:journals/lisp/DouenceF07},
the \emph{essence} of the KAM): they introduce a logarithmic 
\emph{pointer overhead} that, at best, gives simulations of TMs with a 
$\bigo{n\log n}$ overhead in space, instead of the required $\bigo n$ for 
reasonability.
It was then generally concluded that environments cannot provide reasonable space. 

We here show that, instead, environments make perfect sense also without environment pointers, and so without sharing of environments, by implementing them simply as strings of adjacent symbols and copying their whole content instead of copying only pointers to them---these shall be referred to as \emph{flat environments}. Such an unusual approach is one of the critical ingredient for space reasonability. At the same time, adopting flat environments has two correlated consequences. Firstly, it breaks \emph{time reasonability} (with respect to time considered as the number of $\beta$-steps), as we show in \refsect{unreas-time}. In \refsect{time-spkam}, we discuss how to recover simultaneous reasonability for both time and space.

Secondly, for some terms it leads to extreme inefficiencies for both time and space. Namely, this happens for terms that are \emph{not} encodings of TMs and for which environment sharing provides a speed-up, that can even be exponential for both time and space. While at first sight this fact might seem startling, it is in fact well known that \emph{reasonable} does not mean \emph{efficient}: it only means that the theory can \emph{simulate and be simulated by TMs with negligible overhead}. See Accattoli \cite{DBLP:journals/entcs/Accattoli18} for discussions of this delicate point with respect to time.

\paragraph{Pointers, Abstract Machines, and Abstract Implementations} The literature on abstract 
machines for the $\l$-calculus usually assumes that pointers are used in implementations. Still, pointers are not usually  
 explicitly accounted for in abstract machine specifications. Such an ambiguity can be both positive and negative. On the positive side, it allows one to omit details that might be irrelevant for the intended result, obtaining simpler machines. On the negative side, it makes such specifications ambiguous and prevents precise cost analyses. Different implementations of the machine can indeed be possible, sometimes with very different asymptotic complexities, and forms of sharing can be unhygienically hidden behind the meta-level assumption that some of the machine components are represented via pointers. This is for instance the case of the KAM, which is time reasonable (for time = \# of $\beta$ steps) \emph{only if} environments are shared---as we here show for the first time, providing an example of unreasonable time overhead in absence of environment sharing (\refprop{NaKAM-time-unreas})---even if such sharing is not explicit in the specification of the KAM.

Taking into account pointers and their size is mandatory for the study of space, and even more so for logarithmic space. Therefore, we refine abstract machines by adding specifications of the time and space cost for each component, what we dub \emph{abstract implementations}. It is a methodological contribution of this work to the theory of abstract machines. It is required for the space reasonability result, but we believe that its value is independent of it.

\paragraph{Encoding of TMs and Call-by-Value} Beyond the design of the \SpKAM, our other main contribution is a new encoding of TMs in the $\l$-calculus. A critical point, as already mentioned, is that we change the representation of the input tape, in order to achieve the required logarithmic overhead. A further point is that the new encoding is carefully designed so as to retain the  
 \emph{indifference property} of the reference one, i.e. the fact that 
 it behaves the same under both call-by-name and call-by-value  evaluation. We then build over this design choice by showing that our results smoothly transfer to call-by-value  evaluation.
 
This is in contrast to what happens for time. The study of reasonable time for the $\l$-calculus is also based on a strategy-indifferent encoding of TMs, but in that case the difficult direction is the other one, that is, the simulation of the $\l$-calculus on TMs. To obtain such reasonable simulations, different strategies require different treatments. It turns out, then, that reasonable space can be studied more uniformly than reasonable time.

\paragraph{Sub-Term Property} The techniques for reasonable time and reasonable 
space seem to be at odds, as they make essential but opposite uses of linked 
data structures. Both techniques, however, crucially rely on the \emph{sub-term 
property} of abstract machines, that is, the fact that duplicated terms are 
sub-terms of the initial one. For time, it allows one to bound the cost of 
duplications, while for space it allows one to see sub-terms as (logarithmic) 
pointers to the input. The sub-term property seems to be the \emph{unavoidable 
ingredient for reasonability in the $\l$-calculus}. For extensive discussions about the sub-term property, see Accattoli \cite{DBLP:journals/lmcs/Accattoli23}, particularly Section 3 therein.

\paragraph{Related Work: Safe for (Reasonable Logarithmic) Space} Disabling sharing environments also plays a crucial role in a work about space by Paraskevopoulou and Appel \cite{DBLP:journals/pacmpl/Paraskevopoulou19}. They study \emph{closure conversion}, a program transformation at work in compilers for functional languages, turning abstractions into \emph{closures}, that is, into pairs of transformed abstractions and \emph{environments}. In the compilers literature, \emph{closures} and \emph{environment} refer to concepts that are similar and yet different with respect to those used in the abstract machine literature. Despite the differences (not discussed here), one can see some analogies between \cite{DBLP:journals/pacmpl/Paraskevopoulou19} and our work. The problem studied in \cite{DBLP:journals/pacmpl/Paraskevopoulou19} is which data structure for the (compiler) closures/environments of transformed programs is \emph{safe for space}, that is, preserves the space used by the source program. They show that \emph{flat environments} are safe for space, where \emph{flat} means without environment sharing. Playing with their slogan, one might then say that \emph{flat environments are safe for reasonable logarithmic space}.

In \cite{DBLP:journals/pacmpl/Paraskevopoulou19}, it is stressed that linked environments are \emph{not} safe for space because different environments might share sub-environments, preventing some garbage to be collected. Our study stresses a different danger, namely the pointer overhead introduced by the linked representation, which is unreasonable for logarithmic space. By removing pointers, we also remove the sharing of sub-environments. In \cite{DBLP:journals/pacmpl/Paraskevopoulou19}, flat environments are \emph{records}, that are assumed to not be linked via pointers, so the two approaches agree. Simply, the speech in \cite{DBLP:journals/pacmpl/Paraskevopoulou19} does not mention the danger of the pointer overhead, which is instead crucial for us. To avoid misunderstandings, we stress that both works study flat environments but the studied problems and the used techniques are incompatible: closure converted terms (with flat environments) can have size quadratically bigger than source $\l$-terms, so that closure conversion cannot be used to study reasonable space.

\paragraph{Related Work} The space inefficiency of 
environment machines is also  observed by
Krishnaswami et al.~\cite{DBLP:conf/popl/KrishnaswamiBH12}, who propose techniques to
alleviate it in the context of functional-reactive programming and
based on linear types. A characterization of $\PSPACE$ in the $\l$-calculus is 
given by Gaboardi et al. \cite{DBLP:journals/tocl/GaboardiMR12}, but it relies 
on alternating time rather than on a notion of space. The already-mentioned works by Dal Lago and Sch\"opp \cite{bllspace,dal_lago_computation_2016} and 
Mazza \cite{DBLP:conf/csl/Mazza15} characterize $\L$ in variants of the $\l$-calculus, while Jones characterizes $\L$ using a programming language but not based on the $\l$-calculus \cite{DBLP:journals/tcs/Jones99}. Blelloch and 
coauthors study in various papers how to profile (that is, measure) space consumption of 
functional programs~\cite{DBLP:conf/fpca/BlellochG95,DBLP:conf/icfp/BlellochG96, 
DBLP:journals/jfp/SpoonhowerBHG08}, also done by Sansom and Peyton Jones \cite{DBLP:conf/popl/SansomJ95}. They do not study, however, the reasonability of the cost models, that is, the equivalence with the space of TMs, which is the difficult 
part of our work. Finally, there is an extensive literature on garbage 
collection, as witnessed by the dedicated handbook~\cite{DBLP:books/wi/Jones2011}. We 
here need a basic eager form, that need  not be time efficient, as the \SpKAM is time unreasonable anyway.
%\medskip
%
%\noindent\emph{\textbf{Proofs.}} The proofs are in the \begin{SHORT}associated technical report~\cite{DBLP:journals/corr/abs-2203-00362}\end{SHORT}\begin{LONG}appendix\end{LONG}.

\section{The Theory of Reasonable Cost Models for the \texorpdfstring{$\l$}{λ}-Calculus,\\ and How Our Result Fits in It}
\label{sect:theory-cost}
\paragraph{Reasonable Cost Models} According to the seminal work by Slot and 
van Emde Boas \cite{DBLP:journals/iandc/SlotB88,vanEmdeBoas90}, the adequacy of space and 
time cost models is judged in 
relationship to whether they reflect the corresponding cost models of TMs, the computational theory\footnote{We prefer \emph{computational theory} to \emph{computational model} in order to avoid the overloading of the word \emph{model} already in use for \emph{cost model}.} from which computational 
complexity  stems. Namely, a cost model for a computational theory $T$ is 
\emph{reasonable} if there are mutual simulations of $T$ and TMs (or another 
reasonable theory) working within:
\begin{itemize}
\item for \emph{time}, a polynomial overhead;
\item for \emph{space}, a linear overhead.
\end{itemize}
In many cases, the two bounds hold simultaneously for the same simulation, but 
this is not a strict requirement. The aim is to ensure that the basic hierarchy 
of complexity classes
\begin{equation}
\L \subseteq  \P \subseteq \PSPACE
\subseteq \EXP
\label{eq:hierarchy}
\end{equation}
can be equivalently defined on any reasonable theory, that is, that such classes are \emph{robust}, or theory-independent. Note a slight asymmetry: while for time the complexity of the required overhead (polynomial) coincides with the smallest robust time class ($\P$), for space the smallest robust class is logarithmic ($\L$) and not linear space. In particular, a simulation within a linear space overhead for logarithmic space implies that one needs to preserve logarithmic space.

A typical example of reasonable theory is \emph{random access machines} (RAMs), which are simulated by Turing machines within a \emph{quadratic} time overhead---which justifies the requirement for a polynomial (rather than linear) time overhead---needed for simulating random access on sequential access tapes. The very concept of reasonable cost model was introduced to study the relationship between the space consumption of RAMs and of Turing machines \cite{DBLP:journals/iandc/SlotB88}.

\paragraph{Locked Time and Space} On TMs, space cannot be greater than time, 
because 
using space requires time---we shall then say that space and time are 
\emph{locked}. If both the time and space cost models of a computational theory  
are reasonable, are they also necessarily locked? This seems natural, but it is 
not 
what happens in the $\l$-calculus, at least with respect to its abstract cost 
models. 

\paragraph{(Unreasonable) Abstract Machines} Before diving into the subtleties of cost models for the $\l$-calculus, we clarify concepts and terminology that might be confusing for the unacquainted reader.  The $\l$-calculus is an abstract setting relying on a single rule, $\beta$, which is non-deterministic (but confluent), and that is a non-atomic operation involving three meta-level aspects: 
\begin{enumerate}
\item The search of the $\beta$-redex, and
\item Capture-avoiding substitution, itself based on 
\item On-the-fly $\alpha$-renaming.
\end{enumerate}
 In order to study cost models for the $\l$-calculus, one usually fixes a deterministic evaluation strategy (typically call-by-name or call-by-value) and some micro-step formalism, typically an abstract machine, which simulates $\beta$ and explicitly accounts for the three meta-level aspects. Therefore, abstract machines are intermediate settings used to study simulations of the $\l$-calculus into TMs (or other reasonable theories).

We shall often say that a certain abstract machine is unreasonable for time or space. The use of \emph{unreasonable} in these cases is different than when referred to a computational theory (such as the $\l$-calculus or TMs). It means that the simulation realized by the machine, \emph{and only that simulation}, works in bounds that exceed those for reasonable time or space. In contrast, a theory is unreasonable if \emph{all} possible simulations exceed the bounds.

\paragraph{Why Studying Reasonability for the $\l$-Calculus?} As another clarifying preliminary, let us answer such a question. There are \emph{two} reasons. The \emph{theoretical} and \emph{external} motivation is being able to define the hierarchy \refequ{hierarchy} in the $\l$-calculus, hoping to help resolving the decades-long separation between the $\l$-calculus and mainstream computer science. The \emph{concrete} and \emph{internal} motivation is instead to better understand how to evaluate $\l$-terms, for which there are many different approaches and no general theory.

\paragraph{Cost Models for the $\l$-Calculus} For the $\l$-calculus there are a few natural candidates as cost models. Fix an evaluation strategy $\tox$. Then, we have three candidates for the time cost of a $\tox$ evaluation $\deriv:\tm_0\tox \tm_1 \tox \tm_2 \tox\ldots \tox \tm_n$ of length $n$:
\begin{center}
\begin{tabular}{rl}
\emph{Ink time}: & the time taken by printing out all the terms $\tm_i$ for $i\in\set{0,\ldots,n}$;
\\
\emph{Abstract time}: & the number $n$ of $\tox$ steps; 
\\
\emph{Low-level time}: & the time taken by an abstract machine implementing $\deriv$. 
\end{tabular}
\end{center}
For space, the ink and abstract notions coincide:
\begin{center}
\begin{tabular}{rl}
\emph{Ink/abstract space}: & the size of the largest $\l$-term among the $\tm_i$;
\\
\emph{Low-level space}: & the maximum space used by an abstract machine implementing $\deriv$. 
\end{tabular}
\end{center}
Let us first discuss time. Ink time is locked with ink space and easily shown to be reasonable. The problem with it is twofold: on the one hand, it is too generous a notion, since the cost of functional programs is not usually estimated in this way (in functional programming practice), on the other hand, it is difficult to reason with such a notion, as it is not abstract enough.

Low-level time is a better measure, which is locked with low-level space and easily proved to be reasonable. It differs from ink time in that abstract machines usually rely on some form of sharing to avoid managing the ink representation of all the terms of the evaluation sequence. The obvious drawback is that low-level time depends on the details of the implementation and on which optimizations are enabled. It is thus not abstract, nor fixed once for all, not even when the evaluation strategy is fixed, and not even when the abstract machine is fixed, because the choice of data structures for a concrete implementation usually affects the complexity. It is rather a family of cost models. In particular, it does not have the \emph{distance-from-implementative-details} that is distinctive of the $\l$-calculus. Similar arguments apply to low-level space.

Abstract time is the best notion, since it does not depend on an implementation and it is close to the practice of cost estimates, which does count the number of \emph{function calls}, that, roughly, is the number of $\beta$-steps.  The puzzling 
point is that it is \emph{not locked} with ink space: ink space can be \emph{exponential} in 
abstract time (independently of the strategy), a degeneracy known as \emph{size 
explosion}---we shall say that time and space are 
\emph{explosive}.

Is the $\l$-calculus reasonable? It certainly is, with respect to unsatisfying cost models. The question rather is whether abstract time is a reasonable cost model. This was unclear for a while, because of the intuition that reasonable cost models have to be locked.

The next paragraphs shall discuss the reasonability of abstract time and then explain the delicate aspects of reasonable space, but let us anticipate a point that shall seem to contradict what we just explained. About logarithmic space, there is an unresolvable tension. On the one hand, the principle of \emph{distance-from-implementative-details} would require abstract space---that is, ink space---to be reasonable in the logarithmic case. On the other hand, ink space is log-insensitive and thus cannot be sub-linear. The tension shall be resolved in this paper by abandoning the abstract principle and turning to a low-level notion of logarithmic space. While not ideal, this is the only available solution, at present. A more abstract cost model for logarithmic space would be preferable, in principle, but it cannot be ink space---because of the mentioned tension---and it is far from clear that an alternative abstract notion of space is possible, given the many difficulties discussed in this paper. In fact, before our work, even the existence of a low-level space cost model accounting for logarithmic space was a longstanding open problem. Implementations of the $\l$-calculus, indeed, are tuned for time-efficiency and inevitably use space in an unreasonable way. Therefore, simply adopting low-level space is \emph{not enough}.

\paragraph{Abstract Time is Reasonable} 
In the study of abstract time, what is delicate is the simulation of the 
$\l$-calculus into a reasonable theory, which typically is the one of random access machines 
rather than TMs. The difficulty stems from the explosiveness of abstract time, 
and requires a slight  paradigm shift. To circumvent the exponential 
explosion in space, $\l$-terms are usually evaluated \emph{up to sharing}, that 
is, in abstract machines with sharing that compute shared representations of 
the results. These representations can be exponentially smaller than the 
results themselves: explosiveness is then encapsulated in the sharing unfolding 
process (which itself has to satisfy some reasonable properties, see 
\cite{accattoli_leftmost-outermost_2016,DBLP:conf/ppdp/CondoluciAC19}). The 
number of $\beta$ steps (according to various evaluation strategies) then turns 
out to be a reasonable time cost model (up to sharing), despite  explosiveness. Resorting to sharing amounts to studying abstract time assuming that the underlying notion of space is \emph{low-level space} rather than ink space (which forbids sharing). It is important to point out that the adopted notion of low-level space is \emph{not} proved to be space reasonable. 

The first such result is for weak evaluation by Blelloch and Greiner 
\cite{DBLP:conf/fpca/BlellochG95}, then extended to strong call-by-name evaluation by 
Accattoli and Dal Lago \cite{accattoli_leftmost-outermost_2016}, and very 
recently transferred to strong call-by-value by Accattoli et al. 
\cite{DBLP:conf/lics/AccattoliCC21} and to a variant of strong call-by-value and to strong call-by-need by Biernacka et al. \cite{DBLP:conf/ppdp/BiernackaCD21,DBLP:journals/pacmpl/BiernackaCD22}.

%Abstractly, the point is that the data structure at work in TM, namely the tape, is not linked, and thus uses no sharing. The difficulty in studying reasonable space for the $\l$-calculus is using as little space as TM. Motivated by considerations about time, all machines for the $\l$-calculus rely on some shared data structure, which makes them space unreasonable. While reasonable time requires sharing, reasonable space requires to be able to evaluate without sharing when it is not needed, as it is the case for the encoding of TM.

\paragraph{Ink Space.} For space, the difficult direction is, instead, the 
simulation of TMs in the $\l$-calculus. TMs are space-minimalist, as their only 
data structure, the tape, is a flat data structure that \emph{juxtaposes} cells rather then linking them via pointers---this is one of the key points. Motivated by time-efficiency, all abstract 
machines for the $\l$-calculus rely instead on linked data structures, and---as already pointed out in the introduction---the 
linking pointers add a logarithmic factor to the overhead for the simulation of 
TMs that is space unreasonable. 
Therefore, reasonable space requires to  evaluate without using linked data 
structures when they are not needed, as it is the case for the encoding of TMs.

It is a recent insight by Forster et al. \cite{cbv_reasonable} that evaluating without \emph{any} data structure (via plain rewriting, without sharing) is reasonable for linear ink space even if unreasonable for abstract time (because of explosiveness). An interesting aspect of this result is that it establishes that the rigid (\ie non-postponable) management of garbage collection provided by the $\l$-calculus is enough for reasonable linear space.

\paragraph{Pairing up Abstract Time and Ink Space} Forster et al. 
\cite{cbv_reasonable} also show a surprising fact.  Given two simulations, one 
that is reasonable for ink space but not abstract time, and one that is reasonable for abstract time 
but not ink space, there is a smart way of interleaving them as to obtain 
reasonability for abstract time and ink space \emph{simultaneously}. Their result therefore 
shows that, surprisingly, a computational theory can be reasonable for \emph{unlocked} and explosive notions of
time \emph{and} space. Whenever their interleaving technique applies, however, it also induces a second (low-level) reasonable cost model for space that is locked with the time cost model.

Forster et al.'s is a remarkable contribution to the external interest in reasonable cost model (that is, for defining complexity classes) but their interleaving machine is not a machine that one would use for concrete implementations. That is, the result does not help in the internal (\ie implementation-oriented) understanding of reasonable space.
%, if one allows time to be reasonable up to sharing. %Their result is a 
%%theoretical step further, but concerning how to best evaluate $\l$-terms 
%%their 
%%interleaving machine is of no use.

\paragraph{Work vs Ink Space} A puzzling fact is that sub-linear 
space \emph{cannot} be measured using the ink space cost model, and is then \emph{not} 
covered by Forster et al.'s result. The reason is that if space is the maximum 
size of terms in an evaluation sequence, the first of which contains the input, 
then space simply \emph{cannot be} sub-linear. How could one accommodate for 
\emph{logarithmic} reasonable space? As already explained in the introduction, 
one needs \emph{log-sensitivity}, that is, a distinction between an immutable \emph{input space}, which is not counted for space complexity, and a (smaller) mutable \emph{work space}, that is counted.  It is then natural to switch, again, from ink space to a form of low-level space. Namely, one considers the work space as the low-level space used by an input-preserving abstract machine. 

\paragraph{Low-Level Space}  Turning to low-level space, however, does not immediately provide a solution. As we already mentioned, most abstract machines rely on sharing mechanisms which allow one to prove that abstract time is reasonable while, unfortunately, also make those simulations space unreasonable, as they rely on pointers which add an unreasonable space overhead. It was thus unclear whether low-level space could be reasonable \emph{at all}. The machines for time reasonability realize sharing via environments. The community believed that abstract machines that rely on so-called \emph{tokens} (related to the geometry of interaction), rather than on environments, might provide reasonable notions of low-level space, but it was showed by Accattoli et al. that this is not the case \cite{DBLP:conf/lics/AccattoliLV21}. 

The intuition that one could use environments and yet disable their sharing, as in the \SpKAM, is a contribution of the present work. Our main result is that the low-level space of the \SpKAM---from now on referred  to as \emph{work space}---is indeed reasonable.
%Of course, the machine 
%shall need some data structures, and their design is critical.

\paragraph{\SpKAM and Time} We also study the time behavior of the \SpKAM. Adopting flat environments implies giving up environment sharing, which---we show with an example---makes the \SpKAM unreasonable for abstract time. The situation is then a familiar one: abstract time and work space are explosive. Work space is in this respect a conservative refinement of ink space. %The example also confirms that---as for time---\emph{space reasonable} and \emph{space efficient} are different concepts: the example uses exponentially less space (and it is thus more efficient) if environment pointers (which are unreasonable) are enabled. The key point is that the $\l$-terms of the example are not in the image of the encoding of TMs: for them environments sharing provides an exponential speed-up, while on the image of the encoding it only provides a slow-down.
On the other hand, we prove that the \emph{low-level time} of the \SpKAM is an alternative \emph{reasonable} cost model, obviously locked with work space. 

\paragraph{Is Work Space a Good Cost Model?} It might be argued that work space, being a low-level notion of space, is unsatisfactory. While this is partly true, we believe that it would be an unfair assessment. A first argument against this criticism is that it is unclear what would be the alternative, given that ink space is ruled out by its log-insensitivity. A second argument is that the requirements for reasonable logarithmic space are so strict that they almost \emph{dictate} how the \SpKAM has to be. There does not seem to be much room for alternative designs. That is, it is a cost model given by an abstract machine, but it is a quite special machine: the criticism to low-level cost models, amounting to the fact that different machines would provide different notions of cost, does not seem to apply here. Moreover, we also prove that the same cost model works also for call-by-value.

\paragraph{A Notion of Abstract Space} The space of the \SpKAM is obtained by taking the maximum number of closures during its execution, and weighing every closure with the size of its sub-term pointer (which is not necessarily logarithmic in the size of the input). Taking only the number of closures---ignoring the size of pointers---provides a more abstract notion of \emph{closure space} that is not the actual cost model (if adopted as space cost model, the simulation of the $\l$-calculus in TMs has an additional and unreasonable logarithmic overhead) and yet it provides a useful abstraction. For instance, this closure space is stable by $\eta$-equivalence, as we show. Additionally, in a companion paper we show how to measure closure space via multi types \cite{ICFP2022}, thus abstracting away the low-level details of the machine.

\paragraph{What Does Our Result Say for Concrete Implementations?} The \SpKAM is a realistic machine as long as logarithmic space is the main concern. The example showing that abstract time and work space are explosive also confirms that---as for time---\emph{space reasonable} and \emph{space efficient} are in some sense different concepts: the example uses exponentially less space (and it is thus more efficient) if environment pointers (which are space unreasonable) are enabled. The key point is that the $\l$-terms of the example are not in the image of the encoding of TMs: for them, environments sharing provides an exponential speed-up (for both time and space), while on the image of the encoding it only provides a slow-down. In other words, the \SpKAM is efficient for first-order, that is, TMs-like logarithmic space computations, but beyond them it can be \emph{desperately} space inefficient.

%\paragraph{Functional Languages are Almost Space Reasonable} Our results can be read in a concrete perspective. Implementations of functional languages often rely on local environments with both garbage collection and environments pointers. According to our study, they are then space \emph{unreasonable}, but only for a logarithmic factor. Moreover, striving for removing that factor makes the execution space reasonable but time unreasonable and potentially space inefficient.
%\paragraph{Summing up} Sub-linear space requires log-sensitivity, which is modeled via an abstract machine. Reasonability further requires a new encoding of TM and tuning the abstract machine. For the machine, a first important ingredient is garbage collection. A second, perhaps surprising, new ingredient is adopting data structures that are not structured via pointers. More generally, studying logarithmic space forces a finer approach to abstract machines where pointers are made explicit. This is not really surprising, as logarithmic space is essentially the study of pointers. Finally, the obtained notion of logarithmic space  is shown to have the same explosive relationship with abstract time than abstract space. 
% !TeX spellcheck = en_US
% !TEX root = main.tex
%%%%%%%%%%%%%%%%%%%%%%%%%%%%%%%%%%%%%%%%%%%%%%%%%%%%%%%%%%%%%%%%%%%%%
\section{Bird's Eye View of the Problems and Their Solutions}
\label{sect:critical-points}
There shall be six critical points which---only when are all solved simultaneously---shall allow us to achieve the main result of this paper, namely the fact that the space of the \SpKAM is a reasonable space cost model for the $\l$-calculus accounting for logarithmic space. None of them is particularly difficult to deal with, but each one of them is critical, and solving all of them at once does make the solution somewhat involved. In particular, a reader can easily get lost and lose sight of where the essence of the result is. Here we give an overview of these critical points and of how we address them. One of our main contributions is the careful identification of these points.

Let's start by listing some key points of the simulation of TMs by an abstract machine.
\begin{itemize}
\item \emph{Pointers}: TMs do not use pointers, while in the $\l$-calculus and its execution via abstract machines, three kinds of pointers play a role. Some are essential and some are dangerous for logarithmic reasonable space.
\begin{itemize}
\item \emph{Variable pointers}: variables are pointers to their binders. Such pointers are unavoidable. Different term representations induce pointers with different properties.
\item \emph{Sub-term pointers}: abstract machines often manipulate pointers to sub-terms rather than the sub-terms themselves. This is essential for logarithmic space.
\item \emph{Data pointers}: data structures are often implemented as pointer-based structures, such as pointer-based linked lists. These pointers are also often responsible for forms of data-structures sharing, which is essential in time-reasonable computations. They are dangerous for reasonable space.
\end{itemize}
\item \emph{Tapes}: the input tape of the TM is meant to be represented as part of the initial code fed to the abstract machine, the rest of which is dedicated to represent the transition function of the TM. The work tape is instead represented by the data structures of the abstract machine, namely the applicative stack and the environments. 

\item \emph{Overhead}: there are three points in which the encoding and the simulation incur some overhead.
\begin{enumerate}
\item \emph{Input representation}: this is due to the representation of input strings $\str$ as $\l$-terms. One expects such an overhead to be $\bigo{\size\str}$, that is, linear in the length $\size\str$ of the input string. It depends on the fixed encoding of strings \emph{and} on the adopted notion of variable pointer, that is, on the fixed representation of $\l$-terms. 

Technically speaking, the linear overhead here is not mandatory: a reasonable simulation is possible also if the input representation overhead is not linear (it can be polynomial, usually $\bigo{\size\str\log\size\str}$, but not exponential), as long as sub-term pointers and the overhead for scrolling the input tape (discussed below) are both logarithmic in the size of $\str$ (rather than in the size of the encoding of $\str$ as a $\l$-term). We shall however show how to obtain a linear input representation overhead, as it seems natural and it reinforces the trust in the correctness of the result.

\item \emph{(Work) tape representation}: tapes juxtapose cells using zero space for such a juxtaposition. For the input tape, in fact it is not mandatory that the representation overhead is linear, as discussed in the previous point. The representation of the work tape, instead, has to be linear. Since it is given by the data structures of the abstract machine, it usually rests on data pointers, which add an unreasonable logarithmic overhead. It is important to point out that here the additional logarithmic factor is not in the size of the input: data pointers are generated along the execution and thus they depend on the amount of space currently in use. Thus data pointers add a bureaucratic space factor that is logarithmic on the space used by TMs.
\item \emph{Tape scrolling}: this is the space overhead needed to simulate in the $\l$-calculus the moving of the TM head over the tape, with respect to the size of the tape. It is a notion that is tricky to define precisely because, in the reference encoding of TMs, a single move over the tape incurs only a constant space overhead. It is better intended as the space overhead generated by scrolling the \emph{whole} tape from, say, left to right. Most TM executions never do such a mono-directional scrolling, but they nonetheless incur such an overhead during their continuous moving over the tape.

For our result, we shall need a space overhead for scrolling the input tape that is logarithmic in its size, while for the work tape a linear space overhead (in its size) is enough. Clearly, the requirement for the input tape is specific to the study of logarithmic space. We shall see, however, that (independently of logarithmic space) the adoption of a low-level space cost model (instead of ink space) forces us to be very careful with respect to the linear space scrolling overhead for the work tape, since the use of machine pointers can easily add an unreasonable logarithmic factor.

Here various ingredients play a role: the way in which tapes are encoded as $\l$-terms, the space optimizations of the abstract machine as well as the way in which sub-term pointers are organized.

\end{enumerate}
\end{itemize}
We now list the ingredients that allow us to meet the requirements for the three overheads.
\begin{enumerate}
\item \emph{Input representation}. Here the optional linear space overhead is achieved by:
\begin{enumerate}
\item Using the Scott encoding of strings, which is also used in the standard enconding of TMs; 
\item Adopting a representation of $\l$-terms for which the Scott encoding uses variable pointers of size proportional to the size of the \emph{alphabet} $\alpone$ of the input string $\str$, rather than to the size of $\str$ itself (which would occur if one would represent variables via textual names and enforce Barendregt's convention, or if one would represent $\l$-terms as proof nets or string diagrams). The actual representation is left unspecified, but for instance de Bruijn indices would do. This is discussed at the beginning of \refsect{moving}.
\end{enumerate}
If the input representation has a more than linear overhead (typically $\bigo{\size\str\log\size\str}$ for input string $\str$), the logarithmic size of sub-terms pointers is obtained by adopting a specific address scheme for the \SpKAM, what we call \emph{left addresses}.  This is discussed in \refsect{lambda-addresses}.
\item \emph{Work tape representation}: in this case the linear space overhead is achieved by totally disabling the use of data pointers in the abstract machine. The stack and the environments are then represented via contiguous cells, as strings, without linking the different cells via data pointers. This is discussed in \refsect{naive-kam} and it is unusual for abstract machines. It has the consequence of removing the sharing of data structures needed for time reasonability. As already mentioned, our \SpKAM  shall be time unreasonable, with respect to the abstract cost model for time.

\item \emph{Tape scrolling}. For scrolling tapes with the required overhead, both the encoding and the abstract machine have to be modified.
\begin{itemize}
\item \emph{Encoding}. The standard encoding of (single tape) TMs is based on a representation of tapes, which comes with costant-time read operations but linear space scrolling overhead. This would simply forbid our result, as the input tape scrolling overhead has to be logarithmic. Therefore, we modify the encoding and adopt a different representation for the input tape---dubbed \emph{mathematical representation} following van Emde Boas \cite{DBLP:conf/sofsem/Boas12}---coming with polynomial time read operations (thus worse for time) and logarithmic space scrolling overhead (but better for space). For the work tape, we keep the standard representation because the work tape requires only a linear space scrolling overhead, and for write operations the standard representation is better suited. This aspect is discussed in \refsect{moving}.
\item \emph{Abstract machine}. There are two critical points here.
\begin{enumerate}
\item \emph{Sub-term pointer addressing}. We mentioned that the standard representation of tapes comes with linear space tape overhead. This is indeed true, but it requires a low-level analysis of the size of sub-term pointers, showing that some of them are of constant rather than logarithmic size, and it holds only if sub-term pointers satisfy some properties. Such properties are satisfied by the \emph{left addressing scheme} that we adopt for the \SpKAM. This aspect is discussed in \refsect{moving}.
\item \emph{Space awareness}. Environment-based abstract machines are usually developed having time efficiency in mind and neglecting space. For achieving the required tape scrolling overheads, they have to be tuned so as to be parsimonious with respect to space. There are two optimizations that have to be realized, and that are formally defined in \refsect{spkam}:
\begin{enumerate}
\item \emph{Eager garbage collection}. In absence of garbage collection, abstract machines have an inflationary, ever increasing use of space. In particular, they allocate space at every $\beta$-step, which is the \emph{tick} of their abstract notion of time. Therefore, their use of space is entangled with their use of time, which is unreasonable for space. Garbage collection is needed in order to disentangle space from time. Eager collection is needed to maximize such disentanglement. In addition, garbage collection has to be implemented in a naive and time-inefficient way, because smart techniques such as reference counters would incur a space overhead---for counters---that would be unreasonable.
\item \emph{Unchaining}. Disentangling space from time by freeing garbage is not enough. One also needs to ensure that non-garbage space is not redundant, by avoiding silly indirections, sometimes referred to as \emph{space leaks}. This is achieved by the \emph{environment unchaining} optimization, which at the moment of creating a closure checks whether it is simply referring to another closure (which happens when the term part of the closure is a variable), in which case it is short-cut.
\end{enumerate}
\end{enumerate}
 \end{itemize}
\end{enumerate}

Summing up, there are three main issues, the last of which actually composed by four sub-problems, for a total of six critical points. We tame them as follows:
\begin{enumerate}
\item \emph{Input representation overhead}: handled by the Scott encoding and the representation of $\l$-terms;
\item \emph{Work tape overhead}: handled by the disabling of data pointers;
\item \emph{Tape scrolling overhead}, composed by the following sub-points:
\begin{enumerate}
\item \emph{Tape representation overhead}: handled by the mathematical representation of the input tape;
\item \emph{Sub-term pointer addressing}: handled by the left addressing of sub-term pointers;
\item \emph{Space awareness}, composed by the following sub-points: 
\begin{enumerate}
\item \emph{Space entangled with time / re-usability of space}: handled by eager garbage collection;
\item \emph{Space leaks / indirections}: handled by environment unchaining.
\end{enumerate}
\end{enumerate}
\end{enumerate}
% !TeX spellcheck = en_US
% !TEX root = main.tex
%%%%%%%%%%%%%%

\section{The \texorpdfstring{$\l$}{λ}-Calculus, Term Sizes, and Addresses}
\label{sect:lambda-addresses}
In this section, we define the $\l$-calculus and discuss the delicate aspects of how to measure to size of terms and notions of addresses for constructors in terms.

\paragraph{$\l$-Calculus.} Let $\mathcal{V}$ be a countable set 
	of variables. 
	Terms of the \emph{$\lambda$-calculus} $\Lambda$ are given by:
	\[\begin{array}{rrcl}
	\textsc{$\l$-terms} & \tm,\tmtwo,\tmthree & \grameq & x\in\mathcal{V}\midd 
	\lambda x.\tm\midd 
	\tm\tmtwo.
	\end{array}\]
	\emph{Free} and \emph{bound variables} are defined as 
	usual: $\la\var\tm$ binds $\var$ in $\tm$. 
	%\emph{Closed} terms are terms without free 	variables. 
	Terms are considered modulo $\alpha$-equivalence. Capture-avoiding 
	(meta-level) substitution of 	all the free occurrences of $\var$ for 
	$\tmtwo$ in $\tm$ is noted $\tm\isub\var\tmtwo$. The computational rule is 
	\emph{$\beta$-reduction}:
\begin{center}
$\begin{array}{cccc}
(\la\var\tm)\tmtwo &\tob& 
\tm\isub\var\tmtwo\end{array}$\end{center}
which can be applied anywhere in a $\l$-term. Here, a strategy $\to$ shall be a sub-relation of $\tob$. Given a relation $\to$, its reflexive-transitive closure is noted $\to^{*}$, and a $\l$-term $\tm$ is $\to$-normal if there are no $\tmtwo$ such that $\tm \to \tmtwo$. A $\to$-sequence is a pair of $\to^{*}$-related terms, often noted $\deriv:\tm \to^{*}\tmtwo$, and it is \emph{complete} if $\tmtwo$ is $\to$-normal.

	\paragraph{The Constructor and Code Sizes of $\l$-Terms} The \emph{(constructor) size} of a $\l$-term  
is defined as follows:
\[\begin{array}{c@{\hspace{.17cm}}c@{\hspace{.17cm}}c@{\hspace{.6cm}}c@{\hspace{.17cm}}c@{\hspace{.17cm}}c@{\hspace{.6cm}}c@{\hspace{.17cm}}c@{\hspace{.17cm}}cccc}
\size\var & \defeq & 1
&
\size{\tm\tmtwo} & \defeq & \size\tm + \size\tmtwo +1
&
\size{\la\var\tm} & \defeq & \size\tm+1
\end{array}\]
The \emph{code size} $\codesize\tm$ of a $\l$-term $\tm$ is instead bounded by $\bigo{ \size\tm\log\size\tm}$. The idea is that, when terms are explicitly 
represented, variables are some abstract kind of pointer (de Bruijn 
indices/levels, names, or actual pointers to the syntax tree), of size  
logarithmic in the number $\size\tm$ of constructors  of $\tm$. Then a 
term with $n$ constructors requires space $\bigo{n\log n}$ to be represented. For 
our study, it is important to stress the difference between the constructor size $\size\tm$ and the code size 
$\codesize\tm$, because given a binary input string $\istr$, at first sight its encoding 
$\tm_{\istr}$ as a $\l$-term satisfies $\size{\tm_{\istr}} = 
\Theta(\size\istr)$ and $\codesize{\tm_{\istr}} = 
\bigo{\size\istr\log\size\istr}$, and so 
$\codesize{\tm_{\istr}}$ has an additional (unreasonable) logarithmic factor. We refer to this issue as \emph{the code-constructor gap}.

\paragraph{Notions of Space and Size} The relevance of the code-constructor gap depends on the notion of space under consideration. If one adopts ink space (that is, the max size of $\l$-terms) as cost model, and thus is concerned with at least \emph{linear} space, the gap is relevant, because the size of terms can indeed have an additional unreasonable logarithmic factor.

In \refsect{moving}, we shall encode strings in the $\l$-calculus using the Scott encoding, which has the property that, with respect to some concrete representations of terms (such as de Bruijn indices), variable pointers have \emph{constant size}, so that $\codesize{\tm_{\istr}} =\size{\tm_{\istr}} = 
\Theta(\size\istr)$, thus removing the code-constructor gap.

If instead one studies logarithmic space, and thus adopts a low-level cost model where the input is separated from the work space, the gap is---perhaps surprisingly---less relevant. It is indeed perfectly possible that the encoding of the input string $\istr$ has size $\bigo{\size\istr\log\size\istr}$, because the size of the input is not counted for the space usage, and still the space cost is logarithmic in $\size\istr$. What is important, indeed, is that pointers to the input---that are inevitably used if the space complexity is logarithmic---have size $\bigo{\log\size\istr}$ (rather than $\bigo{\log(\size\istr\log\size\istr)}$), which is easy to obtain, as we explain in the next paragraph. 

Since log-sensitivity allows us to use term representations also when the encoding of strings suffers from the code-constructor gap, we do not fix a concrete representation of terms. Such a relaxed approach is harmless, because, as we mentioned, there is at least one term representation (namely, de Bruijn indices) for which strings are encoded with no gap. It is however convenient because it allows us to use names rather than indices, improving readability, while obtaining more general results at the same time.

\paragraph{Addresses in Terms} We shall need addresses in terms for two purposes: comparing constructors (the reason is explained in the next section) and sub-term pointers. We then adopt two address schemes, as the two purposes need different properties. For constructor comparisons we adopt \emph{tree addresses}, while for sub-term pointers we rely on \emph{left addresses}, as we now explain.

It is standard that a constructor in a term can be identified via a path in the syntax tree of the term, and that such a path can be described as a binary address. We shall use this notion of tree address for comparing constructors in terms. For that, we also need to fix a way of comparing variables. We choose it to be de Bruijn indices. Note that this choice does not necessarily force the term representation itself to be de Bruijn indices, because given another representation one can usually compute the de Bruijn index of a variable in logarithmic space.

\begin{defi}[Tree addresses]
\label{def:tree-address}
A \emph{tree address} $\addr$ is a binary string. The constructor 
	$\tm|_{\addr}$ of a closed term $\tm$ at $\addr$ is given by the following partial function defined by structural induction:
	\[\begin{array}{rcl@{\hspace{1.2cm}}rlc}
	\staddr\var\ems & \defeq&  \db(\var) & 
		\staddr\var{b\cons\addr} & \defeq&  \bot \\ 

	\staddr{(\tm\tmtwo)}{\ems} & \defeq & @ &
	\staddr{(\tm\tmtwo)}{0\cons\addr} & \defeq &\staddr\tm\addr\\[3pt]
	\staddr{(\la\var\tm)}{\ems} & \defeq & \lambda & 
		\staddr{(\tm\tmtwo)}{1\cons\addr} & \defeq &\staddr\tmtwo\addr\\[3pt]
&&&	\staddr{(\la\var\tm)}{b\cons\addr} & \defeq &\staddr\tm\addr 
	\end{array}\]
where $b\in\set{0,1}$, $\db(\var)$ is the de Bruijn index corresponding to $\var$ (written in binary) if the whole term 
$\tm$ were written using de Bruijn indexes, and $\bot$ denotes that $\tm|_{\addr}$ is undefined.
\end{defi}
Tree addresses are a convenient way to point at a constructor or a sub-term. Unfortunately,  they are not in general logarithmic in the size of the term. It is enough to consider terms the tree structure of which is linear (such as $n$ applications of the identity to itself), so that the structural address of some constructors is linear in the size of the term.

We then consider also a second notion, \emph{left addresses}.

\begin{defi}[Left address]
\label{def:left-address}
Let $\tm$ be a $\l$-term and $c$ be a constructor of $\tm$ (identified by a context or by a tree address). The left address of $c$ in $\tm$ is simply the index (written in binary) of the constructor in the enumeration of the nodes generated by an in-order visit of its syntactic tree, that is, a visit that: 
\begin{enumerate}
\item On applications $\tmthree\tmfour$, it first enumerates the constructors in the left sub-tree $\tmthree$, then the application $@$, and last the constructors in the right sub-tree $\tmfour$, recursively;
\item On abstractions $\la\var\tmthree$, it first enumerates the abstraction $\lambda\var$ and then the constructors in $\tmthree$.
\end{enumerate}
\end{defi}

For instance, the constructors of $\var((\la\vartwo\varthree)\varfour)$ are enumerated in the following order $\var$, $@$, $\lambda\vartwo$, $\varthree$, $@$, and $\varfour$.

Left addresses are always logarithmic in the constructor size of a term (even when the code size is bigger than the constructor size), and they are nothing else than the left-to-right position of the constructor in the string representation of the $\l$-term, whence their name. 

Another property of left addresses that shall be crucial in our complexity analysis is that for an application $\tmtwo\tmthree$, the size of the addresses in $\tmtwo$ is independent from $\tmthree$. Note that this would hold also if one enumerates the constructors according to a visit in pre-order, but not in post-order. 

We avoid left addresses for comparisons because we shall have to compare constructors extracted from machine states, and the extraction from states is considerably simpler if done with respect to tree addresses.
%	The constructor $\caddr\tm\addr$ of $\tm$ at address $\addr$ is the topmost 
%	constructor of $\staddr\tm\addr$, if $\staddr\tm\addr$ is an application, an 
%	abstraction, or the occurrence of a variable $\var$ that is free in $\tm$, 
%	otherwise, \ie if it is the occurrence of a  variable $\vartwo$ bound in 
%	$\tm$, it is the address $\addrtwo$ of the associated binder $\l\vartwo$ in 
%	$\tm$ (which is a prefix of $\addr$).

\section{Reasonable Preliminaries}
\label{sect:prelim}

%\begin{defi}[Reasonable Models]
%	We say that a computational model $C_1$, with output set $O_1$ is 
%\emph{reasonable} if the following points are satisfied.
%	\begin{itemize}
%		\item There exists a reasonable model $C_2$, with a set of outputs $O_2$.
%		\item There exist two bijections $dec_{12}:O_1\to O_2$ and 
%$dec_{21}:O_2\to O_1$ such that $dec_{12}\circ dec_{21}=dec_{21}\circ 
%dec_{12}=id$.
%		\item There exist an encoding $enc_{12}:C_1\to C_2$ such that if $c_1$ 
%terminates with output $o_1\in O_1$ in time $T_1(c_1)$, respectively space 
%$S_1(c_1)$, then $c_2\defeq enc_{12}(c_1)$ terminates with output $o_2\in O_2$ 
%such that $o_2=dec_{12}(o_1)$ in time $T_2(c_2)=poly(T_1(c_1))$, respectively 
%space $S_2(c_2)=kS_1(c_1)$.
%		\item The same but reversed.
%		\item The encoding and decoding functions have to satisfy the complexity 
%requirements... measured in $C_1$.
%	\end{itemize}
%\end{defi}

In the study of reasonable cost models for the $\l$-calculus, it is customary 
to show that the $\l$-calculus simulates TMs reasonably, and 
conversely that the $\l$-calculus can be simulated reasonably by TMs\footnote{In the study of reasonable time, random access machines (RAMs) rather than TMs are usually the target of the encoding of the $\l$-calculus, because RAMs are reasonable and easier to manage for the time analysis of algorithms. For our study, it shall be simpler to use TMs. Therefore, we avoid references to RAMs in the paper.} \emph{up to sharing}. 
Since 
space is more delicate than time, we fix the involved theories and their cost 
measures carefully. 

\paragraph{Turing Machines.} We adopt TMs working on the boolean 
alphabet $\Bool \defeq \set{0,1}$. For a study of logarithmic space 
complexity, one has to distinguish  \emph{input} space and  
\emph{work} space, and to \emph{not} count the input space for space 
complexity. On TMs, this amounts to having \emph{two} tapes, a read-only 
\emph{input tape} on the alphabet  $\Boolin \defeq \set{0,1,\stsym,\ensym}$, 
where $\stsym$ and $\ensym$ are delimiter symbols for the start and the end of 
the input binary 
string, and an ordinary read-and-write \emph{work 
tape} on the boolean alphabet extended with a blank symbol $\Boolb \defeq 
\set{0,1,\elemblank}$.  To keep things simple, we use TMs without any output 
tape. 
The machine rather has two final states $\state_{0}$ and $\state_{1}$ encoding a 
boolean output---there are no difficulties in extending our results to TMs with 
an 
output tape. 

Let us call these machines \emph{log-sensitive TMs}.
A log-sensitive TM $M$ computes the function 
$f:\Bool^*\rightarrow\Bool$ by a sequence of transitions $\run:\config_{\tt in}^M(\inputstr)\to^n 
\config_{\tt fin}^M({f(\istr)})$ where $\istr\in\Bool^*$, $\config_{\tt in}^M(\inputstr)$ is 
an initial configuration of $M$ with input $\istr$ and $\config_{\tt fin}^M({f(\istr)})$ 
is a final configuration of $M$ on the final state $\state_{f(\istr)}$. We define 
the time of the sequence $\run$ as $T_{\textrm{TM}}(\run)\defeq n$ and the space 
$S_{\textrm{TM}}(\run)$ 
of $\run$ as the maximum number of cells of the work tape used by 
$\run$.

While we shall study in detail the 
encoding of TMs in the $\l$-calculus, which is the difficult direction, we are not going to lay out 
the details of the simulation of the $\l$-calculus on TMs, as it is conceptually simpler. We shall provide an 
abstract machine for the $\l$-calculus and study its complexity using standard 
considerations for algorithmic analysis, but 
without 
giving the details of the simulation. 

%\paragraph{Random Access Machines.} The RAM model we target has a read-only 
%input register the space of which is not counted for space complexity, similarly 
%to 
%TMs. 
%For the sake of completeness, we clarify the RAM cost models of reference: the 
%logarithmic measure for time and Slot 
%and Van Emde Boas's $\mathtt{size}_{b}$ measure for space 
%\cite{DBLP:journals/iandc/SlotB88}, counting 0 for 
%unused registers and taking into account the logarithm of both the content and 
%the 
%index of used registers.
%%\ben{---these subtleties shall in fact be irrelevant, as the simulation of the 
%%abstract machine on RAM shall use only a constant number of registers}
%%
%%\red{As we have done with TM, we consider RAM without any output tape, 
%%	the machine rather has two final states $\state_{0}$ and $\state_{1}$ encoding 
%%	a boolean output.} \blue{Given a RAM $R$, we use $T_{R}(\size\istr)$ and 
%%	$S_{R}(\size\istr)$ for 
%%the time 
%%and space used by $R$ to reach a final configuration starting with input 
%%$\istr\in\Bool^*$.}
%Given a RAM $R$, we use $T_{\textrm{RAM}}(\run)$ and $S_{\textrm{RAM}}(\run)$ 
%for the time and 
%space used by $R$ to reach a final configuration with a sequence of transitions 
%$\run$.

\paragraph{Reasonable Cost Models.} We give the simulations and the bounds 
we shall consider for reasonable cost models for the $\l$-calculus.
 The case of the $\l$-calculus is
peculiar  because its simulation on 
TMs is up to sharing, that is, it computes a \emph{compact representation} of 
the result, not the encoding of the result itself. Therefore, some further 
conditions about such a representation are required, expressed via a decoding function. The one about 
space is new and motivated by the paragraphs after the definition.

\begin{defi}[Reasonable cost model for the $\l$-calculus]
\label{def:reasonable}
A \emph{reasonable time (resp. space) cost model} for the $\l$-calculus is an evaluation strategy $\to$ together with a function $T_{\l}$ (resp. $S_{\l}$) from complete $\to$-sequences $\deriv: \tm \to^{*} \tmtwo$ to $\nat$ such that:
\begin{itemize} 
\item \emph{TMs to $\l$}: there is an encoding $\enc\cdot$ of binary strings and TMs 
into the $\l$-calculus such that if the run $\runtwo$ of a TM $M$ on input 
$\istr$ 
ends on a state $\state_{\bool}$ with $\bool\in\Bool$, then there is a complete 
sequence 
$\deriv: \enc M\, \enc{\istr} \to^{*} \enc\bool$  such that $T_{\l}(\deriv) = 
\bigo{poly(T_{\textrm{TM}}(\runtwo), \size\istr)}$ (resp. space 
$S_{\l}(\deriv) = \bigo{S_{\textrm{TM}}(\runtwo) + \log\size\istr}$). Moreover:
\begin{itemize}
\item \emph{Complexity of the encoding:} computing $\enc\istr$ is done in time $\bigo{poly(\size{\istr})}$, and space $\bigo{\log(\size{\istr})}$ (measured on a reasonable model).
\end{itemize}

\item \emph{$\l$ to TMs}: there is an encoding $\underline\cdot$ of $\l$-terms into binary strings, 
a TM $M$, and a decoding $\unf\cdot$ of final configurations for $M$ such that 
if $\deriv: \tm \to^{*} \tmtwo$ is a complete sequence, then the execution 
$\runtwo$ of 
$M$ on input $\underline\tm$ produces a final configuration $C$ such that $\unf C=\tmtwo$ in time 
$T_{\textrm{TM}}(\runtwo) = \bigo{poly(T_{\l}(\deriv), \size\tm)}$ (resp. 
space 
$S_{\textrm{TM}}(\runtwo) = \bigo{S_{\l}(\deriv) + \log\size\tm}$). Moreover:
\begin{itemize}
	\item \emph{Complexity of the encoding:} computing $\underline\tm$ is done in time $\bigo{poly(\size{\tm})}$, and space $\bigo{\log(\size{\tm})}$ (measured on a reasonable model).
	\item \emph{Polytime result equality}: for all final configurations 
	$C'$ of $M$, 
	testing whether $\unf C = \unf {C'}$ can be done in 
	time $\bigo{poly(\size C,\size{C'})}$ (measured on a reasonable model).
	\item \emph{Logarithmic space constructor equality}: given the initial term $\tm$, the final configuration $C$, and a tree address 
	$\addr$, computing $(\unf C)|_\addr$ has space complexity
	$\bigo{\log\size\tm + \log\size C+ \log{\size\addr}}$ (measured on a reasonable model).
\end{itemize} 
\end{itemize}
\end{defi}

\paragraph{Explaining the Complexity of the Encoding Conditions} The two requirements on the complexity of the encodings ensure that no unreasonable overhead is hidden inside the enconding functions. In particular, in going from $\l$ to TMs, we do not want the encoding function to ``execute'' the encoded $\lambda$-term.

Note however that these conditions are somewhat vague, and necessarily so: they ask complexity bounds for the encoding of a computational theory into another, but in which theory are those bounds to be taken? To be precise, one should be able to express both theories inside a third reasonable theory... but how has this third theory been shown reasonable? There is no real way out, because the study of reasonable cost models is \emph{preliminary} to the definition of notions of complexity valid \emph{across theories}. We added the conditions anyway to intuitively point out that encodings doing more than just translating from one formalism to the other should be ruled out.

\paragraph{Explaining the Equality Conditions} Polytime result equality ensures that 
compact results (\ie final configurations) can be compared for equality of the 
underlying unshared result without having to unshare them, which might take 
exponential time. The requirement is with respect to another compact result 
because polynomiality in $\size\state$ and $\size\tmtwo$ is useless if 
$\size\tmtwo$ is exponential in $\size\state$. For sharing as explicit 
substitutions/environments, polytime result equality was first proved by Accattoli 
and Dal Lago \cite{DBLP:conf/rta/AccattoliL12}, and then showed linear (on random access machines with constant-time pointers manipulation) by 
Condoluci et al. \cite{DBLP:conf/ppdp/CondoluciAC19}.
The \emph{logarithmic space constructor equality} requirement is new, and ensures that the compact representation allows one to 
access atomic parts of the result as if the result were unshared. To motivate 
it, consider a result $\tmtwo$  that is 
exponentially bigger than its compact representation as a final configuration $C$, which is what 
happens with size exploding families. The requirement ensures that to read 
atomic parts of $\tmtwo$ out of $C$ there is no need to unfold the sharing 
in $C$, which might require space exponential in $\size C$. Moreover, the requirement essentially states that, to compute $(\unf C)|_\addr$, one has to use only a constant amount of pointers to $\tm$, $C$, and $\addr$. In particular, the requirement is independent of the cost of evaluation. This aspect rules out degenerated simulations 
where a part of the work is hidden in the representation of the result (think of 
the simulation that does nothing and leaves all the work to the 
de-compactification of the result). In our case, the proof that constructor equality can be tested in logarithmic space shall be straightforward, which contrasts with polytime result equality, that requires non-trivial algorithms.

\paragraph{Single Inputs, not Input Lengths} Note that our cost assignments concern \emph{runs}, thus a single input of a given length, rather than the max over all inputs of the same length, as it is usually done in complexity. The study of cost models is somewhat finer, the max can be considered afterwards.

\section{\texorpdfstring{$\l$}{λ}-Calculus and Abstract Machines} 
A term is \emph{closed} when 
	there are no free occurrences of variables in it. 	The operational semantics---that is, the evaluation strategy---that we adopt in most of the paper is \emph{weak head evaluation} $\towh$, defined as follows:
	\[
(\la\var \tm) \tmtwo \tmthree_1 \ldots 
\tmthree_h \ \ \towh \ \ \tm \isub\var \tmtwo 
\tmthree_1 \ldots \tmthree_h.
\]
We further restrict the setting by considering only closed terms, and refer to 
our framework as \emph{\ccallbn} (shortened to \ccbn). Basic well known facts 
are that in \ccbn normal forms are precisely abstractions and that 
$\towh$ is 
deterministic.

\paragraph{Abstract Machines Glossary.}  In this paper, an \emph{abstract 
machine} $\mach = (\States, \tomach, \compil\cdot, \unf\cdot)$ is a transition 
system $\tomach$ over a 
set of states, noted $\States$, together with two functions:
\begin{itemize}
\item \emph{Initialization.} $\compil\cdot:\Lambda\to\States$,  turning $\l$-terms into 
states;
\item \emph{Decoding.} $\unf\cdot:\States\to\Lambda$, turning states into 
$\l$-terms and such that $\unf{\compil\tm}=\tm$ for every $\l$-term. 
\end{itemize}

A state $\state\in\States$ is composed by the \emph{(immutable) code} $\tm_{0}$, the \emph{active term} $\tm$, and some data structures. Since the code never changes, it is usually omitted from the state itself, focussing on \emph{dynamic states} that do not mention the code.
A state $\state$ is \emph{initial} for $\tm$ if $\compil\tm = \state$. In this paper, $\compil\tm$ is always defined as the state having $\tm$ as both the code and the active term and having all the 
data structures empty. Additionally, the code $\tm$ shall always be \emph{closed}, without further mention. A state is \emph{final} if no transitions apply.
 A \emph{run} $\run: \state \tomach^*\statetwo$ from state $\state$ to state $\statetwo$ is a possibly empty sequence of transitions, the length of which is noted 
$\size\run$. If $a$ and $b$ are transitions labels (that is, $\tomachhole{a}\subseteq \tomach$ and 
$\tomachhole{b}\subseteq \tomach$) then $\tomachhole{a,b} \defeq \tomachhole{a}\cup \tomachhole{b}$, $\sizeparam\run a$ 
is the number of $a$ transitions in $\run$, and $\sizeparam\run{\neg a}$ is the size of transitions in $\run$ that are 
not $\tomachhole{a}$. An \emph{initial run} is a run from an initial state $\compil\tm$, and it is also called \emph{a run from $\tm$}. A state $\state$ is \emph{reachable} if it 
is the target state of an initial run. A \emph{complete run} is an initial run ending on a final state.

\paragraph{Abstract Machines and Abstract Implementations} Abstract machines do 
not specify how the (abstract) data structures of the machine are meant to be 
realized. In general an abstract machine can be implemented in various ways, 
inducing different, possibly incomparable performances. Therefore, it is not 
really possible to study the complexity of the machine without some assumptions 
about the implementation of its data structures. The study of reasonable 
space requires to take into account the use, and especially the \emph{size}, of 
pointers, which is instead usually omitted in the coarser study of reasonable 
time. In that context, indeed, pointers are assumed to be manipulable in constant 
time, which is safe because the omitted logarithmic factors are irrelevant for the 
required polynomial overhead. The more constrained study of space instead requires 
to clarify the cost of pointers.
 Switching to such a level of 
detail, apparently innocent gaps between the specification of a machine 
and how it is going to be implemented suddenly become relevant.

To account for these subtleties, we specify for every construct of the abstract 
machine the space that it requires, and for every transition the time that it 
takes, both asymptotically. The adoption of such an \emph{abstract 
implementations} is in our opinion  a contribution of this paper towards a more 
solid theory of abstract machines.

\begin{defi}[Abstract implementation]
Let $\mach$ be an abstract machine and $\run:\compil{\tm_{0}} \tomach^{*} 
\state$ an initial run for $\mach$. An abstract implementation $I$ for $\mach$ 
is the assignment of asymptotic space costs $\sizespace\cdot^{I}$ for every component of $\state$ 
and of asymptotic time costs $\sizetime\cdot^{I}$ for every transition from $\state$.
\end{defi}

Assigning costs to the state components provides the space cost $\sizespace\state^{I}$ of each state $\state$, by summing over all components. 

\begin{defi}[Space and time of runs]
\label{def:space-time-runs}
	Let $\run:\state_{0} \to^k\state_{k}$ be an initial run of an abstract machine $\mach$ and $I$ an abstract implementation for $\mach$. 
	\begin{enumerate}
	\item The $I$-space cost of $\run$ is 
	$\sizespace{\run}^{I}\defeq\max_{\state\in\run}\sizespace{\state}^{I}$.
	\item The $I$-time cost of $\run$ is  $\sizetime{\run}^{I}\defeq\sum_{i=0}^{k-1}\sizetime{\state_{i} \to \state_{i+1}}^{I}$.
	\end{enumerate}
\end{defi}

\paragraph{A Technical Remark} Note that abstract implementations do not specify the space cost of transitions $\state \tomach \statetwo$. According to the space cost for a run, such a cost has to be the difference $\sizespace{\statetwo}^{I} - \sizespace{\state}^{I}$ in space between the two involved states, which can be inferred by the size of the states. Therefore, it need not be specified by an abstract implementation. Note also a subtlety: implementing a transition might require auxiliary space temporarily 
exceeding both $\sizespace{\state}$ and $\sizespace{\statetwo}$, which we are not accounting for. The point is that 
for all the machines considered in this paper, such a temporary extra space is 
bounded by the current space (that is, $\sizespace{\state}^{I}$), and taking it into 
account would 
affect the globally used space only linearly, which is reasonable for space. 
Therefore, the auxiliary use of space can, and shall, be safely omitted. %Such a  use of space, however, is generally reflected by the time cost.

\section{The KAM and its Implementations}
\label{sect:naive-kam}
 \begin{figure*}[t]
	\input{machines/KAM}
	\vspace{-8pt}
	\caption{Data structures and transitions of the \OutKAM.}
	\vspace{-8pt}
	\label{fig:kam}
\end{figure*}
The Krivine abstract machine \cite{krivine_call-by-name_2007} is a 
standard environment-based machine for \ccbn, often defined as in \reffig{kam}. We 
refer to it as \emph{the \OutKAM}, to distinguish it from forthcoming variants and to stress that its specification is too abstract, as the role of low-level aspects such as the immutable initial code, pointers, and how the abstract data structures are concretely implemented is not made explicit, while it is essential for complexity analyses.

The machine evaluates closed $\l$-terms to weak head normal form via three 
transitions, the union of which is noted $\tooutkam$: 
\begin{itemize}
\item $\tokamsea$ looks for redexes descending on the left of topmost applications 
of the active term, accumulating arguments on the stack; 
\item $\tokamb$ fires a $\beta$ redex (given by an abstraction as active term 
having as argument the first entry of the stack) but delays the associated 
meta-level substitutions, adding a corresponding explicit substitution to the 
environment;
\item $\tokamsub$ is a form of micro-step substitution: when the active term is $\var$, the machine looks up the environment and retrieves the delayed replacement for $\var$.
\end{itemize}
The abstract data structures used by the \OutKAM are \emph{local environments}, \emph{closures}, and a \emph{stack}. \emph{Local environments}, which we shall simply refer to as \emph{environments}, are defined by mutual induction with \emph{closures}.
The idea is that every (potentially open) term $\tm$ in a dynamic state comes with an 
environment $\lenv$ that \emph{closes} it, thus forming a 
closure $\clos = (\tm,\lenv)$, and, in turn, environments are lists of entries 
$\esub\var\clos$ associating to each open variable $\var$ of $\tm$, a closure 
$\clos$ \ie, essentially, a closed term.
The \emph{stack} simply collects the closures associated to the arguments met during the search for $\beta$-redexes.

A dynamic state $\state$ of the \OutKAM is the pair $(\clos,\stack)$ of a 
closure $\clos$ and a stack $\stack$, but we rather see it as a triple 
$(\tm,\lenv,\stack)$ by spelling out the two components of the closure $\clos = 
(\tm,\lenv)$.
Initial dynamic states of the \OutKAM
are defined as $\compil{\tm_{0}}\defeq\kstate{\tm_{0}}{\stempty}{\stempty}$ (where $\tm_{0}$ is a closed $\l$-term, and also the code). The decoding of closures and states is as follows:
\[\begin{array}{r@{\hspace{1cm}}rcl@{\hspace{1cm}} rcl}
\textsc{Closures}
&
\unf{(\tm, \stempty)} & \defeq & \tm
&
\unf{(\tm, \esub\var\clos\cons\lenv)} & \defeq & \unf{(\tm\isub\var{\unf\clos},\lenv)}

\\[3pt]
\textsc{States}
&
\unf{(\tm, \lenv, \stempty)} & \defeq & \unf{(\tm, \lenv)}
&
\unf{(\tm, \lenv, \stack\cons\clos)} & \defeq & \unf{(\tm, \lenv, \stack)}\, \unf\clos
\end{array}\]

\paragraph{Basic Qualitative Properties.} Some standard facts about the \OutKAM  
follow. Let  $\run:\compil{\tm_{0}} \tooutkam^{*} \state$ be a run.
\begin{itemize}
\item \emph{Closures-are-closed invariant}: if the code $\tm_{0}$ is closed (that is the only case we consider here) then every closure 
$(\tmtwo, \lenv)$ in $\state$ is \emph{closed}, that is, for any free variable 
$\var$ of $\tmtwo$ there is an entry $\esub\var\clos$ in $\lenv$, and recursively so for the closures in $\lenv$. Thus 
$\unf{(\tmtwo,\lenv)}$ is a closed term, whence the name \emph{closures}.
\item \emph{Final states}: the previous fact implies that the machine is never stuck on the left of a $\tokamsub$ transition because the environment does not contain an entry for the active variable. Final states then have shape 
$\kstate{\la\var\tmtwo}{\lenv}{\stempty}$. 
\end{itemize}
\begin{thm}[Implementation]
\label{thm:kam-impl}
The \OutKAM implements \ccbn, that is, there is a complete $\towh$-sequence $\tm \towh^{n} \tmtwo$ if and only if there is a complete run $\run: \compil\tm \tonakam^{*} \state$ such that $\unf\state = \tmtwo$ and $\sizeb\run = n$.
\end{thm}
The proof of the facts and of the theorem are standard and omitted. Similar statements hold for all the variants of the \KAM that we shall see, with similar proofs which shall be omitted as well (we shall also omit the decoding of the variants of the \KAM).

The key point is that there is a bijection between $\towh$ 
steps and $\tokamb$ transitions, so that we can identify the two. Moreover, the 
number of  $\towh$ steps is a reasonable cost model for time, as first proved 
by Sands et al. \cite{DBLP:conf/birthday/SandsGM02}.
\paragraph{Quantitative Properties.} We recall also some less known \emph{quantitative} facts, for runs as above, from papers by Accattoli and co-authors \cite{DBLP:conf/rta/AccattoliL12,DBLP:conf/icfp/AccattoliBM14,DBLP:conf/ppdp/AccattoliB17}. The aim is to bound quantities relative to the run $\run$ and the reachable state $\state$. The bounds are given with respect to two parameters: the size $\size{\tm_{0}}$ of the code and the number $\sizeb\run$ of $\beta$-transitions, which, as mentioned, is an abstract notion of time for \ccbn.
\begin{itemize}
\item \emph{Number of transitions}: the number $\sizesub\run$ of $\subsym$ 
transitions in $\run$ is bounded by $\bigo{\sizeb\run^{2}}$, and there are terms 
on which the bound is tight. The number $\sizesea\run$ of $\seasym$ transitions is 
bounded by $\bigo{\sizeb\run^{2}\cdot \size{\tm_{0}}}$, but on complete runs the 
bound improves to $\bigo{\sizeb\run}$.

\item \emph{Sub-term invariant}: every term $\tmtwo$ in every closure 
$(\tmtwo,\lenv)$ in every reachable state is a literal 
(that is, \emph{not} up to $\alpha$-renaming) sub-term of the code $\tm_{0}$ . Therefore, 
in particular $\size\tmtwo \leq \size{\tm_{0}}$.

\item \emph{The length of a single environment}: the number of entries in a single 
environment is bounded by the size $\size{\tm_{0}}$ of the code.

\item \emph{The number of environments}: the number of distinct environments in 
$\state$ is bounded only by $\sizeb\run$.
\item \emph{The length of the stack}: the length of the stack in $\state$ is 
bounded by $\bigo{\sizeb\run^{2}\cdot \size{\tm_{0}}}$.

\end{itemize}

\paragraph{Sub-Term Pointers and Data Pointers: the \LinkKAM} The \OutKAM is usually implemented using an immutable initial code and \emph{two} forms of pointers, obtaining what here refer to as the \emph{\LinkKAM}:
\begin{enumerate}
\item \emph{Sub-term pointers} $\pointer\tm, \pointer\tmtwo$: the initial term $\tm_{0}$ provides the initial immutable code. The essential sub-term invariant mentioned above allows one to represent the active term and the terms $\tmtwo$ in every closure $(\tmtwo,\lenv)$ of every reachable state with a pointer $\pointer\tmtwo$ to $\tm_{0}$, instead that with a copy of $\tmtwo$. %Here, we shall assume that these pointers are left addresses (see \refdef{left-address}), which we stress using the notation $\lpointer_\tm, \lpointer_\tmtwo$.

\item \emph{Data (structure) pointers} $\pointer\lenv, \pointer{\lenvtwo}$: to ensure that the duplication of the environment $\lenv$ in transition $\tokamsea$ can be implemented efficiently (in time), environments are shared so that what is duplicated is just a pointer to an environment, and not the environment itself. This means that environments entries are stored in the \emph{heap} (or global environment), a new data structure which is simply a store, and that environments are pointers to heap entries.
\end{enumerate}
The \LinkKAM is in \reffig{LinkKAM}; explanations and comments follow.
\begin{itemize}
\item \emph{Environment entries and look-up}. Note that environment entries now have shape $\esub{\pointer{\la\var\tm}}\clos$ instead of $\esub{\var}\clos$, since they pair the pointer $\pointer{\la\var\tm}$ to the binder of $\var$ (with the closure) rather than the variable $\var$. Consequently, the function $\lenv(\var)$ looking up environments is now replaced by a function $\lenv(\pointer\var)$ acting on pointers, which first retrieves the pointer $\pointer{\la\var\tm}$ of the binder of $\pointer\var$ and then looks up the list structure of $\lenv$ in the heap for the closure $(\pointer\tmtwo,\pointer\lenvtwo)$ associated to $\pointer{\la\var\tm}$.

\item \emph{Stack pointers}. The representation of stacks can also be (data-)pointer-based or flat. In the \LinkKAM, we assume that stacks are flat. In extensions of the $\l$-calculus with control operators (which are not treated in this paper), stacks can be duplicated and so therein it would be more natural to have pointer-based representations of stacks, to enable their sharing.

\item \emph{Sub-term pointers}. Note that having made pointers explicit is still not concrete enough for precise complexity analyses, in the case of sub-term pointers. Using tree addresses (\refdef{tree-address}), one can obtain $\pointer\tm$ and $\pointer\tmtwo$ from $\pointer{\tm\tmtwo}$ (as required by transition $\tokamsea$) in constant time, but tree addresses in general have size $\bigo{\size{\tm_0}}$. Using left addresses  (\refdef{tree-address}), obtaining $\pointer\tm$ and $\pointer\tmtwo$ requires instead more time, namely $\bigo{\size{\tm_0}}$, but left addresses use less space, namely $\bigo{\log(\size{\tm_0})}$.
\end{itemize}

 \begin{figure*}[t]
	\input{machines/LinkKAM}
	\vspace{-8pt}
	\caption{Data structures and transitions of the \LinkKAM.}
	\vspace{-8pt}
	\label{fig:LinkKAM}
\end{figure*}

There is a big difference between sub-term and data pointers. As already mentioned, sub-term pointers can be crafted as to have size
$\bigo{\log\size{\tm_{0}}}$. For the present discussion they are 
\emph{space-friendly}, because their size does not depend on the length of the 
run. 
%We shall inspect them in \refsect{moving}, where we shall  ensure that also their number is under control, that is, independent of the length of the run. 
Data pointers, on the 
other hand, are \emph{space-hostile}, because (as recalled above) the number of 
environments is bounded only by $\sizeb\run$, that is, \emph{(abstract) time}. Data pointers 
have thus size $\bigo{\log{\sizeb\run}}$, entangling space with time, which is 
unreasonable for space. In the next section we shall add garbage collection, which disentangles space from time, but data pointers shall still add a $\bigo{\log{E}}$ overhead, where $E$ is the number of environment entries, which is excessive for space reasonability. Therefore, we now remove data pointers altogether, turning to a \emph{flat} representation of environments (and stacks), as explained in the next paragraph.

 \begin{figure*}[t]
	\input{machines/STKAM}
	\vspace{-8pt}
	\caption{\STKAM.}
	\vspace{-8pt}
	\label{fig:STKAM}
\end{figure*}

\paragraph{Sub-Term and Naive KAM} Summing up, as a reference we want a KAM adopting sub-term pointers but avoiding data pointers. Such a KAM can be presented in two ways. The first one is the \STKAM in \reffig{STKAM}, where the immutable code and sub-terms pointers are explicit, while environments are presented as in the \OutKAM, to suggest that they are flat rather than linked via data pointers.

We shall however adopt a different presentation, for two reasons. Firstly, the \SpKAM of the next section needs further tweaks of the KAM (namely garbage collection and environment unchaining) and, to avoid a too heavy treatment, we prefer to hide sub-term pointers and the immutable code. We hope that the discussions of this section have clarified how these two points are managed. Secondly, as already mentioned, the \STKAM is still not detailed enough for unambiguous complexity analysis. Therefore, we rather revert to the \OutKAM \emph{but} we pair it with an abstract implementation specifying the cost of the involved components and transitions as intended for the \STKAM (with left sub-term addresses). The pair of the \OutKAM and the naive abstract implementation is referred to as the \emph{\NaKAM}, defined in \reffig{NaKAM}, the transition relation of which is noted $\tonakam$.

%Since data pointers are hostile, we want them to be explicitly accounted for by the specification of the machine. Therefore, we consider the \NaKAM as being implemented with sub-term pointers but \emph{not} data pointers. This fact is expressed by the abstract implementation of the \NaKAM.
 \begin{figure*}[t]
	\input{machines/NaKAM}
	\vspace{-8pt}
	\caption{\NaKAM = \OutKAM + Naive Abstract Implementation.}
	\vspace{-8pt}
	\label{fig:NaKAM}
\end{figure*}
\paragraph{Naive Abstract Implementation} Implementing the \NaKAM without 
data pointers means that environments and stacks are implemented as 
unstructured \emph{strings}, in a linear syntax. We abstract from the actual 
encoding, what we retain is the abstract implementation in 
\reffig{kam}, which captures its essence. Sub-terms are assumed to be of size $\log(\size{\tm_0})$, which is obtainable by adopting left addresses.

The time cost of all $\tonakam$ transition depends polynomially on the size of the whole source state 
$\size\state$, because the lack of data sharing forces to use a new string for the 
new stack and the new environment; in particular, transition $\tokamsea$ requires to copy the whole string representing $\lenv$. To be precise, one could develop a finer 
analysis, thus obtaining slightly better bounds, but this would require entering 
in the details of the implementation and would not give a substantial 
advantage. As we shall see, indeed, the \NaKAM is unreasonable for abstract time (\refprop{spkam-not-reas-for-nat-time}). Via an analysis of the \NaKAM execution of the encoding of TMs, it shall turn out that also the space usage of the \NaKAM is unreasonable (\refprop{toy}). A space-reasonable refinement of the \NaKAM is the topic of the next section.

%Transition 
%$\tokamb$ does not require additional space but it takes time proportional to the 
%size of the source state $\state$. Transition $\tokamsub$ potentially frees 
%space, 
%and takes time bound by $\size\lenv$, for looking $\var$ up in $\lenv$ and 
%extracting $\lenvtwo$. 

%In general, the absence of pointers and sharing disentangles the \NaKAM use of 
%space from the  
%\NaKAM time, this can be noticed, \eg, in the $\tokamsub$ transition: \ben{with environment pointers and sharing, $\lenv$ might be referenced by some other closure and cannot in general be collected, while here instead it can (or, rather, parts of it can, as $\lenvtwo$ has to stay).} 

%The following sections, 
%however, shall show that, as it is, the 
%\NaKAM is not reasonable for time nor space. In the next Section, we discuss a 
%standard refinement (the \TiKAM) which is reasonable for time but not for 
%space. 
%Then we shall develop another machine (the \SpKAM) that is reasonable for 
%space 
%and not for time, and finally we shall discuss how to obtain reasonability for 
%time and space at the same time. 

%What about its time complexity? It crucially 
%depends on the size of $\lenv$. We shall show (in \refsect{yyy}) that the size of 
%$\lenv$ can be exponential in $\sizeb\run$. Therefore, the \NaKAM is not a time 
%reasonable machine for \ccbn, at least when time in \ccbn is set to be the number 
%of $\beta$ steps.

% !TeX spellcheck = en_US
% !TEX root = main.tex
%%%%%%%%%%%%%%

\section{The Space KAM}
\label{sect:spkam}
Here we define an abstract space optimization of the \OutKAM, dubbed \emph{\CoKAM}, which when paired with the same abstract implementation of the \NaKAM shall give the \SpKAM. The \CoKAM is 
derived from the \OutKAM by adding two 
modifications aimed at space efficiency: namely \emph{unchaining} and 
\emph{eager 
garbage collection}.
 %A less visible tweak is that the \SpKAM also adopts a specific address scheme for sub-term pointers.

%leading to a simulation of the toy algorithm of \refprop{toy} in 
%logarithmic space. 

\paragraph{Unchaining.} Environment unchaining is a 
folklore 
optimization for abstract 
machines  bringing speed-ups with respect to both time and space, used for instance 
by Sands et al. \cite{DBLP:conf/birthday/SandsGM02}, Wand 
\cite{DBLP:journals/lisp/Wand07}, Friedman et al. 
\cite{DBLP:journals/lisp/FriedmanGSW07}, and Sestoft 
\cite{DBLP:journals/jfp/Sestoft97}. Its first 
systematic study is by Accattoli and Sacerdoti Coen in 
\cite{DBLP:journals/iandc/AccattoliC17}, with respect to time. The optimization 
prevents the creation of chains of \emph{renamings} in environments, that is, 
of delayed substitutions of variables for variables, of which the simplest 
shape in the \KAM is:
\[\esub{\var_0}{(\var_1, \esub{\var_1}{(\var_2, 
			\esub{\var_2}{\ldots})})}\]
where the links of the chain are generated by $\beta$-redexes having a variable 
as argument. On some families of terms, these chains keep growing, 
leading to the quadratic dependency of 
the number of transitions from $\sizeb\run$.

\paragraph{Eager Garbage Collection.} Besides the malicious chains connected to 
unchaining,  the \OutKAM is not parsimonious with space also because there is no 
garbage collection (shortened to GC). In transition $\tokamsub$, the current 
environment is discarded, so something is collected, but this is not enough. 
It is thus natural to modify the machine as to maximize GC and space re-usage, that is, as 
to perform it \emph{eagerly}.

%Somewhat surprisingly, perhaps, the complete run from $\toy\, 
%\encp\strone\Bool$ of the modified machine still uses space linear in 
%$\size{\strone}$. A further \emph{unchaining} optimization is needed.
\begin{figure*}[t]
	\input{machines/SpaceKAM}
	\vspace{-8pt}
	\caption{Transitions of the \CoKAM, which is the abstract layer of the \SpKAM.}
	\vspace{-8pt}
	\label{fig:spkam}
\end{figure*}

\paragraph{The \CoKAM} The \OutKAM optimized with both eager GC and unchaining 
(both optimizations are mandatory for space reasonability)
is here called \CoKAM and it is defined in \reffig{spkam}. The data structures, 
namely closures and (local) environments, are defined as before---the novelty concerns the machine transitions only. Unchaining is realized by transition 
$\tokamseav$, while
eager garbage collection is realized by transition $\tokambw$, which 
collects the argument if the variable of the $\beta$ redex does not occur, and by 
transitions $\tokamseav$ and $\tokamseanv$, by 
restricting the environment to the occurring variables, when the environment is 
propagated to sub-terms. As a consequence, we obtain the following invariant.

\begin{lem}[Environment domain invariant]
\label{l:spkam-invariants} % \reflemmap{spkam-invariants}{env}
Let $\state$ be a \CoKAM reachable state. Then $\dom\lenv = \fv\tm$ for every 
closure $(\tm,\lenv)$ in $\state$.
\end{lem}
Because of the invariant, which concerns also the closure given by the active term and the local environment of the state, the substitution transition $\tokamsub$ simplifies as follows:
\[\begin{array}{l@{\hspace{.3cm}} 
				l@{\hspace{.3cm}}l|l|l@{\hspace{.3cm}} 
				l@{\hspace{.3cm}}l}
			\mathsf{Term}   & 
			\mathsf{Env} & \mathsf{Stack}
			&&	\mathsf{Term} 
			& 
			\mathsf{Env}
			& \mathsf{Stack} 
			\\
			\cline{1-7}
			 \var    &
			\esub\var{(\tmtwo,\lenv)}
			& \stack
			 &
			\tokamsub &%\\[-2pt]
			 \tmtwo  & \lenv & \stack
	\end{array}\]
	
	\paragraph{Sub-Term Pointers and Abstract Implementation: the \SpKAM} We now add an abstract implementation to the \CoKAM, finally obtaining the \SpKAM. The abstract implementation to be added is the same of the \NaKAM, but for a crucial point. For the \NaKAM, the size of a term $\size{\tmtwo}$ is given by $\log{\size{\tm_{0}}}$, and we said that this is obtainable via left addresses. For the \SpKAM, we need a finer approach, or rather a finer definition. 
	\begin{defi}[\SpKAM]
The \SpKAM is defined as the \CoKAM together with the abstract implementation $I$ defined as for the \NaKAM except for the size of sub-terms, which is redefined as 
$\sizespace{\tmtwo}^{I} \defeq \size{{\tt ladd}(\tmtwo,\tm_{0})}$, where ${\tt ladd}(\tmtwo,\tm_{0})$ is the left address of $\tmtwo$ in the initial code $\tm_0$.
\end{defi}
In the complexity analysis of the encoding of TMs we shall use the fact that, for some sub-terms, $\size{{\tt ladd}(\tmtwo,\tm_{0})}$ is of size $\bigo{1}$ rather than $\bigo{\log{\size{\tm_{0}}}}$, and this shall be crucial for the space reasonability result. In fact, other addressing schemes might work as well. What is important for our result is that the size of sub-term pointers is always $\bigo{\log\size{\tm_0}}$ and that, for an application $\tmtwo\tmthree$, the size of the addresses in $\tmtwo$ is independent of $\tmthree$, which shall be used to infer that some pointers have size $\bigo{1}$ because of how the encoding of TMs is built.

Note that left addresses do not respect the \emph{tree locality} of terms, in the sense that in $\tm \tmtwo$ the first constructors of $\tm$ and $\tmtwo$ are not the address of the root application $\pm 1$. Their addresses can however be retrieved in time polynomial and space logarithmic in $\size{\tm_0}$, which is what is important. Note that such a non-locality of left addresses is one of the ways in which time is traded for space in the Naive/\SpKAM: manipulating terms via left addresses requires a (polynomial) time overhead with respect to use actual pointers to the code.

For the abstract implementation, it is fine to keep the same time bounds used for the \NaKAM,
because
the garbage collector has a time cost which however stays within the polynomial (in the 
size of the 
states) cost of the transitions. It is mandatory that it is implemented  
by naively and repeatedly checking whether variables occur, and \emph{not} via 
pointers or counters, as they would add an unreasonable space overhead. This 
fact is implicit in using the same abstract implementation of the \NaKAM, as a 
less naive GC would alter the space requirements.

\paragraph{Space Cost and Closure Space}
When we defined the space cost of a generic machine run (\refdef{space-time-runs}), we considered the max over the size of states. The size of states is a very concrete notion. Now that we have defined the \CoKAM and the \SpKAM, it is possible to define also a second, more abstract notion of space, here dubbed \emph{closure space}. For ease of language, we define it on the \SpKAM, but the abstract aspect of closure space is evident from the fact that it can already be defined on the \CoKAM.

\begin{defi}[Closure space]
Let $\run:\state_{0} \to^k\state_{k}$ be an initial run of the \CoKAM and use $\sizecl\state$ for the number of closures in a state $\state$. Then the closure space of $\run$ is defined as
	$\sizeaspace{\run}\defeq\max_{\state\in\run}\sizecl{\state}$.
\end{defi} 
Since concrete space is obtained by considering for each closure also the size of its sub-term pointer, one has $\sizeaspace{\run} \leq \lm{\run} \leq \bigo{\sizeaspace{\run}\cdot\log{\size{\tm_0}}}$. This is similar to what happens with abstract time (that is, the number of $\beta$-steps), where the actual time is more precisely bounded by abstract time times the size of the initial term. There is however an important difference. On the encoding of TMs, it shall turn out that $\sizespace{\run}^I \neq \Theta(\sizeaspace{\run}\cdot\log{\size{\tm_0}})$, because some sub-term pointers shall have size $\bigo{1}$ rather than $\bigo{\log\size{\tm_0}}$. Additionally, $\Theta(\sizeaspace{\run}\cdot\log{\size{\tm_0}})$ would not be a reasonable use of space. It is nonetheless useful to consider such an abstract notion of closure space, as the next paragraph shows.

\paragraph{Closure space and $\eta$-equivalence.}
We point out that the space consumption of the \SpKAM is almost
invariant with respect to $\eta$-expansion, thanks to unchaining. In particular, $\eta$-expansion preserves closure space, but not the concrete one:
$\eta$-expanding $\tm$ for $n$ times preserves the number of closures but causes the size of sub-term pointers in each closure to grow by no more than $\log(n)$. 
Formally, let 
us define $\eta$-expansion as the function $\eta:\Lambda\to\Lambda$ such that:
\[
\eta(\tm)\defeq \la\var\tm\var \ \ \ \mbox{with }\var\notin\fv\tm
\]
\begin{lem}
	Let $\tm$ be a closed $\lambda$-term,  and $\run_n$ the complete run 
	from $\spkamstate{\eta^n(\tm)}{\lenv}{\clos\cdot\stack}$. Then $\sizeaspace{\run_n} = \sizeaspace{\run_0}$ and 
	$\lm{\run_n}\in\bigo{\log(n)\cdot\lm{\run_0}}$.
\end{lem}
\begin{proof}
	We first prove that $\spkamstate{\eta^n(\tm)}{\lenv}{\clos\cdot\stack} 
	\rightarrow^{2n}_{\text{SpKAM}} \spkamstate{\tm}{\lenv}{\clos\cdot\stack}$. 
	We proceed by induction 
	on $n$ executing the \SpKAM. The case $n=0$ is 
	trivial. Then, we consider the case $n = m + 1$.
\[\begin{array}{l|l|ll}
	\mathsf{Term}   & 
	\mathsf{Env} & \mathsf{Stack}\\
    \hhline{---} \rule{0pt}{2.6ex} % introduced space over top row without gap in right border https://tex.stackexchange.com/q/65127/310868
	\eta^{m+1}(\tm)\defeq\la\var\eta^m(\tm)\var & \lenv & \clos\cdot\stack & 
	\tokambnw\\
	\tm\var & \esub\var\clos\cdot\lenv & \stack & \tokamseav\\
	\eta^m(\tm) & \lenv & \clos\cdot\stack & 	\rightarrow^{2n}_{\text{SpKAM}}\ (\ih)\\
	\tm & \lenv & \clos\cdot\stack
\end{array}\]
We observe that no space is consumed during this transitions. About pointer size, 
let us call $\tm_n$ the initial code. We have that the number of constructors of 
$\tm_n$ is $\size{\tm_n}=3n+\size{\tm_0}$ and thus its pointer size is 
$\log(3n+\size{\tm_0})$. Then, let us call $k$ the maximum 
number of closures stored in $\run_n$ (we have already observed that this is 
independent of $n$). Then
\[
\begin{array}{rcl}\lm{\run_n}&=&k\cdot \log(3n+\size{\tm_0}) \leq 
k\cdot(\log(3n)+\log(\size{\tm_0}))=k\cdot\log (3n) + 
k\cdot\log(\size{\tm_0})\\[3pt]
&\leq& k\cdot\log (3n)\cdot\log(\size{\tm_0}) + k\cdot\log(\size{\tm_0})=\log 
(3n)\cdot\lm{\run_0}+\lm{\run_0}\in\bigo{\log(n)\cdot\lm{\run_0}}
\end{array}\]
\end{proof}

% We can finally simulate the toy algorithm in 
%logarithmic space.
%\begin{proposition}
%\label{prop:toy}
%Let $\strone\in\Bool^*$. The sequence $\toy\, \encp\strone\Bool 
%\towh^{\bigo{\size\strone}} \Id$ is simulated by the \SpKAM using  space 
%$\bigo{\log\size\strone}$.
%\end{proposition}
%Please notice that \emph{both} the optimizations, unchaining and GC, are 
%required in order to obtain the logarithmic bound.
% !TeX spellcheck = en_US
% !TEX root = main.tex
%%%%%%%%%%%%%%
\section{Encoding and Moving over Strings}\label{sect:moving}
We now turn to the analysis of the encoding of TMs, taking as reference the one 
by Dal Lago and Accattoli based over the Scott encoding of strings 
\cite{DBLP:journals/corr/abs-1711-10078}. The first key step is understanding 
how to scroll Scott strings.

\paragraph{Encoding alphabets.} Let $\alpone=\{\elone_1,\ldots,\elone_n\}$ be a 
finite alphabet.
%---the actual encoding of TM uses three alphabets, the binary one 
%extended with the delimiters $\Boolin\defeq \set{0,1,\stsym,\ensym}$, the 
%binary 
%one extended with a blank symbols $\Boolb\defeq \set{0,1,\elemblank}$ and one 
%representing states $\States=\{\state_1,\ldots,\state_m\}$
 Elements of $\alpone$ 
are encoded in the $\l$-calculus in accordance to a fixed (but arbitrary) total 
order of the elements of $\alpone$ as follows:
\[
\cod{\elone_i}{\alpone} \defeq \lambda\var_1.\ldots.\lambda\var_n.\var_i \;.
\]
Note that the representation of an element $\cod{\elone_i}{\alpone}$ requires a 
number of constructors that is linear (and not logarithmic) in $\size\alpone=n$. 
Since the alphabet $\alpone$ shall not depend on the input of the TM, however, the 
cost in space is actually constant.

\paragraph{Encoding strings.} A string in
$\strone\in\alpone^*$ is represented by a term
$\encp\strone{\alpone^*}$. The encoding exploits the fact that a string is a concatenation of characters \emph{followed by the empty string $\varepsilon$} (which is generally omitted). For that, the encoding uses $\size\alpone+1$ abstractions, the extra one ($\var_\varepsilon$ in the definition below) being used to represent $\varepsilon$. The encoding is defined by induction on the structure of $\strone$ as follows:
\begin{align*}
	\enc{\varepsilon}^{\alpone^*} & \defeq 
	\la{\var_1}\ldots\la{\var_n}\la{\var_\varepsilon}\var_\varepsilon\;,\\
	\enc{\elone_i\strtwo}^{\alpone^*} & \defeq 
	\la{\var_1}\ldots\la{\var_n}\la{\var_\varepsilon}\var_i\enc{\strtwo}^{\alpone^*}.
\end{align*}
Note that the representation depends on the cardinality
of $\alpone$. As before, however, the alphabet is a fixed parameter, and so such a 
dependency is irrelevant. As an example, the encoding of the string $aba$ with respect to the alphabet $\set{a,b}$ ordered as $a<b$ is 
\begin{center}
$\begin{array}{lll}
\enc{aba}^{\set{a,b}} &=& \la{\var_a}\la{\var_b}\la{\var_\varepsilon}\var_a (\la{\var_a}\la{\var_b}\la{\var_\varepsilon}\var_b(\la{\var_a}\la{\var_b}\la{\var_\varepsilon}\var_a(\la{\var_a}\la{\var_b}\la{\var_\varepsilon}{\var_\varepsilon})))\end{array}$
\end{center}

\paragraph{Linear Space Representation Overhead.} As announced in \refsect{prelim}, we now explain how to remove the code-constructor gap, that is, how to make the size of variable pointers irrelevant for space (for the Scott encoding of strings). As already mentioned this is not mandatory for our result (it would be mandatory in a log-insensitive approach to linear space), but it is interesting to see that it is possible. 

Note that in $\encp\strone{\alpone^*}$ every variable occurrence is bound inside the list of binders immediately preceding the occurrence. If de Bruijn indices are used to represent $\l$-terms, one needs only indices---that is, variable pointers---between $1$ and $\size\alpone+1$, that is, of \emph{constant size}. Note that, similarly, if variables are represented with textual names, again having only $\size\alpone +1$ distinct names is enough \emph{if} one permits that different sequences of abstractions re-use the same names, that is, if one accepts Barendregt's convention to be broken. Remarkably, a notable folklore property of the (Space) KAM is that its implementation theorem does not need Barendregt's convention to hold. Therefore, de bruijn indices are not the only possible approach that removes the code-constructor gap on the Scott encoding of strings.

In their result about reasonable space, Forster et al. \cite{cbv_reasonable} also rely on the Scott encoding of strings and they represent $\l$-terms using de Bruijn indices. Since they study (super-)linear space in a log-insensitive way, their choice of de Bruijn indices is crucial for their result to hold, even if they do not stress it.

\paragraph{Recursion and Fix-Points} The encoding of TMs crucially relies on 
the 
use of a fix-point operator to implement recursion. Precisely, fix-points are 
used to model the transition function, making a copy of the (sub-term encoding 
the) transition table at each step. It is the only point of the encoding where 
duplication occurs, and it is thus where the expressive power is 
encapsulated. The rest of the encoding is affine---note 
that the representation of strings is affine.

\paragraph{Fix-Points and Toy Scrolling Algorithms.} To understand the delicate 
interplay between the space of the \KAM and fix-points, we analyze it via 
simple toy algorithms on strings. The first, simplest one is the 
\emph{consuming scrolling algorithm}: going through an input string $\str$ 
doing nothing and accepting when arriving at the end of the string, without 
having to preserve the string itself---the aim is just to see the space used 
for scrolling a string. The toy algorithm is a very rough approximations of the 
moving of TMs over a tape, which is the most delicate aspect of the space 
reasonable simulation of TMs in the $\l$-calculus that we shall develop. It is 
used to illustrate the key aspects of the problems that arise and of their 
solutions, without having to deal with all the details of the encoding of TMs 
at 
once. On TMs, scrolling a string obviously runs in constant space, and on 
log-sensitive TMs the consuming aspect cannot be modeled---we shall consider 
non-consuming scrolling later in this section. 

We encode the algorithm as a $\l$-term over Scott strings, where a fix-point combinator is used to iterate over the (term $\tm_{\str}$ encoding) the input string $\str$. Since the input string $\str$ is consumed in the process, the normal form would be the encoding of the accepting state $\state_{1}$ of the TM, which for simplicity here is simply given by the identity combinator $\Id$.

We use Turing’s fix-point combinator and the boolean alphabet 
$\Bool\defeq\{0,1\}$. Let $\fixnospace \defeq \theta\theta$, where $
\theta \defeq \la\var \la\vartwo \vartwo(\var \var \vartwo )$.
Given a term $\tmtwo$, $\fix\tmtwo$ is a fix-point of $\tmtwo$.
\[\begin{array}{l@{\hspace{.3cm}}c@{\hspace{.3cm}}l}
\fix\tmtwo  
& = &
(\la\var \la\vartwo \vartwo(\var \var \vartwo)) \theta 
\tmtwo 
\\
& \tob & 
(\la\vartwo \vartwo(\theta \theta \vartwo)) \tmtwo 

\quad \tob  \quad
\tmtwo (\theta \theta \tmtwo)
\quad = \quad
\tmtwo ( \fix \tmtwo )
\end{array}\]
Algorithms moving over binary Scott strings always follow the same structure. They 
are given by the fix-point iteration of a term that does pattern matching on the 
leftmost character of the string and for each of the possible outcomes (in our 
case, the first character is $0$, $1$, or the empty string $\ems$) does the corresponding 
action. The general term is $\fix (\la\varf\la {\varthree}\varthree 
A_{0}A_{1}A_{\ems})$, where $\varf$ is the variable for the recursive call and 
$A_{0}$, $A_{1}$, and $A_{\ems}$ represent the three actions, which in our case 
are simply given by $A_{0}=A_{1} = \varf$ and $A_{\ems} = \Id$, using the identity 
$\Id$ as encoding of the accepting state. Formally, we have the following Proposition, the proof of which is in Appendix~\ref{app:moving}.1.
%$\toyaux \defeq \la\varf{{\la {\strone'}{\strone' N N N_{\ems} }}}$
%where:
%\begin{itemize}
%	\item $N\defeq \la {\str'} \varf \str'$
%	\item $N_{\ems}\defeq \encp{\ems}\Bool$
%\end{itemize}
\begin{prop}
\label{prop:toy}
Let $\strone\in\Bool^*$ and $\toy \defeq \fix (\la\varf\la {\varthree}\varthree 
\varf\varf \Id )$.
\begin{enumerate}
\item $\toy\, \encp\strone\Bool \towh^{\Theta(\size\strone)} \Id$.
\item \label{point:nakam} The \NaKAM evaluates $\toy\, \encp\strone\Bool$ in 
space $\Omega(2^{\size\strone})$.
\item The \SpKAM evaluates $\toy\, \encp\strone\Bool$ in space 
$\Theta(\log\size\strone)$.
\end{enumerate}
\end{prop}

%Scrolling strings is done in constant space on TM. A space reasonable 
%execution 
%schema in the $\l$-calculus should then take at most logarithmic space. The 
%\NaKAM 
%uses at least linear space, and it is then space \emph{unreasonable}.
% !TeX spellcheck = en_US
% !TEX root = main.tex
%%%%%%%%%%%%%%
%\section{Scrolling Tapes}
%\label{sect:scrolling-tapes}
We can see that the \NaKAM is desperately inefficient for space, while the 
\SpKAM works within reasonable space bounds.
It turns out, however, that the \SpKAM is still not enough in order to obtain a space reasonable 
simulation of TMs. The problem now concerns the standard encoding of TMs and 
its 
managing of the tapes, rather than the use of space by the abstract machine 
itself. The issues  can be explained using further toy algorithms.

\paragraph{String-Preserving Scrolling} Consider the same scrolling algorithm 
as above, except that now the input string $\str$ is \emph{not} consumed by 
the moving over $\str$, that is, it has to be given back as output of the 
$\l$-term implementing the algorithm. This variant is a step forward towards 
approximating  what happens to the tapes of TMs \emph{during} the computation: 
the TM moves over the tapes \emph{without consuming them}, it is only at the  
end of the computation that the TM can be seen as discarding the tapes.
There are two ways of implementing the new algorithm:
\begin{enumerate}
\item \emph{Local copy}: moving over the string $\str$ while accumulating in a new accumulator string $\strtwo$ the characters that have already been visited, returning $\strtwo$. 
\item \emph{Global copy}: making a copy $\strtwo$ of the string $\str$, and 
then moving over $\str$ in a consuming way, returning $\strtwo$.
\end{enumerate}

\paragraph{Local Copy} The local approach is the one underlying  the 
reference encoding of TMs. In particular, it is almost affine, as duplication 
is 
isolated in the fix-point. The $\l$-term $\localcopy$ realizing it uses the same fix-point schema as before, but with different, more involved 
action terms $A_{0}$, $A_{1}$, and $A_{\ems}$. We provide the following Proposition (and the next one) without proof, since $\localcopy$ is just a fragment of the encoding of TMs, for which we detail the execution by the \SpKAM  in Appendix~\ref{sec:main-proof}.
\begin{prop}
\label{prop:toytape}
Let $\strone\in\Bool^*$.
\begin{enumerate}
\item $\localcopy\, \encp\strone\Bool \towh^{\Theta(\size\strone)} \encp\strone\Bool$.
\item \label{p:toytape-2}The \SpKAM evaluates $\localcopy\, \encp\strone\Bool$ 
in space $\bigo{\size\strone\log\size\strone}$.
\end{enumerate}
\end{prop}
The $\bigo{\size\strone\log\size\strone}$ bound in point 2 is problematic for the space reasonable modeling in the $\l$-calculus of both the input and the work tapes, for different reasons. %We first address the work tape.

\paragraph{Work Tape and Left Addresses} For a space-reasonable managing of the work tape, a scrolling algorithm should rather work in space $\bigo{\size\strone}$. This improvement can be obtained by keeping the same algorithm and refining the complexity analysis. In \refpropp{toytape}{2}, the cost comes from the use of $\bigo{\size\str}$ sub-term pointers to the code $\localcopy\, \encp\strone\Bool$ used by the \SpKAM. These pointers have size $\bigo{\log\size\str}$ because $\size{\localcopy\, \encp\strone\Bool} = \bigo{\size\str}$, that is, the size of $\localcopy$ is independent of $\size\str$ and thus constant. A close inspection of the \SpKAM run in \refpropp{toytape}{2} shows that, of the used $\bigo{\size\str}$ pointers, only $\bigo{1}$ of them actually point to $\encp\strone\Bool$, while all the others (that is, an $\bigo{\size\str}$ amount) point to $\localcopy$. Since $\localcopy$ is of size independent from $\size\str$,  if one admits separate address spaces for $\localcopy$ and $\encp\strone\Bool$, as it is done using left addresses, then the pointers to $\localcopy$ have size $\bigo{1}$. And indeed the size of the left addresses that the \SpKAM uses for sub-terms pointers give addresses to the sub-terms in $\localcopy$ is independent from $\encp\strone\Bool$. Therefore, one obtains that the space cost is given by 
\[\underbrace{\bigo{\size\str}\cdot\bigo{1}}_{\mbox{pointers to 
}\localcopy}\ +\quad \underbrace{\bigo{1}\cdot 
\bigo{\log\size\str}}_{\mbox{pointers to }\encp\strone\Bool} = 
\bigo{\size\str}.\]

\begin{prop}[Linear Space Local-Copy Scrolling]
\label{prop:toytape2}
Let $\strone\in\Bool^*$. The \SpKAM evaluates $\localcopy\, \encp\strone\Bool$ 
in space $\bigo{\size\strone}$.
\end{prop}

\paragraph{Input Tape and Global Copy} For the input tape, a linear space bound for scrolling is unreasonable, if one aims at preserving logarithmic space complexity. For
meeting the required $\bigo{\log\size\strone}$ bound, we need a more radical solution, which shall be possible because the tape is read-only (and thus the solution does not directly apply to the work tape). 

The first step is the straightforward modification of the consuming scrolling algorithm into a global-copy string-preserving algorithm: it is enough to capture the input at the beginning with an extra abstraction $\l\var$ and to give it back at the end with the action $A_{\ems}$, that is, having $A_{\ems} \defeq \var$. Namely, let $\globalcopy \defeq \la\var(\fix (\la\varf\la {\strone'}\strone' \varf\varf \var) \var)$. Clearly, this approach breaks the almost affinity of the encoding, as copying is no longer encapsulated only in the fix-point. Formally, we have the following Proposition, the proof of which is in Appendix~\ref{app:moving}.2.
\begin{prop}
\label{prop:toytapeglobal}
Let $\strone\in\Bool^*$.
\begin{enumerate}
\item $\globalcopy\, \encp\strone\Bool \towh^{\Theta(\size\strone)} \encp\strone\Bool$.
\item The \SpKAM evaluates $\globalcopy\, \encp\strone\Bool$ in space 
$\Theta(\log\size\strone)$.
\end{enumerate}
\end{prop}
Interestingly, the space cost stays logarithmic, because the global copy of the input  in point 1 (in fact there actually is a copy for every iteration of the fix-point) is not performed by the \SpKAM, which instead only copies a \emph{pointer} to it. The second step is refining this scheme as to implement a read-only tape, rather than just scrolling the tape. A slight digression is in order.

\paragraph{Intrinsic and Mathematical Tape Representations}  A TM tape is a string plus a distinguished position, representing the head. There are two tape representations, dubbed \emph{intrinsic} and \emph{mathematical} by van Emde Boas in  \cite{DBLP:conf/sofsem/Boas12}. 
\begin{itemize}
\item The \emph{intrinsic} one represents both the string $\str$ and the current position of the head as the triple $\str = \str_{l} \cdot h \cdot \str_{r}$, where $\str_{l}$ and $\str_{r}$ are the prefix and suffix of $\str$ surrounding the character $h$ read by the head. This is the representation underlying the local-copy scrolling algorithm as well as the reference encoding of TMs. In this approach, reading from the tape costs $\bigo{1}$ time but the reading mechanism comes with a $\bigo{\size\str}$ space overhead, as showcased by \refprop{toytape2}.

%%%%%%%
\item
 The \emph{mathematical} representation, instead, is simply given by the index $n\in\nat$ of the head position, that is, the triple $\str_{l} \cdot h \cdot \str_{r}$ is replaced by the pair $(\str, \size{\str_{l}}+1)$. The index $\size{\str_{l}}+1$ has the role of a pointer, of logarithmic size when represented in binary. 
\end{itemize}

\paragraph{Mathematical Input and Global Copy} Given a \emph{mathematical} read-only tape $(\str,n)$, one can use the global-copy scrolling scheme for a simulation in the $\l$-calculus in space $\bigo{\log\size\str}$. The idea is to represent $n$ as a binary string $\ntostr{n}$. Since $n\leq \size\str$, we have $\size{\ntostr{n}}\leq \log\size\str$. Moreover, it is possible to pass from $\ntostr{n}$ to $\ntostr{n+1}$ or $\ntostr{n-1}$---which is needed to move the position of the head---in $\bigo{\log\size\str}$ space. Finally, reading from the tape, that is, given a tape $(\str,n)$, returning $(\str,n)$ plus the $n$-th character $\str_n$ of $\str$, is doable in space $\bigo{\log\size\str}$ via the following composite operation:
\begin{itemize}
\item making a global copy of the tape (returned at the end), 
\item scrolling the current copy of $n$ positions, 
\item extracting the head $\str_n$ of the obtained suffix, and 
\item discarding the tail.
\end{itemize}

Two remarks. First, this approach works because the tape is read-only, so that one can keep making global copies of the same immutable tape, and only changing the index of the head. Second, there is a (reasonable) time slowdown, because at each read the simulation has to scroll sequentially the input tape to get to the $n$-th character. Such a scrolling has time cost $\bigo{n\log n}$ because it is done by moving of (the encoding of ) a cell at a time, and  decrementing of one the index, until the index is $0$. Therefore, accessing the right cell requires to decrement the index $n$ times, each time requiring time $\log n$, because the index is represented in binary. Therefore, the time cost of a read operation is $\bigo{\size\str\log\size\str}$.

\section{The \SpKAM is Reasonable for Space}\label{sect:reasonability}
We are ready for our main result. It is based on a new variant over 
Dal Lago and Accattoli encoding of TMs into the $\l$-calculus 
\cite{DBLP:journals/corr/abs-1711-10078} which is defined in a separate document \cite{log-encoding} (including also the encoding in the $\l$-calculus of the binary arithmetic needed for the mathematical representation of the input tape), to spare the tedious details of the encoding to the reader. The key points 
are: 
\begin{itemize}
\item \emph{Refined TMs}: the notion of TM we work with is 
log-sensitive TMs with mathematical input tape and intrinsic work tape (the 
formal definition of TMs is in \cite{log-encoding}). 

\item \emph{CPS and indifference}: following \cite{DBLP:journals/corr/abs-1711-10078}, the encoding is in \emph{continuation-passing style}, and carefully designed (by adding some $\eta$-expansions) as to fall into the \emph{deterministic $\l$-calculus} $\detLam$, a particularly simple fragment of the $\l$-calculus where the right sub-terms of applications can only be variables or abstractions and where, consequently, call-by-name and call-by-value collapse on the same evaluation strategy $\tobdet$. We shall exploit this \emph{indifference} property in \refsect{cbv}.
\item \emph{Duplication}: duplication is isolated in the unfolding of fix-points and in the managing of the input tape, all other operations are affine.
\end{itemize}

\begin{thm}[TMs are simulated by the \SpKAM in reasonable space]
\label{thm:main-simulation}
There is an encoding $\enc{\cdot}$ of log-sensitive TMs into 
$\detLam$ such that if the run $\run$ of the TM $M$ on input $\istr\in\Bool^{*}$:
\begin{enumerate}
\item \label{p:main-simulation-one}
\emph{Termination}: ends in $\state_{\bool}$ with $\bool\in\Bool$, then
there is a complete sequence $\runtwo:\enc M\, \enc\istr \tobdet^n \enc{\state_{\bool}}$
where $n=\Theta((T_{\textrm{TM}}(\run)+1)\cdot \size\istr\cdot 
\log{\size\istr})$.

\item \label{p:main-simulation-two}
\emph{Divergence}: diverges, then $\enc M\, \enc\istr$ is $\tobdet$-divergent.

\item \label{p:main-simulation-three}
\emph{\SpKAM}: the space used by the \SpKAM to simulate the evaluation of  
point 1 is $\bigo{S_{\textrm{TM}}(\run) +\log\size\istr}$.
\end{enumerate}
\end{thm}
The previous theorem provides the subtle and important half of the space reasonability result. The first two points, proved in \cite{log-encoding}, establish the qualitative part of the simulation in the $\l$-calculus, together with the time bound (with respect to the number of $\beta$ steps). They are connected to the \SpKAM by the fact that the \SpKAM implements closed call-by-name (that coincides with $\tobdet$ in $\detLam$), via an omitted minor variant of \refthm{kam-impl}.
The third point is the important one, as it provides the space-reasonable simulation of the $\l$-calculus by the \SpKAM. It is proved by directly executing on the \SpKAM the $\l$-terms which are the image of the encoding of TMs. The proof is a tedious analysis of machine executions, and it is thus developed in Appendix~\ref{sec:main-proof}. The main ingredient is an invariant stating that, during the execution of the encoding on the \SpKAM, configurations of the encoded TM are represented by \SpKAM configurations of the same size.

The other half of the reasonability result amounts to showing that the \SpKAM can be simulated 
on TMs within the space costs claimed in \refsect{spkam}. The idea is 
that it can clearly be simulated reasonably by a multi-tape TM using one work tape 
for the active term (as a pointer to the fixed initial code), one for the 
environment, one for the stack, plus one for auxiliary pointers manipulations. Note that to encode the $\l$-calculus we use a notion of TM which is different from the one that we encoded in the $\l$-calculus, as there is no binary output, rather there is an output tape which, for the simulation, at the end of the execution is filled in with the content of the tape for the active term and the tape for the environment, which provide a shared representation of the result $\l$-term.% and that (the RAM 
%representation of the) final states of the \SpKAM verify the \emph{linear space 
%constructor equality} constraint of \refdef{reasonable}.

\begin{prop}[\SpKAM is simulated by TMs in reasonable space]
\label{thm:converse-simulation}
Let $\tm$ be a $\l$-term. Every \SpKAM run $\exec: \compil\tm 
\tospkam^{*} \state$  is implemented on TMs in space $\bigo{\lm{\run}}$. 
%Moreover, the constructor equality test of \refdef{reasonable} can 
%be done in space linear in the size of the address on RAMs.
\end{prop}
\paragraph{Logarithmic Space Constructor Equality.} The definition of reasonable space cost model (\refdef{reasonable}) also requires that final configurations of the TM used in 
the simulation can be inspected in logarithmic space. For that, we provide a pseudo-code algorithm, presented as an inductive definition, that, given the initial term $\tm_0$, a closure $\clos_0=(\tm_{\symfont{fin}},\lenv_{\symfont{fin}})$ (meant to be the closure of the final state of the \SpKAM run), and a tree address 
$\addr_0$ (with respect to the decoded final term $\unf\clos$), returns the constructor of the term $\unf {\clos_0}$ at $\addr_0$. The algorithm navigates through $\clos_0$ using a notion of current closure $\clos=(\tmtwo,\env)$ which is either $\clos_0$ or a closure somewhere in $\lenv_{\symfont{fin}}$. The navigation process is realized via three pointers:
\begin{itemize}
\item \emph{Term pointer}: a pointer inside $\tmtwo$, which we represent abstractly as a decomposition of $\tmtwo$ into a context $\ctx$ and a sub-term $\tmthree$, thus using a triple $(\tmthree,\ctx,\lenv)$ to represent the closure with pointer $(\ctxp\tmthree,\lenv)$. Since $\tmtwo$ is a sub-term of the initial code, the term pointer actually moves over the initial code $\tm_0$. 

\item \emph{Address pointer}: a pointer inside $\addr_0$, which, similarly to the term pointer, is represented as a pair of addresses $(\addr,\addrtwo)$ such that $\addr\cdot\addrtwo=\addr_0$.

\item \emph{Closure pointer}: a pointer inside $\clos_0$ to the current closure. For the sake of simplicity, we do not represent the closure pointer explicitly (it would require to introduce closure contexts, and the technicality is not worth it).
\end{itemize}
Given the initial closure $\clos_0=(\tm_{\symfont{fin}},\lenv_{\symfont{fin}})$ and the initial address $\addr_0$, the algorithm is invoked as $\staddr{(\tm_{\symfont{fin}},\ctxhole,\env_{\symfont{fin}})}{(\ems,\addr_0)}$. It is defined as follows.
%\[\begin{array}{lcl@{\hspace{.9cm}}llc}
%		\staddr{(\var,\lenv)}\addr & \defeq&  \begin{cases}
%			\db(\var) & \text{ if }\var\not\in\dom\lenv \text{ and }\addr=\ems,\\
%			\bot & \text{ if }\var\not\in\dom\lenv \text{ and }\addr\neq\ems,\\
%			%\staddr{\lenv(\var)}\addr & \text{ otherwise.}
%			\staddr{\lenv(\var)}\addr &  \text{ if }\var\in\dom\lenv 
%		\end{cases} & 
%		\staddr{(\tm\tmtwo,\lenv)}{0\cons\addr} & \defeq 
%		&\staddr{(\tm,\lenv|_{\tm})}\addr\\[25pt]
%		\staddr{(\la\var\tm,\lenv)}{b\cons\addr} & \defeq 
%		&\staddr{(\tm,\lenv)}\addr & 
%		\staddr{(\tm\tmtwo,\lenv)}{1\cons\addr} & \defeq 
%		&\staddr{(\tmtwo,\lenv|_{\tmtwo})}\addr\\[3pt]
%		\staddr{(\la\var\tm,\lenv)}{\ems} & \defeq & \lambda & 
%		\staddr{(\tm\tmtwo,\lenv)}{\ems} & \defeq & @\\
%	\end{array}\]
	\[\begin{array}{lcl@{\hspace{.9cm}}lll}
	\staddr{(\tm\tmtwo,\ctx,\lenv)}{(\addr,0\cons\addrtwo)} & \defeq 
	&\staddr{(\tm,\ctxp{\ctxhole\tmtwo},\lenv)}{(\addr\cons0,\addrtwo)} &
	 \staddr{(\tm\tmtwo,\lenv)}{(\addr,\ems)} & \defeq & @\\[3pt]
	\staddr{(\tm\tmtwo,\ctx,\lenv)}{(\addr,1\cons\addrtwo)} & \defeq 
	&\staddr{(\tmtwo,\ctxp{\tm\ctxhole},\lenv)}{(\addr\cons1,\addrtwo)} &
	\staddr{(\la\var\tm,\ctx,\lenv)}{(\addr,\ems)} & \defeq & \lambda \\[3pt] 
	\staddr{(\la\var\tm,\ctx,\lenv)}{(\addr,b\cons\addrtwo)} & \defeq 
	&\staddr{(\tm,\ctxp{\la\var\ctxhole},\lenv)}{(\addr\cons b,\addrtwo)}
	\\
	\end{array}\]
		\[\begin{array}{lcl@{\hspace{.9cm}}llc}
	\staddr{(\var,\ctx,\lenv)}{(\addr,\addrtwo)} & \defeq&  \begin{cases}
	\db(\var) & \text{ if }\var\not\in\dom\lenv \text{ and }\addrtwo=\ems,\\
	\staddr{(\tm,\ctxhole,\lenvtwo)}{(\addr,\addrtwo)} &  \text{ if }\var\in\dom\lenv \text{ and } \lenv(\var)=(\tm,\lenvtwo)\\
	\bot & \text{ if }\var\not\in\dom\lenv \text{ and }\addrtwo\neq\ems
	\end{cases} \end{array}\]
The algorithm moves the pointer inside a sub-term according to the address $\addrtwo$. The key case is the one of a variable $\var$, for which there are three possible outcomes: 
\begin{itemize}
\item \emph{Return}: if the address $\addrtwo$ is empty, the de Bruijn index of the variable is returned (the initial code is closed by hypothesis). Note that terms are not necessarily represented with de Bruin indices: it is a convenient representation, but also with other representations, one can usually compute the de Bruijn index of a variable occurrence in logarithmic space;
%%%
\item \emph{Jump}: if the address $\addrtwo$ is non-empty and $\var$ is bound by the environment, that is, $\var\in\dom\lenv$, then the algorithm moves to navigate the closure $(\tm, \lenv')$ associated to $\var$ by $\lenv$. Here the implicit closure pointer changes, and it is the only point of the algorithm where it changes;
%%%
\item \emph{Error}: if the address $\addrtwo$ is non-empty and $\var$ is not bound by the environment then it means that $\addr_0$ is an address that it is not compatible with the structure of the result term $\unf\clos$ of the computation. Thus the algorithm returns the undefined symbol $\bot$.
\end{itemize}
Note that the algorithm never needs to backtrack. The space complexity of the algorithm is easily verified to be $\bigo{\log\size{\addr}+\log\size{\tm}+\log\size\clos}$, because essentially one only needs to manipulate the term, address, and closure pointers (plus possibly a constant number of auxiliary pointers to implement the algorithm). Therefore, we obtained the following result.

\begin{prop}[Logarithmic space constructor equality]
\label{prop:lin-space-constr-eq}
Let $\tm$ be a closed $\l$-term, $\exec: \compil\tm 
\tospkam^{*} \state$ be a \SpKAM run,	and $\addr$ be a tree address. Then computing $(\unf \state)|_\addr$ has space complexity
	$\bigo{\log\size\addr + \log\size \state+ \log{\size\tm}}$.
\end{prop}
	
Then, from the two simulations (Theorem \ref{thm:main-simulation} 
and Proposition 
\ref{thm:converse-simulation}) and the logarithmic space constructor equality test (\refprop{lin-space-constr-eq}), our main result follows.
\begin{thm}[The \SpKAM is reasonable for space]\label{thm:main-thm-spkam}
\ccbn evaluation $\towh$ and the space of the \SpKAM provide a reasonable space cost model for the $\l$-calculus.
\end{thm}

\subsection{The \SpKAM is not Reasonable for Abstract Time.} 
\label{sect:unreas-time}
We complete our 
study of the \SpKAM by analyzing its time behavior. For abstract time (in our 
case, the number of \ccbn $\beta$ steps), the \SpKAM is unreasonable, because 
simulating \ccbn at times requires exponential overhead. The number of 
transitions of the \SpKAM is reasonable, while it is the cost of single transitions, 
thus of the manipulation of data structures, that can explode. The failure 
stems from the lack of data sharing, which on the other hand we showed being 
mandatory for space reasonability. Essentially, there are size exploding 
families such that their \SpKAM run produces environments of size exponential 
in the number of $\beta$ steps/transitions.
%, which is the key point in the proof 
%of the next proposition.
%
%\begin{prop}[\SpKAM abstract time overhead explosion]
%\label{prop:spkam-not-reas-for-nat-time}
%There is a family $\{\tm_{n}\}_{n\in\nat}$ of closed $\l$-terms such that its complete evaluation $\run_{n}:\tm_{n} \towh^{n} \tmtwo_{n}$ is simulated by \SpKAM runs $\runtwo_{n}$ taking both space and time exponential in $n$, that is, $\sizespace{\runtwo_{n}}=\sizetime{\runtwo_{n}}=\Omega(2^{n})$.
%\end{prop}

In order to prove this result we need to define the family of $\lambda$-terms $\{\tm_{n}\}_{n\in\nat}$ and prove two auxiliary lemmas. We first define the following data structures:
\[
\begin{array}{rcl@{\hspace{2cm}}rcl}
\lenv_0 & \defeq & \esub{\var_0}{(\Id,\stempty)} & \lenv_{n+1} & \defeq & 
\esub{\var_{n+1}}{\stack_n}\cdot\lenv_n
\\[4pt]
\stack_0 & \defeq & (\var_0\var_0,\lenv_0) & \stack_{n+1} & \defeq & 
(\var_0..\var_{n+1},\lenv_{n+1})
\end{array}
\]
Note that the size of $\lenv_{n}$ is exponential in $n$. 
\begin{lem}\label{l:expds}
	For each $n\in\nat$, $\size{\lenv_{n}}\geq2^n$.
\end{lem}
\begin{proof}
	Since 
	$\lenv_{n+1} \defeq \esub{\var_{n+1}}{\stack_n}\cdot\lenv_n= 
	\esub{\var_{n+1}}{(\var_0..\var_{n},\lenv_{n})}\cdot\lenv_n$, we have
	 $\size{\lenv_{n+1}}\geq 2\size{\lenv_{n}}$ and thus $\size{\lenv_{n}}\geq2^n$.
\end{proof}
Now, we define the family of contexts $\{\ctx_n\}_{n\in\nat}$ as follows:
\[
\begin{array}{rcl}
\ctx_0 &\defeq& \lambda \var_0.\ctxhole(\var_0\var_0)
\\
\ctx_{n+1} &\defeq& 
\lambda\var_{n+1}.\ctxhole(\var_0\ldots\var_{n+1})
\end{array}
\]
The execution of $\ctx_n\ctxholep\tm$  generates the stacks $\stack_n$ and environments $\lenv_n$ defined above, provided that $\var_0,\ldots,\var_n$ appear free in $\tm$.
\begin{lem}\label{l:counter-tr}
	For each $\lambda$-term $\tm$, if the variables $\var_0,\ldots,\var_n$ appear free in $\tm$, then \\
	$\spkamstate{\ctx_0\ctxholep{\ctx_1\ctxholep{\ldots\ctx_{n} 
				\ctxholep{\tm}\ldots}}\Id} \stempty\stempty\rightarrow_{\mathrm\SpKAM}^{\Theta(n)} 
	\spkamstate{\tm}{\lenv_n}{\stack_n}$.
\end{lem}
\begin{proof}
	We proceed by simply executing the \SpKAM.
	\begin{itemize}
	\item Case $n=0$.
	\[\begin{array}{l|l|ll}
	\mathsf{Term}   & 
	\mathsf{Env} & \mathsf{Stack}\\
	\hhline{---}
	\ctx_0\ctxholep\tm\Id & \stempty & 
	\stempty 
	& \tokamseanv
	\\
	\ctx_0\ctxholep\tm\defeq\lambda \var_0.\tm(\var_0\var_0) & \stempty & 
	(\Id,\stempty)
	& \tokambnw
	\\
	\tm(\var_0\var_0) & \esub{\var_0}{(\Id,\stempty)} & \stempty
	& \tokamseanv
	\\
	\tm & \esub{\var_0}{(\Id,\stempty)} & 
	(\var_0\var_0,\esub{\var_0}{(\Id,\stempty)})\\
	\end{array}\]
	
	\item Case $n\geq 1$.
	We observe that $\ctx_0\ctxholep{\ctx_1\ctxholep{\ldots\ctx_{n} 
			\ctxholep{\ctx_{n+1}\ctxholep\tm}\ldots}}\Id$ can be rewritten to\\ 
	$\ctx_0\ctxholep{\ctx_1\ctxholep{\ldots\ctx_{n} 
			\ctxholep{\tmtwo}\ldots}}\Id$, where 
	$\tmtwo\defeq\ctx_{n+1}\ctxholep\tm$. Of 
	course $\var_0,\ldots,\var_n$ appear free in $\tmtwo$. We can thus immediately apply the \ih
	\[\begin{array}{l|l|ll}
	\mathsf{Term}   & 
	\mathsf{Env} & \mathsf{Stack}\\
	\hhline{---}
	\ctx_0\ctxholep{\ctx_1\ctxholep{..\ctx_{n} 
			\ctxholep{\ctx_{n+1}\ctxholep\tm}..}}\Id & \stempty & 
	\stempty & \rightarrow_{\mathrm\SpKAM}^{\Theta(n)}\ (\ih)\\
	\underbrace{\lambda \var_{n+1}.\tm(\var_0..\var_{n+1}) }_{\ctx_{n+1}\ctxholep\tm}
	& \lenv_n & \stack_n
	&\tokambnw\\
	\tm(\var_0..\var_{n+1}) & \underbrace{\esub{\var_{n+1}}{\stack_n} \cons 
	\lenv_n}_{=:\lenv_{n+1}} & 
	\stempty
	&\tokamseanv\\
	\tm & \lenv_{n+1} & (\var_0..\var_{n+1},\lenv_{n+1})=:\stack_{n+1} \\
	\end{array}\qedhere\]

\end{itemize}	
\end{proof}
Finally, we prove that the \SpKAM is not reasonable with respect to abstract time.
\begin{prop}[\SpKAM abstract time overhead explosion]
\label{prop:spkam-not-reas-for-nat-time}
	There is a family $\{\tm_{n}\}_{n\in\nat}$ of closed $\l$-terms such that there is a complete evaluation $\run_{n}:\tm_{n} \towh^{n} \tmtwo_{n}$ simulated by \SpKAM runs $\runtwo_{n}$ taking both space and time exponential in $n$, that is, $\sizespace{\runtwo_{n}}=\sizetime{\runtwo_{n}}=\Omega(2^{n})$.
\end{prop}
\begin{proof}
	Define $\tm_n$ as $\tm_{n} \defeq \ctx_0\ctxholep{\ctx_1\ctxholep{\cdots\ctx_{n} 
			\ctxholep{\la{\vartwo} \Id}\cdots}}\Id$.	 Its \SpKAM execution follows.
	\[\begin{array}{l|l|ll}
	\mathsf{Term}   & 
	\mathsf{Env} & \mathsf{Stack}\\
	\hhline{---}	
	\ctx_0\ctxholep{\ctx_1\ctxholep{\cdots\ctx_{n-1}\ctxholep{\ctx_{n} 
				\ctxholep{\la{\vartwo} \Id}}\cdots}}\Id & \stempty & 
	\stempty & \rightarrow_{\mathrm\SpKAM}^{\Theta(n)}\   (\text{Lemma~\ref{l:counter-tr}})\\
	\underbrace{\lambda\var_{n}.(\la{\vartwo} \Id)(\var_0\ldots\var_{n})}_{=\ctx_n\ctxholep{\la{\vartwo} \Id}} & \lenv_{n-1} 
	& \stack_{n-1}
	&\tokambnw
	\\
	(\la{\vartwo} \Id)(\var_0\ldots\var_{n}) & \lenv_n & \stempty
	&\tokamseanv
	\\
	\la{\vartwo} \Id& \stempty & \stack_n
	&\tokambw
	\\
	\Id & \stempty & \stempty
	\end{array}
	\]
	The space consumed (and thus also the low-level time) is at least exponential in $n$ 
	because the size of $\lenv_n$ is exponential in $n$ (Lemma~\ref{l:expds}).
\end{proof}
Since the \NaKAM is less efficient than the \SpKAM, an analogous reasoning also shows that the \NaKAM is unreasonable with respect to abstract time.

\begin{prop}[\NaKAM abstract time overhead explosion]
\label{prop:NaKAM-time-unreas}
	There is a family $\{\tm_{n}\}_{n\in\nat}$ of closed $\l$-terms such that there is a complete evaluation $\run_{n}:\tm_{n} \towh^{n} \tmtwo_{n}$ simulated by \NaKAM runs $\runtwo_{n}$ taking both space and time at least exponential in $n$, that is, $\sizespace{\runtwo_{n}}=\sizetime{\runtwo_{n}}=\Omega(2^{n})$.
\end{prop}

% !TeX spellcheck = en_US
% !TEX root = main.tex
%%%%%%%%%%%%%%

\pagebreak
\section{Time vs Space}
\label{sect:time-spkam}
Here we discuss how to obtain, or approximate, reasonability for both space and time.

\paragraph{Reasonable Low-Level Time.} One way of recovering time reasonability is changing the time cost model from the number of $\beta$ steps to the time taken by the \SpKAM itself, which is a low-level notion of time. Such a time cost model is indeed reasonable. The key point is that the explosions of \refprop{spkam-not-reas-for-nat-time} never happen on $\l$-terms encoding TMs.

\begin{thm}[TMs are simulated by the \SpKAM in reasonable low-level time]
\label{thm:reas-time-spkam}
\hfill
\begin{enumerate}
\item Every TM run $\run$ can be simulated by the \SpKAM in time 
$\bigo{poly(\size\run)}$.
\item Every \SpKAM run $\exec: \compil\tm \tospkam^{*} \state$  can be 
implemented on TMs in time $\bigo{\sizetime{\run}}$.
\item \ccbn and the time of the \SpKAM provide a reasonable time cost 
model for the $\l$-calculus.
\end{enumerate}
\end{thm}
\begin{proof}
	The first point is the only one which is non-trivial. We have already 
	proved that the Space KAM can simulate TMs runs $\run$ in a number of 
	\emph{transitions} which is polynomial in $\size\run$. However, this does 
	not 
	necessarily means that the (low-level) time is also polynomial in $\run$, 
	see Proposition~\ref{prop:spkam-not-reas-for-nat-time}. About the execution 
	of terms which are the image of the encoding of TMs into the $\l$-calculus, 
	we can say however that the overhead stays polynomial. Indeed, the 
	exponential blowup comes from the fact that environments are duplicated in 
	an uncontrolled way. This does not happen in the execution of the encoding 
	of TMs, where duplication is restrained to the fix-point operator and to 
	the input components of the state. In other words, we duplicate only 
	objects of fixed size, thus confirming the polynomial bound.
	
	There is another, indirect, way of proving the same results. If the 
	(low-level) time were exponential in $\size\run$, then the space should be 
	at least linear in $\size\run$\footnote{This is because space cannot be 
		less than logarithmic in time, when space and time are locked and reasonable.}. But we have proved that this is not the 
	case since space is linearly related with the \emph{space} consumption of 
	$\run$, and \emph{not} with its length (which is the time consumption).
\end{proof}
The drawback of this solution is that one gives up the natural cost model for 
time. Moreover, the low-level time of the \SpKAM can be very lax in comparison, as 
\refprop{spkam-not-reas-for-nat-time} shows.

%\begin{figure*}[t]
%	\begin{center}
%		\fbox{
%			\input{machines/SharingKAM}
%		}
%	\end{center}
%	\vspace{-8pt}
%	\caption{Data structures, transitions, and abstract implementation of the 
%		\TiKAM.}
%	\label{fig:shkam}
%\end{figure*}

\paragraph{KAM and Sharing of Environments} The \LinkKAM of \reffig{LinkKAM} (page \pageref{fig:LinkKAM}) adopts sharing of environments, which is the most common way of turning the \OutKAM into a machine reasonable for abstract time. The drawback is that such an implementation schema does not work with the space-oriented optimizations of the \CoKAM, at least not smoothly. The culprit is eager garbage collection (there are no problems instead for unchaining), because, when an environment $\lenv$ is shared among many closures, knowing that an entry $\esub\var\clos$ of $\lenv$ is garbage for a closure $(\tm,\lenv)$ is not enough to remove $\esub\var\clos$ from $\lenv$, since $\esub\var\clos$ might not be garbage for another closure using $\lenv$.

Without eager garbage collection, the space consumption of the \LinkKAM depends linearly on time, which is unreasonable for space. Therefore, sharing environments enables reasonability for abstract time at the expenses of space reasonability.

An interesting point is that the counter-example to time efficiency of the \SpKAM in  \refprop{spkam-not-reas-for-nat-time} is executed on the \LinkKAM in exponentially less time \emph{and} space than on the \SpKAM, thanks to the sharing of environments. This fact stresses that space reasonable does not mean space efficient: the \SpKAM is efficient on the encoding of TMs, but beyond the image of the encoding it can be terribly inefficient.

\paragraph{The Interleaving Technique} Forster et al. in 
\cite{cbv_reasonable} show that, given one machine that is reasonable for abstract time 
but not for ink space and one machine that is reasonable for ink space but not for 
abstract time, it is possible to build a third machine that is reasonable for both ink space 
and abstract time---even if the two are explosive together---by interleaving the two machines in a smart way. Despite being 
presented on a specific case, their construction is quite general (in fact it 
is not even limited to the $\l$-calculus), and can be adapted to our case replacing ink space with the space of the \SpKAM (the 
two starting machines being the \LinkKAM and the \SpKAM), under the assumption that the two machines share the same input, which is essential for logarithmic space\footnote{The results in Forster et al. 
\cite{cbv_reasonable} hold for decision problems (where the output is either \emph{yes} or \emph{no}), instead of computation problems (where the output can be any value) as in this paper, and are also given using fixed simulations. Their technique however is flexible and fairly independent from the notion of problem and also from the specific simulations, which are used as black boxes. The requirements for the technique are very lax, essentially that 1) the logarithm of time is linear in space (which is a fact true for most choices of cost models), and 2) the simulations should be runnable 'in rounds' (see \cite{cbv_reasonable}).}. The drawback of this 
solution is that it admits space exponential in time, as 
\refprop{spkam-not-reas-for-nat-time} shows. %\red{Additionally, the interleaving machine is a smart theoretical tool but not a new practical way of executing $\l$-terms. Nemmeno la nostra, in realtà.}

\paragraph{Trading Time for Space} From a practical rather than theoretical 
point of view, there is a further \emph{semi solution} that we now outline. 
The idea is to modify the \LinkKAM as to share \emph{closures} rather then environments, thus copying environments when crossing applications, but copying only their shallow structure, because for closures one would just copy pointers to them. Let us refer to this schema as to the \emph{Closure KAM}. The Closure KAM is reasonable for abstract time (even if slightly slower than the \LinkKAM) and---crucially---is compatible with the \CoKAM. Compared with the \LinkKAM, it thus has the advantage of disentangling space from time. It is not space reasonable, because of data pointers for sharing closures, that add an unreasonable pointer overhead which is logarithmic in the closure space of the \CoKAM. But such an overhead is only \emph{mildly} unreasonable. 

The Closure KAM is probably the best compromise between reasonability and efficiency for the practice of implementing functional programs. We plan to study it in future work.

 %\ben{Note 
%also that, since implementations of functional languages all rely on data 
%pointers (which add an unreasonable logarithmic pointer overhead for space) 
%but 
%implement GC (which disentangles the pointer overhead from time), they are all 
%space unreasonable and yet not too much so.}
% !TeX spellcheck = en_US
% !TEX root = main.tex
%%%%%%%%%%%%%%
\section{Call-by-Value and Other Strategies}
\label{sect:cbv}

\begin{figure*}[!t]
	\input{machines/SpaceLAM}
	\vspace{-8pt}
	\caption{The Collecting \LAM, the abstract layer of the \SpLAM.}
	\vspace{-8pt}
	\label{fig:splam}
\end{figure*}

How robust is our space cost model to changes of the evaluation strategy? The 
short answer is \emph{very robust}.

\paragraph{Closed Call-by-Value}
Our results smoothly adapt to weak call-by-value evaluation with closed terms, which we refer  to as \emph{\ccbv} and define as follows. Values and right-to-left \cbv evaluation contexts are given by:
\[\begin{array}{rrcl}
	\textsc{Values} & \val & \grameq & \lambda x.\tm\\[3pt]
	\textsc{Right-to-Left CbV Ctxs} & \evctx & \grameq & \ctxhole \midd \evctx\val \midd \tm\evctx
	\end{array}\]
The (deterministic) reduction strategy $\toval$ is defined as the contextual closure of the $\beta_\vsym$ variant of the $\beta$ rule
\[(\la\var\tm)\val\mapsto_{\beta_\vsym}\tm\isub{\var}\val\]
by right-to-left CbV evaluation contexts.

Our results smoothly adapt to such a setting, as we now explain.
First, it is easy to adapt the \SpKAM to \ccbv. The LAM (Leroy Abstract Machine) is a right-to-left\footnote{The argument presented here smoothly adapts to the left-to-right order.} \cbv  analogue of the KAM defined by Accattoli et al. in \cite{DBLP:conf/icfp/AccattoliBM14} and modeled after Leroy's ZINC \cite{Leroy-ZINC} (whence the name). It uses a further data structure, the \emph{dump}, storing the left sub-terms of applications yet to be evaluated. It is upgraded to the Collecting \LAM in \reffig{splam} by removing data pointers and adding eager GC.  Lastly, it gives the \SpLAM by adopting an abstract implementation analogous to that of the \SpKAM. Unchaining comes for free in \cbv, if one considers values to be only abstractions, see Accattoli and Sacerdoti Coen \cite{DBLP:journals/iandc/AccattoliC17}.

The next step is realizing that, because of the mentioned \emph{indifference property} of the deterministic $\l$-calculus $\detLam$ (containing the image of the encoding of TMs), the run of 
the \SpLAM on a term $\tm\in\detLam$ is almost identical (technically, weakly 
bisimilar) to the one of the \SpKAM on $\tm$.

\begin{prop}\label{prop:cbv-equiv}
	The Space KAM and the Space LAM are weakly bisimilar when executed on 
	$\detLam$-terms. Moreover, their space consumption is the same.
\end{prop}
 \begin{proof}
	The transitions of the \SpKAM not dealing with applications 
	are 
	identical to the corresponding ones of the \SpLAM (if one ignores the 
	dump, 
	that 
	remains untouched). For the two transitions of the \SpKAM dealing with 
	applications, we show that, when the argument is a variable or an 
	abstraction 
	(as in $\detLam$), the \SpLAM behaves as the \SpKAM. If the active term is 
	$\tm\var$, indeed, the $\tokamseav$ transition of the \SpKAM is simulated 
	on 
	the \SpLAM by (with $\lenv(\var)=(\la\vartwo\tmtwo,\lenvtwo)$):
	\[\begin{array}{lll}
		\splamstate{\stempty}{\tm\var}{\lenv}{\stack}&\tosplam& 
		\splamstate{\dentry{(\tm,\lenv|_\tm)}\stack}{\var}{\lenv|_\var}{\stempty}\\
		&\tosplam&\splamstate{\dentry{(\tm,\lenv|_\tm)}\stack}{\la\vartwo\tmtwo}{\lenvtwo}{\stempty}\\
		&\tosplam&\splamstate{\stempty}{\tm}{\lenv|_\tm}{(\la\vartwo\tmtwo,\lenvtwo)\cons\stack}\\
		&=&\splamstate{\stempty}{\tm}{\lenv|_\tm}{\lenv(\var)\cons\stack}
		\end{array}\]
	If the active term instead is 
	$\tm(\la\var\tmtwo)$, the $\tokamseanv$ transition of the \SpKAM is 
	simulated 
	on the \SpLAM by:
	\[\begin{array}{lll}
		\splamstate{\stempty}{\tm(\la\var\tmtwo)}{\lenv}{\stack}&\tosplam& 
		\splamstate{\dentry{(\tm,\lenv|_\tm)}\stack}{\la\var\tmtwo}{\lenv|_{\la\var\tmtwo}}{\stempty}\\
		&\tosplam&\splamstate{\stempty}{\tm}{\lenv|_\tm}{(\la\var\tmtwo,\lenv|_{\la\var\tmtwo})\cdot\stack}\\
		\end{array}\]
	In particular, these macro steps show that to evaluate TMs there is no need 
	of 
	the dump.
	Now, by defining a relation $\rel$ between states of the \SpKAM and the 
	\SpLAM 
	as
	\[
		\state_K\,\rel\,\state_L\qquad\text{ iff }\qquad 
		\state_K=\spkamstate{\tm}{\lenv}{\stack}\text{ and } 
		\state_L=\splamstate{\stempty}{\tm}{\lenv}{\stack}
		\]
	the previous reasoning shows that $\rel$ is 
	a weak bisimulation preserving time and space complexity (modulo a constant 
	overhead).
\end{proof}

 Since the simulation of the 
\SpLAM on TMs is as smooth as for the \SpKAM, we have the following 
result.

\begin{thm}[The \SpLAM is reasonable for space]\label{thm:cbv}
	\ccbv evaluation and the space of the \SpLAM provide a reasonable space 
	cost model for the $\l$-calculus.
\end{thm}

\paragraph{Open and Strong Evaluation} Extending \cbn/\cbv evaluation to deal 
with open terms or even under abstractions, which is notoriously very delicate 
in the study of reasonable time, is instead straightforward for space. This is 
because these extensions play no role in the simulation of TMs, which is the 
delicate direction for space. Given the absence of difficulties, we refrain 
from introducing variants of the \SpKAM/LAM for open and strong evaluation.

\paragraph{Call-by-Need} The only major scheme for which our technique breaks is call-by-need (\cbneed) evaluation. To our knowledge, implementations of \cbneed inevitably rely on a heap and on data pointers similar to those of the \LinkKAM, to realize the memoization mechanism at the heart of \cbneed. Therefore, they are space unreasonable. This is not really surprising: being a time optimization of \cbn, \cbneed trades space for time, sacrificing space reasonability.
\begin{figure*}[!t]
\small
	\renewcommand{\arraystretch}{1.2}
	\centering
	\begin{tabular}{|c|c|c|c|ccc}
	%{17cm}{|p{3cm}|>{\centering} p{4.8cm}|>{\centering} 
			%p{4.6cm}|>{\centering\arraybackslash}p{3cm}|}
		\cline{1-4}
		& \textbf{Abstract Time}  
		& 
		\textbf{Low-Level Time}   
		& 
		\textbf{Low-Level Space} 
		\\
		& \textbf{Reasonable}  
		& 
		\textbf{Reasonable}   
		& 
		\textbf{Reasonable} 
		\\
			& (polyn. in \# of $\beta$)
			 &
			 (actual implementation cost)
			 &

			 \\
		\cline{1-4}
		\textbf\NaKAM & No, Proposition~\ref{prop:spkam-not-reas-for-nat-time} 
		& No, Proposition~\ref{prop:toy}.\ref{point:nakam} & No, 
		Proposition~\ref{prop:toy}.\ref{point:nakam}\\
		\cline{1-4}
		\textbf\SpKAM & No, Proposition~\ref{prop:spkam-not-reas-for-nat-time} 
		& Yes, Theorem~\ref{thm:reas-time-spkam} & 
		Yes, 
		Theorem~\ref{thm:main-thm-spkam}\\
		\cline{1-4}
		\textbf\LinkKAM & Yes, see Section~\ref{sect:naive-kam} & Yes, 
		see Section~\ref{sect:naive-kam} & 
		No, 
		see Section~\ref{sect:naive-kam}\\
		\cline{1-4}
		\textbf\SpLAM & No, via Proposition~\ref{prop:cbv-equiv}  & Yes, via Proposition~\ref{prop:cbv-equiv}  & 
		Yes, Theorem~\ref{thm:cbv}
		\\
		(CbV)
		&and Proposition~\ref{prop:spkam-not-reas-for-nat-time}
		&and Theorem.~\ref{thm:reas-time-spkam}
		&
		\\
		\cline{1-4}
	\end{tabular}
%	\vspace{-8pt}
	\caption{Summary of the results of the paper.}
%	\vspace{-8pt}
	\label{fig:summary}
\end{figure*}
% !TeX spellcheck = en_US
% !TEX root = main.tex
%%%%%%%%%%%%%%
\section{Conclusions}
Via a fine study of abstract machines and of the encoding of Turing machines, 
we  provide the first space cost model for the $\l$-calculus accounting for 
logarithmic space. We have reported our main results in \reffig{summary}.

Our cost model is given by an external device, the 700th abstract 
machine for the $\l$-calculus, so how 
canonical is it? The constraints 
for reasonable logarithmic space are \emph{very} strict. It seems that there is no 
room for significant variations in the machine nor in the encoding of TMs. Moreover, 
our cost model for space has the same relationship to abstract time than ink space (that is, it is explosive, as shown by \refprop{spkam-not-reas-for-nat-time}), and it 
smoothly adapts to other evaluation strategies, such as call-by-value. We then 
dare to say that our space cost model is fairly canonical.

We have also isolated an abstract notion of \emph{closure space}, given by the maximum number of closures used by the \SpKAM (without accounting for the size of sub-term pointers in the closures). In a companion paper~\cite{ICFP2022}, we have given a machine independent characterization of closure space based on multi types, providing evidence of the naturality of our cost model.

\medskip

\section*{Acknowledgment}
\noindent This work has been inspired by an old talk by Kazushige Terui on the space efficiency of the KAM~\cite{TeruiSlides}. The second author is partially supported by the ERC CoG ``DIAPASoN'' (GA 818616). The third author is partially supported by the ANR project ``PPS'' (ANR-19-CE48-0014) and by the European Union’s Horizon 2020 research and innovation programme
under the Marie Sklodowska-Curie grant agreement No 101034255.

  %% the following bibliography is gererated manually for the sake of brevity
  %% only; please use a separate .bib file in your submission

\bibliographystyle{alphaurl}
\bibliography{main}

\begin{thebibliography}{ADLV22b}

\bibitem[AB17]{DBLP:conf/ppdp/AccattoliB17}
Beniamino Accattoli and Bruno Barras.
\newblock Environments and the complexity of abstract machines.
\newblock In Wim Vanhoof and Brigitte Pientka, editors, {\em Proceedings of the
  19th International Symposium on Principles and Practice of Declarative
  Programming, Namur, Belgium, October 09 - 11, 2017}, pages 4--16. {ACM},
  2017.
\newblock \href {https://doi.org/10.1145/3131851.3131855}
  {\path{doi:10.1145/3131851.3131855}}.

\bibitem[ABM14]{DBLP:conf/icfp/AccattoliBM14}
Beniamino Accattoli, Pablo Barenbaum, and Damiano Mazza.
\newblock Distilling abstract machines.
\newblock In Johan Jeuring and Manuel M.~T. Chakravarty, editors, {\em
  Proceedings of the 19th {ACM} {SIGPLAN} international conference on
  Functional progranitarmming, Gothenburg, Sweden, September 1-3, 2014}, pages
  363--376. {ACM}, 2014.
\newblock \href {https://doi.org/10.1145/2628136.2628154}
  {\path{doi:10.1145/2628136.2628154}}.

\bibitem[Acc17]{DBLP:journals/entcs/Accattoli18}
Beniamino Accattoli.
\newblock ({I}n){E}fficiency and {R}easonable {C}ost {M}odels.
\newblock In {\em 12th Workshop on Logical and Semantic Frameworks, with
  Applications, {LSFA} 2017, Bras{\'{\i}}lia, Brazil, September 23-24, 2017},
  volume 338 of {\em Electronic Notes in Theoretical Computer Science}, pages
  23--43. Elsevier, 2017.
\newblock \href {https://doi.org/10.1016/j.entcs.2018.10.003}
  {\path{doi:10.1016/j.entcs.2018.10.003}}.

\bibitem[ACC21]{DBLP:conf/lics/AccattoliCC21}
Beniamino Accattoli, Andrea Condoluci, and Claudio~Sacerdoti Coen.
\newblock Strong call-by-value is reasonable, implosively.
\newblock In {\em 36th Annual {ACM/IEEE} Symposium on Logic in Computer
  Science, {LICS} 2021, Rome, Italy, June 29 - July 2, 2021}, pages 1--14.
  {IEEE}, 2021.
\newblock \href {https://doi.org/10.1109/LICS52264.2021.9470630}
  {\path{doi:10.1109/LICS52264.2021.9470630}}.

\bibitem[Acc23]{DBLP:journals/lmcs/Accattoli23}
Beniamino Accattoli.
\newblock Exponentials as substitutions and the cost of cut elimination in
  linear logic.
\newblock {\em Log. Methods Comput. Sci.}, 19(4), 2023.
\newblock \href {https://doi.org/10.46298/LMCS-19(4:23)2023}
  {\path{doi:10.46298/LMCS-19(4:23)2023}}.

\bibitem[ADL12]{DBLP:conf/rta/AccattoliL12}
Beniamino Accattoli and Ugo Dal~Lago.
\newblock On the invariance of the unitary cost model for head reduction.
\newblock In Ashish Tiwari, editor, {\em 23rd International Conference on
  Rewriting Techniques and Applications (RTA'12) , {RTA} 2012, May 28 - June 2,
  2012, Nagoya, Japan}, volume~15 of {\em LIPIcs}, pages 22--37. Schloss
  Dagstuhl - Leibniz-Zentrum f{\"{u}}r Informatik, 2012.
\newblock \href {https://doi.org/10.4230/LIPIcs.RTA.2012.22}
  {\path{doi:10.4230/LIPIcs.RTA.2012.22}}.

\bibitem[ADL16]{accattoli_leftmost-outermost_2016}
Beniamino Accattoli and Ugo Dal~Lago.
\newblock ({Leftmost}-{Outermost}) {Beta} {Reduction} is {Invariant}, {Indeed}.
\newblock {\em Logical Methods in Computer Science}, 12(1), 2016.
\newblock \href {https://doi.org/10.2168/LMCS-12(1:4)2016}
  {\path{doi:10.2168/LMCS-12(1:4)2016}}.

\bibitem[ADLV20]{DBLP:conf/ppdp/AccattoliLV20}
Beniamino Accattoli, Ugo Dal~Lago, and Gabriele Vanoni.
\newblock The machinery of interaction.
\newblock In {\em {PPDP} '20: 22nd International Symposium on Principles and
  Practice of Declarative Programming, Bologna, Italy, 9-10 September, 2020},
  pages 4:1--4:15. {ACM}, 2020.
\newblock \href {https://doi.org/10.1145/3414080.3414108}
  {\path{doi:10.1145/3414080.3414108}}.

\bibitem[ADLV21a]{ADLVPOPL21}
Beniamino Accattoli, Ugo Dal~Lago, and Gabriele Vanoni.
\newblock The (in)efficiency of interaction.
\newblock {\em Proc. {ACM} Program. Lang.}, 5({POPL}):1--33, 2021.
\newblock \href {https://doi.org/10.1145/3434332} {\path{doi:10.1145/3434332}}.

\bibitem[ADLV21b]{DBLP:conf/lics/AccattoliLV21}
Beniamino Accattoli, Ugo Dal~Lago, and Gabriele Vanoni.
\newblock The space of interaction.
\newblock In {\em 2021 36th Annual ACM/IEEE Symposium on Logic in Computer
  Science (LICS)}, pages 1--13, 2021.
\newblock \href {https://doi.org/10.1109/LICS52264.2021.9470726}
  {\path{doi:10.1109/LICS52264.2021.9470726}}.

\bibitem[ADLV22a]{ICFP2022}
Beniamino Accattoli, Ugo Dal~Lago, and Gabriele Vanoni.
\newblock Multi types and reasonable space.
\newblock {\em Proc. {ACM} Program. Lang.}, 6({ICFP}):799--825, 2022.
\newblock \href {https://doi.org/10.1145/3547650} {\path{doi:10.1145/3547650}}.

\bibitem[ADLV22b]{DBLP:conf/lics/AccattoliLV22}
Beniamino Accattoli, Ugo Dal~Lago, and Gabriele Vanoni.
\newblock Reasonable space for the {\(\lambda\)}-calculus, logarithmically.
\newblock In Christel Baier and Dana Fisman, editors, {\em {LICS} '22: 37th
  Annual {ACM/IEEE} Symposium on Logic in Computer Science, Haifa, Israel,
  August 2 - 5, 2022}, pages 47:1--47:13. {ACM}, 2022.
\newblock \href {https://doi.org/10.1145/3531130.3533362}
  {\path{doi:10.1145/3531130.3533362}}.

\bibitem[ADLV23]{log-encoding}
Beniamino Accattoli, Ugo Dal~Lago, and Gabriele Vanoni.
\newblock A log-sensitive encoding of turing machines in the
  $\lambda$-calculus.
\newblock 2023.
\newblock \href {https://doi.org/10.48550/arXiv.2301.12556}
  {\path{doi:10.48550/arXiv.2301.12556}}.

\bibitem[AJM00]{DBLP:journals/iandc/AbramskyJM00}
Samson Abramsky, Radha Jagadeesan, and Pasquale Malacaria.
\newblock Full abstraction for {PCF}.
\newblock {\em Inf. Comput.}, 163(2):409--470, 2000.
\newblock \href {https://doi.org/10.1006/inco.2000.2930}
  {\path{doi:10.1006/inco.2000.2930}}.

\bibitem[ASC17]{DBLP:journals/iandc/AccattoliC17}
Beniamino Accattoli and Claudio Sacerdoti~Coen.
\newblock On the value of variables.
\newblock {\em Inf. Comput.}, 255:224--242, 2017.
\newblock \href {https://doi.org/10.1016/j.ic.2017.01.003}
  {\path{doi:10.1016/j.ic.2017.01.003}}.

\bibitem[BCD21]{DBLP:conf/ppdp/BiernackaCD21}
Malgorzata Biernacka, Witold Charatonik, and Tomasz Drab.
\newblock A derived reasonable abstract machine for strong call by value.
\newblock In Niccol{\`{o}} Veltri, Nick Benton, and Silvia Ghilezan, editors,
  {\em {PPDP} 2021: 23rd International Symposium on Principles and Practice of
  Declarative Programming, Tallinn, Estonia, September 6-8, 2021}, pages
  6:1--6:14. {ACM}, 2021.
\newblock \href {https://doi.org/10.1145/3479394.3479401}
  {\path{doi:10.1145/3479394.3479401}}.

\bibitem[BCD22]{DBLP:journals/pacmpl/BiernackaCD22}
Malgorzata Biernacka, Witold Charatonik, and Tomasz Drab.
\newblock A simple and efficient implementation of strong call by need by an
  abstract machine.
\newblock {\em Proc. {ACM} Program. Lang.}, 6({ICFP}):109--136, 2022.
\newblock \href {https://doi.org/10.1145/3549822} {\path{doi:10.1145/3549822}}.

\bibitem[BG95]{DBLP:conf/fpca/BlellochG95}
Guy~E. Blelloch and John Greiner.
\newblock Parallelism in sequential functional languages.
\newblock In {\em Proceedings of the seventh international conference on
  Functional programming languages and computer architecture, {FPCA} 1995, La
  Jolla, California, USA, June 25-28, 1995}, pages 226--237. {ACM}, 1995.
\newblock \href {https://doi.org/10.1145/224164.224210}
  {\path{doi:10.1145/224164.224210}}.

\bibitem[BG96]{DBLP:conf/icfp/BlellochG96}
Guy~E. Blelloch and John Greiner.
\newblock A provable time and space efficient implementation of {NESL}.
\newblock In Robert Harper and Richard~L. Wexelblat, editors, {\em Proceedings
  of the 1996 {ACM} {SIGPLAN} International Conference on Functional
  Programming, {ICFP} 1996, Philadelphia, Pennsylvania, USA, May 24-26, 1996},
  pages 213--225. {ACM}, 1996.
\newblock \href {https://doi.org/10.1145/232627.232650}
  {\path{doi:10.1145/232627.232650}}.

\bibitem[CAC19]{DBLP:conf/ppdp/CondoluciAC19}
Andrea Condoluci, Beniamino Accattoli, and Claudio~Sacerdoti Coen.
\newblock Sharing equality is linear.
\newblock In Ekaterina Komendantskaya, editor, {\em Proceedings of the 21st
  International Symposium on Principles and Practice of Programming Languages,
  {PPDP} 2019, Porto, Portugal, October 7-9, 2019}, pages 9:1--9:14. {ACM},
  2019.
\newblock \href {https://doi.org/10.1145/3354166.3354174}
  {\path{doi:10.1145/3354166.3354174}}.

\bibitem[DF07]{DBLP:journals/lisp/DouenceF07}
R{\'{e}}mi Douence and Pascal Fradet.
\newblock The next 700 krivine machines.
\newblock {\em High. Order Symb. Comput.}, 20(3):237--255, 2007.
\newblock \href {https://doi.org/10.1007/s10990-007-9016-y}
  {\path{doi:10.1007/s10990-007-9016-y}}.

\bibitem[DLA17]{DBLP:journals/corr/abs-1711-10078}
Ugo Dal~Lago and Beniamino Accattoli.
\newblock Encoding turing machines into the deterministic lambda-calculus.
\newblock {\em CoRR}, abs/1711.10078, 2017.
\newblock URL: \url{http://arxiv.org/abs/1711.10078}, \href
  {http://arxiv.org/abs/1711.10078} {\path{arXiv:1711.10078}}.

\bibitem[DLS10]{DBLP:conf/esop/LagoS10}
Ugo Dal~Lago and Ulrich Sch{\"{o}}pp.
\newblock Functional programming in sublinear space.
\newblock In Andrew~D. Gordon, editor, {\em 19th European Symposium on
  Programming, {ESOP} 2010, Paphos, Cyprus, March 20-28, 2010, Proceedings.},
  volume 6012 of {\em Lecture Notes in Computer Science}, pages 205--225.
  Springer, 2010.
\newblock \href {https://doi.org/10.1007/978-3-642-11957-6\_12}
  {\path{doi:10.1007/978-3-642-11957-6\_12}}.

\bibitem[DLS16]{dal_lago_computation_2016}
Ugo Dal~Lago and Ulrich Schöpp.
\newblock Computation by interaction for space-bounded functional programming.
\newblock {\em Information and Computation}, 248:150--194, 2016.
\newblock \href {https://doi.org/10.1016/j.ic.2015.04.006}
  {\path{doi:10.1016/j.ic.2015.04.006}}.

\bibitem[DR95]{danos_regnier_1995}
Vincent Danos and Laurent Regnier.
\newblock Proof-nets and the hilbert space.
\newblock In {\em Proceedings of the Workshop on Advances in Linear Logic},
  pages 307--328, USA, 1995. Cambridge University Press.
\newblock \href {https://doi.org/10.1017/CBO9780511629150.016}
  {\path{doi:10.1017/CBO9780511629150.016}}.

\bibitem[FGSW07]{DBLP:journals/lisp/FriedmanGSW07}
Daniel~P. Friedman, Abdulaziz Ghuloum, Jeremy~G. Siek, and Onnie~Lynn
  Winebarger.
\newblock Improving the lazy krivine machine.
\newblock {\em High. Order Symb. Comput.}, 20(3):271--293, 2007.
\newblock \href {https://doi.org/10.1007/s10990-007-9014-0}
  {\path{doi:10.1007/s10990-007-9014-0}}.

\bibitem[FKR20]{cbv_reasonable}
Yannick Forster, Fabian Kunze, and Marc Roth.
\newblock The weak call-by-value {\(\lambda\)}-calculus is reasonable for both
  time and space.
\newblock {\em Proc. {ACM} Program. Lang.}, 4({POPL}):27:1--27:23, 2020.
\newblock \href {https://doi.org/10.1145/3371095} {\path{doi:10.1145/3371095}}.

\bibitem[FS08]{DBLP:journals/entcs/FernandezS09}
Maribel Fern{\'{a}}ndez and Nikolaos Siafakas.
\newblock New developments in environment machines.
\newblock In Aart Middeldorp, editor, {\em Proceedings of the 8th International
  Workshop on Reduction Strategies in Rewriting and Programming, WRS@RTA 2008,
  Hagenberg, Austria, July 14, 2008}, volume 237 of {\em Electronic Notes in
  Theoretical Computer Science}, pages 57--73. Elsevier, 2008.
\newblock \href {https://doi.org/10.1016/J.ENTCS.2009.03.035}
  {\path{doi:10.1016/J.ENTCS.2009.03.035}}.

\bibitem[Ghi07]{ghica_geometry_2007}
Dan~R. Ghica.
\newblock Geometry of synthesis: a structured approach to {VLSI} design.
\newblock In Martin Hofmann and Matthias Felleisen, editors, {\em Proceedings
  of the 34th {ACM} {SIGPLAN-SIGACT} Symposium on Principles of Programming
  Languages, {POPL} 2007, Nice, France, January 17-19, 2007}, pages 363--375.
  {ACM}, 2007.
\newblock \href {https://doi.org/10.1145/1190216.1190269}
  {\path{doi:10.1145/1190216.1190269}}.

\bibitem[Gir89]{girard_geometry_1989}
Jean-Yves Girard.
\newblock Geometry of {Interaction} 1: {Interpretation} of {System} {F}.
\newblock In R.~Ferro, C.~Bonotto, S.~Valentini, and A.~Zanardo, editors, {\em
  Logic Colloquium '88}, volume 127 of {\em Studies in Logic and the
  Foundations of Mathematics}, pages 221--260. Elsevier, 1989.
\newblock \href {https://doi.org/10.1016/S0049-237X(08)70271-4}
  {\path{doi:10.1016/S0049-237X(08)70271-4}}.

\bibitem[GMR12]{DBLP:journals/tocl/GaboardiMR12}
Marco Gaboardi, Jean{-}Yves Marion, and Simona Ronchi~Della Rocca.
\newblock An implicit characterization of {PSPACE}.
\newblock {\em {ACM} Trans. Comput. Log.}, 13(2):18:1--18:36, 2012.
\newblock \href {https://doi.org/10.1145/2159531.2159540}
  {\path{doi:10.1145/2159531.2159540}}.

\bibitem[HU79]{DBLP:books/aw/HopcroftU79}
John~E. Hopcroft and Jeffrey~D. Ullman.
\newblock {\em Introduction to Automata Theory, Languages and Computation}.
\newblock Addison-Wesley, 1979.

\bibitem[JHM11]{DBLP:books/wi/Jones2011}
Richard~E. Jones, Antony~L. Hosking, and J.~Eliot~B. Moss.
\newblock {\em The Garbage Collection Handbook: The art of automatic memory
  management}.
\newblock Chapman and Hall / {CRC} Applied Algorithms and Data Structures
  Series. {CRC} Press, 2011.
\newblock URL: \url{http://gchandbook.org/}.

\bibitem[Jon99]{DBLP:journals/tcs/Jones99}
Neil~D. Jones.
\newblock {LOGSPACE} and {PTIME} characterized by programming languages.
\newblock {\em Theor. Comput. Sci.}, 228(1-2):151--174, 1999.
\newblock \href {https://doi.org/10.1016/S0304-3975(98)00357-0}
  {\path{doi:10.1016/S0304-3975(98)00357-0}}.

\bibitem[KBH12]{DBLP:conf/popl/KrishnaswamiBH12}
Neelakantan~R. Krishnaswami, Nick Benton, and Jan Hoffmann.
\newblock Higher-order functional reactive programming in bounded space.
\newblock In John Field and Michael Hicks, editors, {\em Proceedings of the
  39th {ACM} {SIGPLAN-SIGACT} Symposium on Principles of Programming Languages,
  {POPL} 2012, Philadelphia, Pennsylvania, USA, January 22-28, 2012}, pages
  45--58. {ACM}, 2012.
\newblock \href {https://doi.org/10.1145/2103656.2103665}
  {\path{doi:10.1145/2103656.2103665}}.

\bibitem[Kri07]{krivine_call-by-name_2007}
Jean-Louis Krivine.
\newblock A {Call}-by-name {Lambda}-calculus {Machine}.
\newblock {\em Higher Order Symbol. Comput.}, 20(3):199--207, 2007.
\newblock \href {https://doi.org/10.1007/s10990-007-9018-9}
  {\path{doi:10.1007/s10990-007-9018-9}}.

\bibitem[Ler90]{Leroy-ZINC}
Xavier Leroy.
\newblock {The {ZINC} experiment: an economical implementation of the {ML}
  language}.
\newblock Technical report 117, INRIA, 1990.
\newblock URL: \url{http://gallium.inria.fr/~xleroy/publi/ZINC.pdf}.

\bibitem[Mac95]{mackie_geometry_1995}
Ian Mackie.
\newblock The {Geometry} of {Interaction} {Machine}.
\newblock In Ron~K. Cytron and Peter Lee, editors, {\em Conference Record of
  POPL'95: 22nd {ACM} {SIGPLAN-SIGACT} Symposium on Principles of Programming
  Languages, San Francisco, California, USA, January 23-25, 1995}, pages
  198--208. {ACM} Press, 1995.
\newblock \href {https://doi.org/10.1145/199448.199483}
  {\path{doi:10.1145/199448.199483}}.

\bibitem[Maz15]{DBLP:conf/csl/Mazza15}
Damiano Mazza.
\newblock Simple parsimonious types and logarithmic space.
\newblock In Stephan Kreutzer, editor, {\em 24th {EACSL} Annual Conference on
  Computer Science Logic, {CSL} 2015, September 7-10, 2015, Berlin, Germany},
  volume~41 of {\em LIPIcs}, pages 24--40. Schloss Dagstuhl - Leibniz-Zentrum
  f{\"{u}}r Informatik, 2015.
\newblock \href {https://doi.org/10.4230/LIPIcs.CSL.2015.24}
  {\path{doi:10.4230/LIPIcs.CSL.2015.24}}.

\bibitem[PA19]{DBLP:journals/pacmpl/Paraskevopoulou19}
Zoe Paraskevopoulou and Andrew~W. Appel.
\newblock Closure conversion is safe for space.
\newblock {\em Proc. {ACM} Program. Lang.}, 3({ICFP}):83:1--83:29, 2019.
\newblock \href {https://doi.org/10.1145/3341687} {\path{doi:10.1145/3341687}}.

\bibitem[SBHG10]{DBLP:journals/jfp/SpoonhowerBHG08}
Daniel Spoonhower, Guy~E. Blelloch, Robert Harper, and Phillip~B. Gibbons.
\newblock Space profiling for parallel functional programs.
\newblock {\em J. Funct. Program.}, 20(5-6):417--461, 2010.
\newblock \href {https://doi.org/10.1017/S0956796810000146}
  {\path{doi:10.1017/S0956796810000146}}.

\bibitem[Sch06]{DBLP:conf/csl/Schopp06}
Ulrich Sch{\"{o}}pp.
\newblock Space-efficient computation by interaction.
\newblock In Zolt{\'{a}}n {\'{E}}sik, editor, {\em Computer Science Logic, 20th
  International Workshop, {CSL} 2006, 15th Annual Conference of the EACSL,
  Szeged, Hungary, September 25-29, 2006, Proceedings}, volume 4207 of {\em
  Lecture Notes in Computer Science}, pages 606--621. Springer, 2006.
\newblock \href {https://doi.org/10.1007/11874683\_40}
  {\path{doi:10.1007/11874683\_40}}.

\bibitem[Sch07]{bllspace}
Ulrich Schopp.
\newblock Stratified bounded affine logic for logarithmic space.
\newblock In {\em 22nd {IEEE} Symposium on Logic in Computer Science {(LICS}
  2007), 10-12 July 2007, Wroclaw, Poland, Proceedings}, pages 411--420. {IEEE}
  Computer Society, 2007.
\newblock \href {https://doi.org/10.1109/LICS.2007.45}
  {\path{doi:10.1109/LICS.2007.45}}.

\bibitem[Ses97]{DBLP:journals/jfp/Sestoft97}
Peter Sestoft.
\newblock Deriving a lazy abstract machine.
\newblock {\em J. Funct. Program.}, 7(3):231--264, 1997.
\newblock URL:
  \url{http://journals.cambridge.org/action/displayAbstract?aid=44087}.

\bibitem[SGM02]{DBLP:conf/birthday/SandsGM02}
David Sands, J{\"o}rgen Gustavsson, and Andrew Moran.
\newblock Lambda calculi and linear speedups.
\newblock In Torben~{\AE}. Mogensen, David~A. Schmidt, and Ivan~Hal Sudborough,
  editors, {\em The Essence of Computation, Complexity, Analysis,
  Transformation. Essays Dedicated to Neil D. Jones [on occasion of his 60th
  birthday]}, volume 2566 of {\em Lecture Notes in Computer Science}, pages
  60--84. Springer, 2002.
\newblock \href {https://doi.org/10.1007/3-540-36377-7\_4}
  {\path{doi:10.1007/3-540-36377-7\_4}}.

\bibitem[SJ95]{DBLP:conf/popl/SansomJ95}
Patrick~M. Sansom and Simon L.~Peyton Jones.
\newblock Time and space profiling for non-strict higher-order functional
  languages.
\newblock In Ron~K. Cytron and Peter Lee, editors, {\em Conference Record of
  POPL'95: 22nd {ACM} {SIGPLAN-SIGACT} Symposium on Principles of Programming
  Languages, San Francisco, California, USA, January 23-25, 1995}, pages
  355--366. {ACM} Press, 1995.
\newblock \href {https://doi.org/10.1145/199448.199531}
  {\path{doi:10.1145/199448.199531}}.

\bibitem[SvEB88]{DBLP:journals/iandc/SlotB88}
Cees~F. Slot and Peter van Emde~Boas.
\newblock The problem of space invariance for sequential machines.
\newblock {\em Inf. Comput.}, 77(2):93--122, 1988.
\newblock \href {https://doi.org/10.1016/0890-5401(88)90052-1}
  {\path{doi:10.1016/0890-5401(88)90052-1}}.

\bibitem[Ter08]{TeruiSlides}
Kazushige Terui.
\newblock On space efficiency of krivine’s abstract machine and hyland-ong
  games.
\newblock \url{https://www.kurims.kyoto-u.ac.jp/~terui/space2.pdf}, 2008.
\newblock Accessed: 2022-05-31.

\bibitem[vEB90]{vanEmdeBoas90}
Peter van Emde~Boas.
\newblock Machine models and simulation.
\newblock In {\em Handbook of Theoretical Computer Science, Volume A:
  Algorithms and Complexity (A)}, pages 1--66. MIT Press, 1990.

\bibitem[vEB12]{DBLP:conf/sofsem/Boas12}
Peter van Emde~Boas.
\newblock Turing machines for dummies - why representations do matter.
\newblock In M{\'{a}}ria Bielikov{\'{a}}, Gerhard Friedrich, Georg Gottlob,
  Stefan Katzenbeisser, and Gy{\"{o}}rgy Tur{\'{a}}n, editors, {\em {SOFSEM}
  2012: Theory and Practice of Computer Science - 38th Conference on Current
  Trends in Theory and Practice of Computer Science, {\v{S}}pindler{\r{u}}v
  Ml{\'{y}}n, Czech Republic, January 21-27, 2012. Proceedings}, volume 7147 of
  {\em Lecture Notes in Computer Science}, pages 14--30. Springer, 2012.
\newblock \href {https://doi.org/10.1007/978-3-642-27660-6\_2}
  {\path{doi:10.1007/978-3-642-27660-6\_2}}.

\bibitem[Wan07]{DBLP:journals/lisp/Wand07}
Mitchell Wand.
\newblock On the correctness of the krivine machine.
\newblock {\em High. Order Symb. Comput.}, 20(3):231--235, 2007.
\newblock \href {https://doi.org/10.1007/s10990-007-9019-8}
  {\path{doi:10.1007/s10990-007-9019-8}}.

\end{thebibliography}

\pagebreak
\appendix
% !TeX spellcheck = en_US
% !TEX root = main.tex
\section{Proofs of Section~\ref{sect:moving}}\label{app:moving}
In the following we will often use the execution of the fixed point combinator $\fix\defeq\theta\theta$, where $\theta\defeq\la\var \la\vartwo \vartwo(\var \var \vartwo )$. 
For this reason, we encapsulate its execution by the \SpKAM in a lemma. 
\begin{lem}\label{l:fix}
	For each term $\tmtwo$,  
	$\spkamstate{\theta}\stempty{(\theta,\stempty)\cons(\tmtwo,\lenv)\cons\stack}
	\tospkam^{\bigo{1}} 
	\spkamstate\tmtwo\lenv{\fix^{\kop}\cons\stack}$
	where $\fix^{\kop}\defeq(\var \var 
	\vartwo,\esub\vartwo{(\tmtwo,\lenv)}\cons\esub\var{(\theta,\stempty)})$ 
	consuming space $\bigo{\size\lenv+\size\stack+\log(\size\tmtwo)}$.
\end{lem}
\begin{proof}
	\begin{align*}
	\begin{array}{l|l|ll}
	\mathsf{Term}   & 
	\mathsf{Env} & \mathsf{Stack}\\
	\hline
	\theta\defeq\la\var \la\vartwo \vartwo(\var \var \vartwo ) & \stempty & 
	(\theta,\stempty)\cons(\tmtwo,\lenv)\cons\stack&\tokambnw\\
	\la\vartwo \vartwo(\var \var \vartwo ) & \esub\var{(\theta,\stempty)} & 
	(\tmtwo,\lenv)\cons\stack&\tokambnw\\
	\vartwo(\var \var \vartwo ) & 
	\esub\vartwo{(\tmtwo,\lenv)}\cons\esub\var{(\theta,\stempty)} & 
	\stack&\tokamseanv\\
	\vartwo & \esub\vartwo{(\tmtwo,\lenv)} & \overbrace{(\var \var 
		\vartwo,\esub\vartwo{(\tmtwo,\lenv)}\cons\esub\var{(\theta,\stempty)})}^{\fix^{\kop}}\cons\stack&\tokamsub\\
	\tmtwo & \lenv & \fix^{\kop}\cons\stack
	\end{array}\\
	\tag*{\qedhere}
\end{align*}
\end{proof}
\subsection{Proof of Proposition~\ref{prop:toy}}

We prove the following proposition in a top-down style, \ie required lemmata 
are below. This is done because otherwise lemmata statements would seem quite 
arbitrary.
Nonetheless, we need a preliminary definition of some specific environments.
\begin{defi}
	Let $\str\defeq\bool_1\cons\ldots\cons\bool_n\cons\ems$ be a string of length $n\geq 0$. 
	Then, for each $0\leq i\leq n$ we can define 
	$\lenv_i,\lenvtwo_i,\lenvthree_i$  as follows:
	\[
	\begin{array}{rcl@{\hspace{1cm}}rcl}
	\lenv_0 &\defeq& \esub\var{(\theta,\stempty)}\cons 
	\esub \vartwo{(\toyaux,\stempty)} & \lenv_{i+1}&\defeq& 
	\esub\var{(\var,\lenv_i)}\cons 
	\esub \vartwo{(\vartwo,\lenv_i)}\\[4pt]
	\lenvthree_0&\defeq& \stempty& 
	\lenvthree_{i+1}&\defeq&\esub{\var_\ems}{(\Id,\lenvtwo_{i})} \cons 
	\esub{\var_1}{(\varf,\lenvtwo_{i})} \cons 
	\esub{\var_0}{(\varf,\lenvtwo_{i})}\cons\lenvthree_i\\[4pt]
	\multicolumn{6}{c}{\lenvtwo_{i}\defeq 
	\esub\varthree{(\enc{\bool_{i+1}..\bool_n\cons\ems},\lenvthree_{i})} 
	\cons\esub\varf{(\var \var \vartwo,\lenv_{i})}}\\[4pt]
	\end{array}\]
\end{defi}
One can easily notice that the sizes of $\lenv_i,\lenvtwo_i,\lenvthree_i$ are 
exponential in $i$.

\begin{prop}\label{p:toyapp}
	Let $\strone\in\Bool^*$ and $\toy \defeq \fix \toyaux$ where $\toyaux\defeq\la\varf\la 
	{\varthree}\varthree 
	\varf\varf \Id$.
	\begin{enumerate}
		\item $\toy\, \encp\strone\Bool \towh^{\Theta(\size\strone)} \Id$.
		\item The \NaKAM evaluates $\toy\, 
		\encp\strone\Bool$ in 
		space 
		$\Omega(2^{\size\strone})$.
		\item The \SpKAM evaluates $\toy\, \encp\strone\Bool$ in space 
		$\Theta(\log\size\strone)$.
	\end{enumerate}
\end{prop}
\begin{proof}Since $\Bool$ is the only 
	alphabet that we are using, we remove all the superscripts.
	\begin{enumerate}
		\item This point follows from the implementation theorem, applied to 
		the sequence of point 3.
		\item We prove the statement executing 
		$\toy\, \enc\strone$ with the \NaKAM.
		\[
		\begin{array}{l|l|ll}
		\mathsf{Term}   & 
		\mathsf{Env} & \mathsf{Stack}\\
		\hline
		\toy\, \enc\strone & \stempty & \stempty& \tokamsea\\
		\toy\defeq\fix\toyaux  & \stempty & (\enc\strone,\stempty) & \tokamsea\\
		\fix\defeq\theta\theta & \stempty & 
		(\toyaux,\stempty)\cons(\enc\strone,\stempty) & \tokamsea\\
		\theta\defeq\la\var \la\vartwo \vartwo(\var \var \vartwo ) & 
		\stempty & (\theta,\stempty) \cons (\toyaux,\stempty) 
		\cons (\enc\strone,\stempty) & \tokamb^2\\
		\vartwo(\var \var \vartwo ) & \esub\var{(\theta,\stempty)}\cons 
		\esub \vartwo{(\toyaux,\stempty)}=:\lenv_0 & (\enc\strone,\stempty) 
		& 
		\text{Lemma~\ref{l:inv-kam-scrolling}}\\
		\Id & \lenvtwo_{\size\str} & \stempty
		\end{array}
		\]
		The space bound is proved since $\lenvtwo_{\size\str}$ is exponential 
		in $\size\str$.
		\item We prove the statement executing 
		$\toy\, \enc\strone$ with the \SpKAM.
		\[
		\begin{array}{l|l|ll}
		\mathsf{Term}   & 
		\mathsf{Env} & \mathsf{Stack}\\
		\hline
		\toy\, \enc\strone & \stempty & \stempty & \tokamseanv\\
		\toy\defeq\fix\toyaux  & \stempty & (\enc\strone,\stempty) & \tokamseanv\\
		\fix\defeq\theta\theta & \stempty & 
		(\toyaux,\stempty)\cons(\enc\strone,\stempty) & \tokamseanv\\
		\theta\defeq\la\var \la\vartwo \vartwo(\var \var \vartwo ) & 
		\stempty & (\theta,\stempty) \cons (\toyaux,\stempty) 
		\cons (\enc\strone,\stempty) & \tokambnw^2\\
		\vartwo(\var \var \vartwo ) & \esub\var{(\theta,\stempty)}\cons 
		\esub \vartwo{(\toyaux,\stempty)}=:\lenv_0 & (\enc\strone,\stempty) 
		& 
		\text{Lemma~\ref{l:inv-spkam-scrolling}}\\
		\Id & \stempty & \stempty
		\end{array}
		\]
		The space bound is proved considering the bound in 
		Lemma~\ref{l:inv-spkam-scrolling}. \qedhere
	\end{enumerate}
\end{proof}

The second point in the statement of the previous proposition needs the following auxiliary lemma, proved by induction. 
\begin{lem}\label{l:inv-kam-scrolling}
	$\spkamstate{\vartwo(\var \var \vartwo 
		)}{\lenv_i}{(\enc\str,\lenvthree_i)}\tonakam^{\Omega(\size\str)} 
	\spkamstate{\Id}{\lenvtwo_{i+\size\str}}{\stempty}$.
\end{lem}
\begin{proof}
	By induction on the structure of $\str$.
	\[
	\begin{array}{l|l|ll}
	\mathsf{Term}   & 
	\mathsf{Env} & \mathsf{Stack}\\
	\hline \rule{0pt}{2.3ex}
	\vartwo(\var \var \vartwo ) & \lenv_i & (\enc\strone,\lenvthree_i) & \tokamsea\\
	\vartwo & \lenv_i & (\var \var \vartwo,\lenv_i)\cons 
	(\enc\strone,\lenvthree_i)&\tokamsub^{i+1}\\
	\toyaux\defeq\la\varf\la 
	{\varthree}\varthree 
	\varf\varf \Id & \stempty & (\var \var \vartwo,\lenv_i)\cons 
	(\enc\strone,\lenvthree_i) & \tokamb^2\\
	\varthree \varf\varf \Id & 
	\esub\varthree{(\enc\strone,\lenvthree_i)}\cons\esub\varf{(\var \var 
		\vartwo,\lenv_i)}=:\lenvtwo_i & \stempty& \tokamsea^3\\
	\varthree & 
	\esub\varthree{(\enc\strone,\lenvthree_i)}\cons\esub\varf{(\var \var 
		\vartwo,\lenv_i)} & (\varf,\lenvtwo_i) \cons (\varf,\lenvtwo_i) 
	\cons (\Id,\lenvtwo_i) & \tokamsub\\
	\enc\strone & \lenvthree_i & (\varf,\lenvtwo_i) \cons 
	(\varf,\lenvtwo_i) 
	\cons (\Id,\lenvtwo_i)
	\end{array}
	\]
	Case $\str=\ems$.
	\[
	\begin{array}{l|l|ll}
	\mathsf{Term}   & 
	\mathsf{Env} & \mathsf{Stack}\\
	\hline \rule{0pt}{2.3ex}
	\enc\strone\defeq\la{\var_0}\la{\var_1}\la{\var_\ems}\var_\ems & 
	\lenvthree_i & (\varf,\lenvtwo_i) \cons (\varf,\lenvtwo_i) 
	\cons (\Id,\lenvtwo_i) & \tokamb^3\\
	\var_\ems & \esub{\var_\ems}{(\Id,\lenvtwo_i)} \cons 
	\esub{\var_1}{(\varf,\lenvtwo_i)} \cons 
	\esub{\var_0}{(\varf,\lenvtwo_i)}\cons \lenvthree_i & \stempty & \tokamsub\\
	\Id & \lenvtwo_i & \stempty
	\end{array}
	\]
	Case $\str=\bool\cons\strtwo$.
	\begin{align*}
	\footnotesize\begin{array}{l|l|ll}
	\mathsf{Term}   & 
	\mathsf{Env} & \mathsf{Stack}\\
	\hline \rule{0pt}{2.3ex}
	\enc\strone\defeq\la{\var_0}\la{\var_1}\la{\var_\ems}\var_b\enc\strtwo
	& \lenvthree_i & (\varf,\lenvtwo_i) \cons (\varf,\lenvtwo_i) 
	\cons (\Id,\lenvtwo_i) & \tokamb^3\\
	\var_\bool\enc\strtwo & \esub{\var_\ems}{(\Id,\lenvtwo_i)} \cons 
	\esub{\var_1}{(\varf,\lenvtwo_i)} \cons 
	\esub{\var_0}{(\varf,\lenvtwo_i)}\cons\lenvthree_i=:\lenvthree_{i+1} & 
	\stempty & \tokamsea\\
	\var_\bool & \esub{\var_\ems}{(\Id,\lenvtwo_i)} \cons 
	\esub{\var_1}{(\varf,\lenvtwo_i)} \cons 
	\esub{\var_0}{(\varf,\lenvtwo_i)}\cons \lenvthree_i & 
	(\enc\strtwo,\lenvthree_{i+1}) & \tokamsub\\
	\varf & \esub\varthree{(\enc\strone,\lenvthree_i)}\cons\esub\varf{(\var 
		\var \vartwo,\lenv_i)} & (\enc\strtwo,\lenvthree_{i+1}) & \tokamsub\\
	\var \var \vartwo & \esub{\var}{(\var,\lenv_{i-1})} \cons 
	\esub{\vartwo}{(\vartwo,\lenv_{i-1})} & (\enc\strtwo,\lenvthree_{i+1}) & \tokamsea^2\\
	\var &  \esub{\var}{(\var,\lenv_{i-1})} \cons 
	\esub{\vartwo}{(\vartwo,\lenv_{i-1})} & (\var,\lenv_i) 
	\cons(\vartwo,\lenv_i) \cons (\enc\strtwo,\lenvthree_{i+1}) & 
	\tokamsub^{i+1}\\
	\theta\defeq\la\var \la\vartwo \vartwo(\var \var \vartwo ) & 
	\stempty & (\var,\lenv_i) 
	\cons(\vartwo,\lenv_i) \cons (\enc\strtwo,\lenvthree_{i+1}) &\tokamb^2\\
	\vartwo(\var \var \vartwo ) & \esub{\var}{(\var,\lenv_i)} \cons 
	\esub{\vartwo}{(\vartwo,\lenv_i)}=:\lenv_{i+1} & 
	(\enc\strtwo,\lenvthree_{i+1}) & \ih\\
	\Id & \lenvtwo_{i+\size\str} & \stempty
	\end{array}\\
	\tag*{\qedhere}
	\end{align*}
\end{proof}
The third point in the statement of the proposition~\ref{p:toyapp} needs the following auxiliary lemma, proved by induction.
\begin{lem}\label{l:inv-spkam-scrolling}
	The \SpKAM executes the reduction $\spkamstate{\vartwo(\var \var \vartwo 
		)}{\lenv_0}{(\enc\str,\stempty)}\tospkam^{\Theta(\size\str)} 
	\spkamstate{\Id}{\stempty}{\stempty}$ consuming $\bigo{\log(\size\str)}$ 
	space.
\end{lem}
\begin{proof}
	By induction on the structure of $\str$.
	\[
	\begin{array}{l|l|ll}
	\mathsf{Term}   & 
	\mathsf{Env} & \mathsf{Stack}\\
	\hline
	\vartwo(\var \var \vartwo ) & \lenv_0 & (\enc\strone,\stempty)& \tokamseanv\\
	\vartwo & \esub \vartwo{(\toyaux,\stempty)} & (\var \var 
	\vartwo,\lenv_0)\cons 
	(\enc\strone,\stempty) & \tokamsub\\
	\toyaux\defeq\la\varf\la {\varthree}\varthree 
	\varf\varf \Id & \stempty & (\var \var \vartwo,\lenv_0)\cons 
	(\enc\strone,\stempty) & \tokambnw^2\\
	\varthree \varf\varf \Id & 
	\esub\varthree{(\enc\strone,\stempty)}\cons\esub\varf{(\var \var 
		\vartwo,\lenv_0)} & \stempty & \tokamseanv\tokamseav^2\\
	\varthree & 
	\esub\varthree{(\enc\strone,\stempty)} & (\var \var \vartwo,\lenv_0) 
	\cons (\var \var 
	\vartwo,\lenv_0) 
	\cons (\Id,\stempty)& \tokamsub\\
	\enc\strone & \stempty & (\var \var \vartwo,\lenv_0) 
	\cons (\var \var 
	\vartwo,\lenv_0) 
	\cons (\Id,\stempty)
	\end{array}
	\]
	Case $\str=\ems$.
	\[
	\begin{array}{l|l|ll}
	\mathsf{Term}   & 
	\mathsf{Env} & \mathsf{Stack}\\
	\hline \rule{0pt}{2.3ex}
	\enc\strone\defeq\la{\var_0}\la{\var_1}\la{\var_\ems}\var_\ems & 
	\stempty & (\var \var \vartwo,\lenv_0) 
	\cons (\var \var 
	\vartwo,\lenv_0) 
	\cons (\Id,\stempty) & \tokambw^2\tokamb \\
	\var_\ems & \esub{\var_\ems}{(\Id,\stempty)}  & \stempty & \tokamsub\\
	\Id & \stempty & \stempty
	\end{array}
	\]
	Case $\str=\bool\cons\strtwo$.
	\[
	\begin{array}{l|l|ll}
	\mathsf{Term}   & 
	\mathsf{Env} & \mathsf{Stack}\\
	\hline \rule{0pt}{2.3ex}
	\enc\strone\defeq\la{\var_0}\la{\var_1}\la{\var_\ems}\var_b\enc\strtwo
	& \stempty & (\var \var \vartwo,\lenv_0) 
	\cons (\var \var \vartwo,\lenv_0) \cons (\Id,\stempty) & \tokamb^2\tokambw\\
	\var_\bool\enc\strtwo & \esub{\var_b}{(\var \var \vartwo,\lenv_0)}  & 
	\stempty & \tokamsea\\
	\var_\bool & \esub{\var_b}{(\var \var \vartwo,\lenv_0)} & 
	(\enc\strtwo,\stempty) & \tokamsub\\
	\var \var \vartwo & \lenv_0 & (\enc\strtwo,\stempty) & \tokamseav\\
	\var &  \esub{\var}{(\theta,\stempty)} & (\theta,\stempty) 
	\cons(\toyaux,\stempty) \cons (\enc\strtwo,\stempty) & \tokamsub\\
	\theta\defeq\la\var \la\vartwo \vartwo(\var \var \vartwo ) & 
	\stempty & (\theta,\stempty) 
	\cons(\toyaux,\stempty) \cons (\enc\strtwo,\stempty) & \tokambnw\\
	\vartwo(\var \var \vartwo ) & \lenv_0 & 
	(\enc\strtwo,\stempty) & \ih\\
	\Id & \stempty & \stempty
	\end{array}
	\]
	The space bound is proved since there is a static bound, namely 8, on the number of closures stored during the execution, and the size of each closure is clearly bounded by $\log(\size{\str})$.
\end{proof}

\subsection{Proof of Proposition~\ref{prop:toytapeglobal}}
As before an auxiliary lemma is required to prove the proposition. We state and prove it below the main proposition.
\begin{prop}
	Let $\strone\in\Bool^*$ and $\globalcopy \defeq \la\varthree(\fix 
	(\la\varf\la 
	{\strone'}\strone' \varf\varf \varthree) \varthree)$.
	\begin{enumerate}
		\item $\globalcopy\, \encp\strone\Bool \towh^{\Theta(\size\strone)} 
		\encp\strone\Bool$.
		\item The space used by the \SpKAM to simulate the evaluation of the 
		previous point is $\Theta(\log\size\strone)$.
	\end{enumerate}
\end{prop}
\begin{proof}
	The first point is a consequence of the second one, since the \SpKAM correctly implements Closed Call-by-Name. We prove the second point of the statement by directly executing the 
	\SpKAM. Since $\Bool$ is the only 
	alphabet that we are using, we remove all the superscripts. Let us define 
	$\tm\defeq\la\varf\la 
	{\strone'}\strone' \varf\varf \varthree$. 
	
	\[
	\begin{array}{l|l|ll}
	\mathsf{Term} & \mathsf{Environment} & \mathsf{Stack}\\ 
	\cline{1-3}
	\globalcopy\, \enc\strone & \stempty & \stempty & \tokamseanv\\
	\globalcopy\defeq\la\varthree(\fix \tm \varthree) & \stempty & 
	(\enc\strone,\stempty)=:\str^\kop & \tokambnw\\
	\fix \tm \varthree & \esub\varthree{\str^\kop} & 
	\stempty & \tokamseav\\
	\fix \tm & \esub\varthree{\str^\kop} & \str^\kop & \tokamseanv\\
	\fix\defeq \theta \theta & 
	\stempty & (\tm,\esub\varthree{\str^\kop}) \cons \str^\kop & \tokamseanv\\
	\theta\defeq\la\var \la\vartwo \vartwo(\var \var \vartwo) & \stempty & 
	(\theta,\stempty) \cons (\tm,\esub\varthree{\str^\kop}) \cons \str^\kop 
	& \to^{\Theta(\size{\str})}(\text{ Lemma~\ref{l:global}})\\
	\enc\str & \stempty & \stempty
	\end{array}
	\]
	The space bound is immediate considering the space bound of lemma~\ref{l:global}, and the fact that during the execution only a fixed number of closures is stored.
\end{proof}
\begin{lem}\label{l:global}
	Let $\strone\in\Bool^*$ and $\tm\defeq\la\varf\la 
	{\strone'}\strone' \varf\varf \varthree$. Then \\
	$\spkamstate{\theta}{\stempty}{(\theta,\stempty) \cons 
		(\tm,\esub\varthree{(\tmtwo,\lenv)}) \cons 
		\str^\kop}\tospkam^{\Theta(\size{\str})}\spkamstate{\tmtwo}{\lenv}{\stempty}$
	and the space used is 
	$\bigo{\size\lenv+\log\size{\str}+\log\size{\tmtwo}}$.
\end{lem}
\begin{proof}
	We proceed by induction on the 
	structure of $\str$. The first steps are common to both the base case and 
	the induction step. We define $\fix^\kop\defeq(\var\var\vartwo, 
	\esub\vartwo{(\tm,\esub\varthree{(\tmtwo,\lenv)})} \cons 
	\esub\var{(\theta,\stempty)})$.
	
	\[
	\begin{array}{l|l|ll}
	\mathsf{Term} & \mathsf{Environment} & \mathsf{Stack}\\ 
	\cline{1-3} \rule{0pt}{2.3ex}
	\theta\defeq\la\var \la\vartwo \vartwo(\var \var \vartwo) & \stempty & 
	(\theta,\stempty) \cons (\tm,\esub\varthree{(\tmtwo,\lenv)}) \cons 
	\str^\kop 
	& \to(\text{ Lemma~\ref{l:fix}})\\
	\tm\defeq\la\varf\la 
	{\strone'}\strone' \varf\varf \varthree & 
	\esub\varthree{(\tmtwo,\lenv)} & 
	\fix^\kop \cons \str^\kop & \tokambnw^2\\
	\strone' \varf\varf \varthree & \esub\varf{\fix^\kop} 
	\cons\esub{\strone'}{\str^\kop}\cons \esub\varthree{(\tmtwo,\lenv)} & 
	\stempty & \tokamseav^3\\
	\strone' & \esub{\strone'}{\str^\kop} & \fix^\kop \cons \fix^\kop \cons 
	(\tmtwo,\lenv) & \tokamsub\\
	\enc\str & \stempty & \fix^\kop \cons \fix^\kop \cons 
	(\tmtwo,\lenv)
	\end{array}
	\]
	
	Base case: $\str=\ems$.
	
	\[
	\begin{array}{l|l|ll}
	\mathsf{Term} & \mathsf{Environment} & \mathsf{Stack}\\ 
	\cline{1-3} \rule{0pt}{2.3ex}
	\enc\str\defeq\la{\var_0}\la{\var_1}\la{\var_\ems}\var_\ems & \stempty 
	& \fix^\kop \cons \fix^\kop \cons (\tmtwo,\lenv) & \tokambw^2\tokambnw\tokamsub\\
	\tmtwo & \lenv & \stempty
	\end{array}
	\]
	
	Inductive case: $\str:b\cons\strtwo$ where $b\in\{0,1\}$.
	
	\[
	\begin{array}{l|l|ll}
	\mathsf{Term} & \mathsf{Environment} & \mathsf{Stack}\\ 
	\cline{1-3} \rule{0pt}{2.3ex}
	\enc\str\defeq\la{\var_0}\la{\var_1}\la{\var_\ems}\var_b\enc\strtwo & 
	\stempty 
	& \fix^\kop \cons \fix^\kop \cons (\tmtwo,\lenv) & \tokamb^2\tokambw\\
	\var_b\enc\strtwo & \esub{\var_b}{\fix^\kop} & \stempty & \tokamsea\\
	\var_b & \esub{\var_b}{\fix^\kop} & (\enc\strtwo,\stempty)=: 
	\strtwo^\kop & \tokamsub\\
	\var \var \vartwo & \esub\vartwo{(\tm,\esub\varthree{(\tmtwo,\lenv)})} 
	\cons 
	\esub\var{(\theta,\stempty)} & \strtwo^\kop & \tokamseav^2\\
	\var & \esub\var{(\theta,\stempty)} & (\theta,\stempty) \cons 
	(\tm,\esub\varthree{(\tmtwo,\lenv)}) \cons \strtwo^\kop & \tokamsub\\
	\theta & \stempty & (\theta,\stempty) \cons 
	(\tm,\esub\varthree{(\tmtwo,\lenv)}) \cons \strtwo^\kop & 
	\to^{\Theta(\size\strtwo)} 
	\text{ \ih}\\
	\tmtwo & \lenv & \stempty
	\end{array}
	\]
	
	About space, it is immediate to see that all computations are constrained 
	in space $\bigo{\size\lenv+\log\size{\str}+\log\size{\tmtwo}}$, since at 
	any point 
	during the computation there is a \emph{bounded} number of closures, 
	independent from both $\str$ and $\tmtwo$.
\end{proof}
% !TeX spellcheck = en_US
% !TEX root = main.tex
\section{Proofs of Section~\ref{sect:reasonability}}\label{sec:main-proof}
This appendix is devoted to the proof of the main theorem of 
the paper, \ie the space reasonable simulation of TMs into the $\l$-calculus 
(better, the \SpKAM). It is a boring proof, where we simply execute the image 
of the encoding of TMs into the $\l$-calculus with the \SpKAM.

First, we need to understand how a TM configuration is represented in the 
\SpKAM, \ie how it is mapped to environments and closures. This will be a sort of \emph{invariant} of the execution.

\begin{defi}
	A configuration $\config$ of a TM is represented as a KAM closure 
	$\config^{\kop}$ in 
	the 
	following way:
	\[
	(\inputstr,\counter,\strone,\elem,\strtwo,\state)^{\kop}\defeq(\tuple{f,c,m,\cods\elem,
		d,\cods\state},\esub f{(\enc\inputstr,\stempty)},\esub c {n^{\kop}}, 
	\esub m {\str{^{\kop}}}, \esub d {\strtwo^{\kop}})
	\]
	where
	$\str^{\kop}=\begin{cases}
	(\cods\ems,\stempty) & \text{if }\str=\ems\\
	(\lambda \var_1.\ldots.\lambda \var_{\size\alpone}.\lambda 
	\vartwo.\var_{i_\elem}z,\esub z {r^{\kop}}) & \text{if 
	}\str=\elem_i\strtwo
	\end{cases}
	$
\end{defi}

We observe that this representation preserves the space consumption, \ie it is 
reasonable.
\begin{lem}\label{l:config-size}
	Let $\config\defeq(\inputstr,\counter,\strone,\elem,\strtwo,\state)$ be a 
	configuration of a Turing machine and 
	$\size{\config}\defeq\size\strone+\size\strtwo$ its 
	space consumption. Then 
	$\size{\config^{\kop}}=\Theta(\size{\config}+\log(\size 
	i))$.
\end{lem}
In this lemma, we have already considered that the size of pointers inside 
$\counter^\kop$, $\str^\kop$, $\strtwo^\kop$ is constant and that 
$n\leq\log\size\istr$.

Now we are able to prove the main theorem. A series of intermediate lemmata, about 
the different combinators used in the encoding ($\inits,\finals,\transs$), are 
necessary. They are stated 
and proved below the main statement . By $\too^*_f$, we mean that the space 
consumption of that series of transitions is $f$.

\begin{thm}[TM are simulated by the \SpKAM in reasonable space]
	There is an encoding $\enc{\cdot}$ of log-sensitive TM into 
	$\detLam$ such that if the run $\run$ of the TM $M$ on input 
	$\istr\in\Bool^{*}$:
	\begin{enumerate}
		\item 
		\emph{Termination}: ends in $\state_{\bool}$ with $\bool\in\Bool$, then
		there is a complete sequence $\runtwo:\enc M\, \enc\istr \tobdet^n 
		\enc{\bool}$
		where $n=\Theta((T_{\textrm{TM}}(\run)+1)\cdot \size\istr\cdot 
		\log{\size\istr})$.
		
		\item 
		\emph{Divergence}: diverges, then $\enc M\, \enc\istr$ is 
		$\tobdet$-divergent.
		
		\item 
		\emph{\SpKAM}: the space used by the \SpKAM to simulate the evaluation 
		of  point 1 is $\bigo{S_{\textrm{TM}}(\run) +\log\size\istr}$ if $\enc 
		M$ 
		and $\enc\istr$ have separate address spaces.
	\end{enumerate}
\end{thm}
\begin{proof}
	The first two points are proved in~\cite{log-encoding}. We concentrate on the third point.
	
	We simply evaluate $\enc M\, \enc\istr$ with the Space KAM.
	
	\begin{align*}
	\begin{array}{l|l|ll}
	\mathsf{Term}   & 
	\mathsf{Env} & \mathsf{Stack}\\
	\hhline{---} \rule{0pt}{2.3ex}
	\enc M\, \enc\istr\defeq\inits(\transs^M(\finals(\la \var\var)))\enc\istr 
	& 
	\stempty & \stempty & \too^*_{\bigo{\log(\size\istr)}}\text{ 
		(Lemma \ref{l:init})}\\
	\transs(\finals(\la \var\var)) & \stempty & 
	\initconfigs^{\kop} & \too^*_{\bigo{S_{\textrm{TM}}(\run) 
			+\log\size\istr}} \text{ 
		(Lemma \ref{l:trans})}\\
	\finals(\la\var\var) & \stempty & \configtwo^\kop & 
	\too^*_{\bigo{S_{\textrm{TM}}(\run) +\log\size\istr}} \text{ (Lemma 
		\ref{l:final})}\\
	\enc \bool & \stempty & \stempty
	\end{array}\\
	\tag*{\qedhere}
\end{align*}
\end{proof}

\paragraph{Init and Final.} Here, we provide the execution traces for the combinators $\inits$ and $\finals$.

\begin{lem}\label{l:init}
	$\spkamstate{\inits\,\cont\,\enc\inputstr}{ 
		\stempty}{\stempty}\tospkam^{\bigo{1}}\spkamstate\cont 
	\stempty{\initconfigs^{\kop}}$ and consumes space 
	$\Theta(\log(\size\inputstr))$.
\end{lem}
\begin{proof}
	The \SpKAM execution is in Figure~\ref{fig:init}.
	The space bound is immediate by inspecting the execution.
\end{proof}
\begin{figure}
	\begin{sideways} {\scriptsize\(
		\begin{array}{l@{\hspace{1cm}}|l@{\hspace{1cm}}|ll}
		\mathsf{Term}   & 
		\mathsf{Env} & \mathsf{Stack}\\
		\hhline{---} \rule{0pt}{2.3ex}
		\inits\,\cont\,\enc\inputstr & \stempty & \stempty & \tokamseanv \\
		\inits\,\cont & \stempty & (\enc\inputstr,\stempty) & \tokamseanv\\
		\inits\defeq(\la {d} \la {e} \la {f} \la{\cont'} \la{\inputstr'} 
		\cont'
		\tuple{\inputstr', d \csep e ,
			\cods{\elemblank}, f \csep 
			\cods{\statein}})\enc{\ntostr 
			0}\enc{\varepsilon} \enc{\varepsilon} & \stempty 
		& 
		\overbrace{(\cont,\stempty)}^{\cont^{\kop}}\cons 
		\overbrace{(\enc\inputstr,\stempty)}^{\enc\inputstr{^{\kop}}} &\tokamseanv^3\\
		\la {d} \la {e} \la {f} \la{\cont'} \la{\inputstr'} 
		\cont' \tuple{\inputstr', d \csep e ,\cods{\elemblank}, f \csep 
			\cods{\statein}} & \stempty & 
		\overbrace{(\enc{\ntostr 0},\stempty)}
		^{\enc{\ntostr 0}^{\kop}}\cons 
		\overbrace{(\enc\ems,\stempty)}^{\enc{\ems}^{\kop}}
		\cons(\enc\ems,\stempty)\cons\cont^{\kop}\cons
		\enc\inputstr^{\kop} & \tokambnw^5\\
		\cont' \tuple{\inputstr', d \csep e ,\cods{\elemblank}, f \csep 
			\cods{\statein}} & \esub{\cont'} 
		{\cont^{\kop}}\cons\overbrace{\esub {\inputstr'} 
			{{\enc\inputstr}^{\kop}}\cons\esub d {\enc{\ntostr 0}^{\kop}} 
			\cons \esub e {\enc{\ems}^{\kop}}\cons \esub f 
			{\enc{\ems}^{\kop}}}^{\initconfigs^\env}  & 
		\stempty & \tokamseanv\\
		\cont' & \esub{\cont'} {\cont^{\kop}} & 
		(\tuple{\inputstr', d \csep e ,\cods{\elemblank}, f \csep 
			\cods{\statein}},\initconfigs^\env) &\tokamsub\\
		\cont & \stempty & \underbrace{(\tuple{\inputstr', d \csep e 
				,\cods{\elemblank}, f \csep 
				\cods{\statein}},\initconfigs^\env)}_{\initconfigs^{\kop}}
		\end{array}
		\)}\end{sideways} 
	\caption{\SpKAM execution of the $\inits$ combinator.}
	\label{fig:init}
\end{figure}

\begin{lem}\label{l:final} Let $\config$ be a final configuration, \ie 
	$\config\defeq(\inputstr,\counter,\str,\elem,\strtwo,\statefin)$ where 
	$\statefin\in\Statesfin$. Then
	\[\spkamstate{\finals(\la \var 
		\var)}\stempty{\config^\kop}\tospkam^{\bigo{1}}
	\begin{cases}
	\spkamstate{\la \var\la\vartwo\var}\stempty\stempty & \text{ if } 
	\statefin=\statefint\\
	\spkamstate{\la \var\la\vartwo\vartwo}\stempty\stempty & \text{ if } 
	\statefin=\statefinf
	\end{cases}\]
	Moreover, the space consumption is $\Theta(\size{\config^\kop})$.
\end{lem}
\begin{proof}
	Let us define $\tm\defeq\lambda 
	\inputstr'.\lambda n'.\lambda 
	\wstrl'.\lambda \elem'.\lambda 
	\wstrr'.\lambda \state'. \state' N_1\ldots N_{\size{\States}}\cont'$. We execute the Space KAM on this term in Figure~\ref{fig:final}. 
	
	Two cases. If $\statefin=\statefint$, then:
	\[\begin{array}{l|l|ll}
	\mathsf{Term}   & 
	\mathsf{Env} & \mathsf{Stack}\\
	\hhline{---}
	\cods\statefint\defeq\lambda x_1\ldots\lambda x_{\size\States}.x_i & 
	\stempty & (N_i,\stempty)_{1\leq i\leq\size\States}\cons(\la\var 
	\var,\stempty)&\tokamb^{\size{\States}}\\
	\var_i & \esub{\var_i}{(N_i,\stempty)} & (\la\var \var,\stempty) & \tokamsub\\
	N_i\defeq\la{\cont'}{\cont'(\la\var{\la\vartwo\var)}} & \stempty & 
	(\la\var \var,\stempty) & \tokambnw\\
	{\cont'(\la\var{\la\vartwo\var})} & \esub{\cont'}{(\la\var 
		\var,\stempty)} & \stempty & \tokamseanv\\
	\cont' & \esub{\cont'}{(\la\var \var,\stempty)} & 
	(\la\var\la\vartwo\var,\stempty) & \tokamsub\\
	\la\var \var & \stempty & (\la\var\la\vartwo\var,\stempty) &\tokambnw\\
	\var & \esub{\var}{(\la\var\la\vartwo\var,\stempty)} & \stempty & \tokamsub\\
	\la\var\la\vartwo\var & \stempty & \stempty
	\end{array}\]
	
	If $\statefin=\statefinf$, then:
	\[\begin{array}{l|l|ll}
	\mathsf{Term}   & 
	\mathsf{Env} & \mathsf{Stack}\\
	\hhline{---}
	\cods\statefinf\defeq\lambda x_1\ldots\lambda x_{\size\States}.x_i & 
	\stempty & (N_i,\stempty)_{1\leq i\leq\size\States}\cons(\la\var 
	\var,\stempty)&\tokamb^{\size{\States}}\\
	\var_i & \esub{\var_i}{(N_i,\stempty)} & (\la\var \var,\stempty) & \tokamsub\\
	N_i\defeq\la{\cont'}{\cont'(\la\var{\la\vartwo\vartwo})} & \stempty & 
	(\la\var \var,\stempty) & \tokambnw\\
	{\cont'(\la\var{\la\vartwo\vartwo})} & \esub{\cont'}{(\la\var 
		\var,\stempty)} & \stempty & \tokamseanv\\
	\cont' & \esub{\cont'}{(\la\var \var,\stempty)} & 
	(\la\var\la\vartwo\vartwo,\stempty) & \tokamsub\\
	\la\var \var & \stempty & (\la\var\la\vartwo\vartwo,\stempty) &\tokambnw\\
	\var & \esub{\var}{(\la\var\la\vartwo\vartwo,\stempty)} & \stempty & \tokamsub\\
	\la\var\la\vartwo\vartwo & \stempty & \stempty
	\end{array}\]
	The space bound is immediate by inspecting the execution.
\end{proof}

\begin{figure}
	\begin{sideways}
	{\scriptsize$
		\begin{array}{l|l|ll}
		\mathsf{Term}   & 
		\mathsf{Env} & \mathsf{Stack}\\
		\hhline{---}
		\finals(\la \var \var) & \stempty & \config^\kop & \tokamseanv\\
		\finals\defeq\la {\cont'} \la {\config'}\config'\tm & \stempty & 
		(\la\var 
		\var,\stempty)\cons\config^\kop & \tokambnw^2\\
		\config'\tm & \esub{\config'}{\config^\kop}\cons\esub{\cont'}{(\la\var 
			\var,\stempty)} 
		& \stempty & \tokamseanv\\
		\config' & \esub{\config'}{\config^\kop} & (\tm,\esub{\cont'}{(\la\var 
			\var,\stempty)}) & \tokamsub\\
		\config^\kop\defeq\la\var\var fcm\cods\elem d\cods\statefin & \esub 
		f{(\enc\inputstr,\stempty)},\esub c {n^\kop}, \esub m {\str{^\kop}}, 
		\esub 
		d 
		{\strtwo^\kop} & (\tm,\esub{\cont'}{(\la\var 
			\var,\stempty)}) & \tokambnw\\
		\var fcm\cods\elem d\cods\statefin & \esub\var 
		{(\tm,\esub{\cont'}{(I,\stempty)})} \cons \esub 
		f{(\enc\inputstr,\stempty)},\esub c {n^\kop}, \esub m {\str{^\kop}}, 
		\esub 
		d 
		{\strtwo^\kop} & \stempty & \tokamsea^6\\
		\var & \esub\var {(\tm,\esub{\cont'}{(\la\var \var,\stempty)})} & 
		(\enc\inputstr,\stempty)\cons\counter^\kop\cons\str^\kop\cons 
		(\cods\elem,\stempty)\cons\strtwo^\kop\cons(\cods\statefin,\stempty) & \tokamsub\\
		\tm\defeq\lambda 
		\inputstr'.\lambda n'.\lambda 
		\wstrl'.\lambda \elem'.\lambda 
		\wstrr'.\lambda \state'. \state' N_1\ldots N_{\size{\States}}\cont' & 
		\esub{\cont'}{(\la\var 
			\var,\stempty)} & 
		(\enc\inputstr,\stempty)\cons\counter^\kop\cons\str^\kop\cons 
		(\cods\elem,\stempty)\cons\strtwo^\kop\cons(\cods\statefin,\stempty) &\tokamb^6\\
		q'N_1\ldots N_{\size{Q}}k' & \esub{\state'}{(\cods\statefin,\stempty)} 
		\cons 
		\esub{\cont'}{(\la\var \var,\stempty)} & \stempty &\tokamseanv^{1+\size\States}\\
		\state' & \esub{\state'}{(\cods\statefin,\stempty)} & 
		(N_i,\stempty)_{1\leq 
			i\leq\size\States}\cons(\la\var \var,\stempty) & \tokamsub\\
		\cods\statefin\defeq\lambda x_1\ldots\lambda x_{\size\States}.x_i & 
		\stempty & (N_i,\stempty)_{1\leq i\leq\size\States}\cons(\la\var 
		\var,\stempty)
		\end{array}
		$}\end{sideways}
	\caption{First part of the \SpKAM execution of the $\finals$ combinator.}
	\label{fig:final}
\end{figure}

\paragraph{Transition Function.} Here we execute the combinator $\transs^M$ (abbreviated to $\transs$ for readability reasons), \ie the main ingredient of the encoding. First, we execute the initialization steps.

\begin{lem}
	$\spkamstate{\transs\,\cont}\stempty{\initconfigs^{\kop}}\tospkam^{\bigo{1}}\spkamstate{\theta}\stempty{(\theta,\stempty)\cons(\transaux,\stempty)\cons (\cont,\stempty)
		\cons\initconfigs^{\kop}}$ in\linebreak space $\bigo{\log(\size\istr)}$.
\end{lem}

% !TeX spellcheck = en_US
% !TEX root = ../../main.tex
\begin{proof}
	\begin{align*}
	\begin{array}{l|l|ll}
	\mathsf{Term}   & 
	\mathsf{Env} & \mathsf{Stack}\\
	\hhline{---} \rule{0pt}{2.3ex}
	\transs\cont & \stempty & \initconfigs^{\kop} & \tokamseanv\\
	\transs\defeq\fix\transaux & \stempty & 
	(\cont,\stempty)\cons\initconfigs^{\kop}&\to^*\text{ (Lemma }\ref{l:fix})\\
	\theta & \stempty & (\theta,\stempty)\cons(\transaux,\stempty)\cons(\cont,\stempty)
	\cons\initconfigs^{\kop} &
	\end{array} \\
	\tag*{\qedhere}\end{align*}
\end{proof}
Then, we prove the main lemma about the $\transs$ combinator. It simply states that each transition of the Turing machine is simulated by the \SpKAM, with the right, \ie linear, space complexity overhead.
\begin{lem}Let $\config$ be a Turing machine configuration. Then:
	\begin{itemize} 
		\item if $\config$ is a final configuration, then
	$\spkamstate{\theta}\stempty{(\theta,\stempty)\cons(\transaux,\stempty)\cons\cont^{\kop}
	 \cons\config^{\kop}}\tospkam^{\bigo{1}}
	\spkamstate{\cont}\stempty{\config^{\kop}}$ in space 
	$\bigo{\size{\config^\kop}}$;
		\item otherwise if $\config \tomachtur  \configtwo$, then\\
		$\spkamstate{\theta}\stempty{(\theta,\stempty)\cons(\transaux,\stempty)\cons\cont^{\kop}
			\cons\config^{\kop}}\tospkam^{\bigo{1}}
		\spkamstate{\theta}\stempty{(\theta,\stempty)\cons(\transaux,\stempty)\cons\cont^{\kop}
			\cons\configtwo^{\kop}}$ in space\linebreak $\bigo{\size{\config^\kop}}$.
\end{itemize}
\end{lem}
\begin{proof}
	\renewcommand{\transaux}{\mathtt{tx}}
	The first part of the proof is common to both points. \\
	Let us define $\transaux\defeq\mathtt{transaux}$ and
	$\tm\defeq\la{\istr'}\la{\counter'}\la{\wstrl'}\la{\elem'}\la{\wstrr'}\la{\state'}\lookupl\,K
	 \istr' \counter'$. The execution is in Figure~\ref{fig:casezero}.
	\begin{figure}\begin{sideways}{\scriptsize\(
	\begin{array}{l|l|ll}
	\mathsf{Term}   & 
	\mathsf{Env} & \mathsf{Stack}\\
	\hhline{---}
	\theta\defeq\la\var \la\vartwo \vartwo(\var \var \vartwo ) & \stempty & 
	(\theta,\stempty)\cons(\transaux,\stempty)\cons\cont^{\kop} \cons\config^{\kop} & 
	\tokambnw^2\\
	\vartwo(\var \var \vartwo) & \esub \var {(\theta,\stempty)}\cons\esub 
	\vartwo {(\transaux,\stempty)}
		& \cont^{\kop} \cons\config^{\kop} & \tokamseanv\\
	\vartwo & \esub\vartwo{(\transaux,\stempty)} & \overbrace{(\var \var 
		\vartwo,\esub\vartwo{(\transaux,\stempty)}\cons 
		\esub\var{(\theta,\stempty)})}^{\fix^{\kop}}\cons\cont^{\kop} 
		\cons\config^{\kop} & \tokamsub\\
	\transaux & \stempty & \fix^{\kop}\cons\cont^{\kop} \cons\config^{\kop}& =\\
	\lambda x.\la\cont\lambda \config'.\config'\tm & \stempty &  
	\fix^{\kop}\cons\cont^{\kop} \cons\config^{\kop}&\tokambnw^2\\
	\lambda \config'.\config'\tm & \overbrace{\esub 
	{\cont' }
	{\cont^{\kop}}\cons \esub\var {\fix^{\kop}}}^\env & \config^{\kop} & \tokambnw\\
	\config'\tm
	&\esub{\config'}{\config^{\kop}}\cons\env & \stempty & \tokamseanv\\
	\config' & \esub{\config'}{\config^{\kop}} & (\tm,\env) & \tokamsub\\
	\la\var\var fcm\cods{\elem_j}
		d\cods{\state_g} & \esub f{(\enc\inputstr,\stempty)},\esub c 
		{n^{\kop}}, 
	\esub m {\str{^{\kop}}}, \esub d {\strtwo^{\kop}} & (\tm,\env)&\to^7\\
	\var & \esub \var {(\tm,\env)} & 
	\overbrace{(\enc\inputstr,\stempty)}^{\enc\inputstr^{\kop}}\cons 
	n^{\kop}\cons 
	\str^{\kop}\cons \overbrace{(\cods{\elem_j},\stempty)}^{\cods{\elem}^{\kop}}\cons 
	r^{\kop}\cons\overbrace{(\cods{\state_g},\stempty)}^{\cods\state^{\kop}} &\tokamsub\\
	\la{\istr'}\la{\counter'}\la{\wstrl'}\la{\elem'}\la{\wstrr'}\la{\state'}\lookupl\,K
	\istr' \counter' & \env & 
	\enc\inputstr^{\kop}\cons n^{\kop}\cons \str^{\kop}\cons 
	\cods\elem^{\kop}\cons 
	r^{\kop}\cons\cods\state^{\kop}&\tokambnw^6\\
	\lookupl\,K
	\istr' \counter' & \overbrace{\esub 
	{i'} {\enc 
	i^{\kop}}\cons \esub {n'} {n^{\kop}}\cons \esub {\wstrl} 
	{s^{\kop}}\cons \esub {\elem'} 
	{\cods\elem^{\kop}}\cons 
	\esub {\wstrr} {r^{\kop}}\cons \esub {\state'} {\cods 
	q^{\kop}}\cons\env}^{\envtwo}& 
	\stempty & \tokamsea^3\\
	\lookupl & \stempty & 
	(K,\envtwo)\cons  \enc i^{\kop}\cons n^{\kop} & \to^*\\
	K\defeq \la {\bool'} \bool' 
	A_{0}A_{1}A_{\stsym}A_{\ensym}\elem'\state'\var \cont' \istr' 
	\counter'\wstrl' \wstrr' & 
	\envtwo & (\cods{\elem_i},\stempty) &\tokambnw\\
	\bool' 
	A_{0}A_{1}A_{\stsym}A_{\ensym}\elem'\state'\var \cont' \istr' 
	\counter'\wstrl' \wstrr'& 
	\esub{\bool'}{(\cods{\elem_i},\stempty)},\envtwo & \stempty&\tokamsea^{12}\\
	\bool' & \esub{\bool'}{(\cods{\elem_i},\stempty)} & 
	(A_\bool,\stempty)_{b\in\Boolin} \cons \cods\elem^{\kop} \cons 
	\cods\state^{\kop} \cons 
	\fix^{\kop} \cons \cont^{\kop} \cons \enc i^{\kop} \cons \counter^{\kop} 
	\cons \str^{\kop} \cons \strtwo^{\kop} & \tokamsub
	\\
	\cods{\elem_i}\defeq\lambda \var_0.\lambda \var_1.\lambda 
	\var_\stsym.\lambda \var_\ensym.\var_i & \stempty & 
	(A_\bool,\stempty)_{b\in\Boolin}\cons \cods\elem^{\kop} \cons 
	\cods\state^{\kop} \cons 
	\fix^{\kop} \cons \cont^{\kop} \cons \enc i^{\kop} \cons \counter^{\kop} 
	\cons \str^{\kop} \cons \strtwo^{\kop} & \to^5\\
	A_i\defeq \la{\elem'} \elem' B_{i,0} B_{i,1}B_{i,\elemblank} & \stempty & 
	\cods\elem^{\kop} \cons \cods\state^{\kop} \cons 
	\fix^{\kop} \cons \cont^{\kop} \cons \enc i^{\kop} \cons \counter^{\kop} 
	\cons \str^{\kop} \cons \strtwo^{\kop} & \tokambnw\\
	\elem' B_{i,0} B_{i,1}B_{i,\elemblank} & \esub{\elem'}{\cods\elem^{\kop}} & 
	\cods\state^{\kop} \cons 
	\fix^{\kop} \cons \cont^{\kop} \cons \enc i^{\kop} \cons \counter^{\kop} 
	\cons \str^{\kop} \cons \strtwo^{\kop}&\tokamseanv^3\\
	\elem' & \esub{\elem'}{\cods\elem^{\kop}} & 
	(B_{i,\bool},\stempty)_{\bool\in\Boolb} \cons \cods\state^{\kop} \cons 
	\fix^{\kop} \cons \cont^{\kop} \cons \enc i^{\kop} \cons \counter^{\kop} 
	\cons \str^{\kop} \cons \strtwo^{\kop} & \tokamsub\\
	\cods{\elem_j}\defeq \lambda \var_0.\lambda \var_1.\lambda 
	\var_\elemblank.\var_j &  \stempty & 
	(B_{i,\bool},\stempty)_{\bool\in\Boolb} \cons \cods\state^{\kop} \cons 
	\fix^{\kop} \cons \cont^{\kop} \cons \enc i^{\kop} \cons \counter^{\kop} 
	\cons \str^{\kop} \cons \strtwo^{\kop}&\to^4\\
	B_{i,j}\defeq \la{\state'}\state' C_{i,j,\state_{1}}\ldots 
	C_{i,j,\state_{\size\States}} & \stempty & 
	 \cods\state^{\kop} \cons 
	 \fix^{\kop} \cons \cont^{\kop} \cons \enc i^{\kop} \cons \counter^{\kop} 
	 \cons \str^{\kop} \cons \strtwo^{\kop} & \tokambnw\\
	\state' C_{i,j,\state_{1}}\ldots 
	C_{i,j,\state_{\size\States}} & \esub{\state'}{\cods\state^\kop} & 
	\fix^{\kop} \cons \cont^{\kop} \cons \enc i^{\kop} \cons \counter^{\kop} 
	\cons \str^{\kop} \cons \strtwo^{\kop} & 
	\tokamseanv^{\size\States}\\
	\state' & \esub{\state'}{\cods\state^\kop} & 
	(C_{i,j,\state_g},\stempty)_{1\leq g \leq \size\alpone}\cons
	\fix^{\kop} \cons \cont^{\kop} \cons \enc i^{\kop} \cons \counter^{\kop} 
	\cons \str^{\kop} \cons \strtwo^{\kop} & \tokamsub\\
	\cods{\state_g}\defeq\lambda x_1\ldots\lambda x_{\size\States}.x_g & 
	\stempty & (C_{i,j,\state_g},\stempty)_{1\leq g \leq \size\alpone}\cons
	\fix^{\kop} \cons \cont^{\kop} \cons \enc i^{\kop} \cons \counter^{\kop} 
	\cons \str^{\kop} \cons \strtwo^{\kop}&\to^{1+\size\States}\\
	C_{i,j,\state_g} & \stempty & 
	\fix^{\kop} \cons \cont^{\kop} \cons \enc i^{\kop} \cons \counter^{\kop} 
	\cons \str^{\kop} \cons \strtwo^{\kop}\\
	\end{array}\)}\end{sideways}
\caption{The first part of the \SpKAM execution of the combinator $\transs$.}
\label{fig:casezero} \end{figure} 
	Cases of the transition to apply:
	\begin{itemize}
	\item \emph{No transition}, that is, $\config$ is a final configuration, 
	which happens when $\state_g\in\Statesfin$. \\We have
    $C_{i,j,\state_g}\defeq\la\var \la{\cont'} \la{\istr'} \la{\counter'} 
    \la{\wstrl'} \la{\wstrr'}\cont' \tuple{\istr',\counter' \csep 
    \wstrl',\cods{\elem_j}, 
    	\wstrr' \csep \cods{\state_g} }$, and\\ $\config^\kop\defeq(\tuple{\istr',\counter' \csep 
    	\wstrl',\cods{\elem_j}, 
    	\wstrr' \csep \cods{\state_g} },\env_2)$, where $E_2\defeq\esub {\wstrl'} {s^{\kop}}\cons\esub 
    {\wstrr'} 
{r^{\kop}}\cons\esub {i'} {\enc i^{\kop}}\cons\esub {n'} 
{n^{\kop}}$
    
    \[\scriptsize
    \begin{array}{l|l|ll}
    	\mathsf{Term}   & 
    	\mathsf{Env} & \mathsf{Stack}\\
    	\hhline{---} \rule{0pt}{2.6ex}
    	C_{i,j,\state_g} & \stempty& \fix^{\kop} \cons \cont^{\kop} \cons \enc 
    	i^{\kop} \cons \counter^{\kop} 
    	\cons \str^{\kop} \cons \strtwo^{\kop}&\tokamb^6\\
    	\cont' \tuple{\istr',\counter' \csep 
    		\wstrl',\cods{\elem_j}, 
    		\wstrr' \csep \cods{\state_g} } & \esub {k'} 
    	{k^{\kop}}\cons\overbrace{\esub {\wstrl'} {s^{\kop}}\cons\esub 
    	{\wstrr'} 
    	{r^{\kop}}\cons\esub {i'} {\enc i^{\kop}}\cons\esub {n'} 
    	{n^{\kop}}}^{\env_2}& 
    	\stempty & \tokamseanv\\
    	\cont' & \esub {k'} {k^{\kop}=:(k,\stempty)} & 
    	\config^{\kop} & \tokamsub\\
    	\cont & \env & \config^{\kop}
    \end{array}
	\]
	
	%%%% 
	\item \emph{The heads do not move}, that is, 
	$\delta(\elem_i,\elem_j,\state_g)=(0 \csep \elem_h,\downarrow \csep 
	\state_l)$. We set\\ $\configtwo^\kop\defeq(\tuple{\istr',\counter''\csep 
		\wstrl',\cods{\elem_h}, \wstrr' \csep 
		\cods{\state_l}},E_{2})$, where $E_2\defeq\esub {\wstrl'} {s^{\kop}}\cons\esub 
	{\wstrr'} 
{r^{\kop}}\cons\esub {i'} {\enc i^{\kop}}\cons\esub {n''} 
{n^{\kop}}$.

 {\scriptsize\[
	\arraycolsep=1.5pt
    \begin{array}{l|l|ll}
    	\mathsf{Term}   & 
    	\mathsf{Env} & \mathsf{Stack}\\
    	\hhline{---} \rule{0pt}{2.6ex}
    	C_{i,j,\state_g}\defeq \la\var \la{\cont'} \la{\istr'} \la{\counter'} 
    	\la{\wstrl'} \la{\wstrr'}S\counter'& \stempty& \fix^{\kop} \cons 
    	\cont^{\kop} \cons \enc 
    	i^{\kop} \cons \counter^{\kop} 
    	\cons \str^{\kop} \cons \strtwo^{\kop} & \tokambnw^6\\
	
	S\counter'
    & \esub\var{\fix^{\kop}}\cons \esub {k'} 
    {k^{\kop}}\cons\esub {\wstrl'} {s^{\kop}}\cons\esub 
    	{\wstrr'} 
    	{r^{\kop}}\cons\esub {i'} {\enc i^{\kop}}\cons\esub {n'} 
    	{n^{\kop}}
    & \stempty & \tokamseav\\
    S \defeq \la {\counter''} \var\cont' \tuple{\istr',\counter''\csep 
    	\wstrl',\cods{\elem_h}, \wstrr' \csep \cods{\state_l}}  & 
    	\esub\var{\fix^{\kop}}\cons \esub {k'} 
    {k^{\kop}}\cons\esub {\wstrl'} {s^{\kop}}\cons\esub 
    	{\wstrr'} 
    	{r^{\kop}}\cons\esub {i'} {\enc i^{\kop}} & \counter^{\kop} & \tokambnw\\
    \var\cont' \tuple{\istr',\counter''\csep 
    	\wstrl',\cods{\elem_h}, \wstrr' \csep \cods{\state_l}} &  
    	\esub\var{\fix^{\kop}}\cons \esub {k'} 
    {k^{\kop}}\cons\overbrace{\esub {\wstrl'} {s^{\kop}}\cons\esub 
    {\wstrr'} 
    {r^{\kop}}\cons\esub {i'} {\enc i^{\kop}}\cons\esub {n''} 
    {n^{\kop}}}^{\env_2} & \stempty& \tokamsea\\
	x
    & \esub \var {\fix^{\kop}}
    & \cont^{\kop}\cons\configtwo^{\kop} & \tokamsub
    \\
    \var \var \vartwo
	&\esub\vartwo{(\transaux,\stempty)}\cons\esub\var{(\theta,\stempty)}
	&\cont^{\kop}\cons \configtwo^{\kop}
	& \tokamseav^{2}
    \\
    \var
	&\esub\var{(\theta,\stempty)}
	&(\theta,\stempty) \cons (\transaux,\stempty) \cons \cont^{\kop} \cons \configtwo^{\kop} & \tokamsub
    \\
    \theta
	&\stempty
	&(\theta,\stempty) \cons (\transaux,\stempty) \cons \cont^{\kop} \cons \configtwo^{\kop}
    \end{array}
	\]}
	
		\item \emph{The heads move right}, that is, 
	$\delta(\elem_i,\elem_j,\state_g)=(1 \csep \elem_h,\rightarrow \csep 
	\state_l)$. The execution of the first part is in Figure~\ref{fig:caseone}.
	\begin{figure}\begin{sideways}{\scriptsize\(
	\begin{array}{l|l|ll}
		\mathsf{Term}   & 
		\mathsf{Env} & \mathsf{Stack}\\
		\hhline{---} \rule{0pt}{2.6ex}
		C_{i,j,\state_g}\defeq \la\var \la{\cont'} \la{\istr'} 
		\la{\counter'} 
		\la{\wstrl'} \la{\wstrr'}\succl  R \counter' & \stempty& \fix^{\kop} \cons 
		\cont^{\kop} \cons \enc 
		i^{\kop} \cons \counter^{\kop} 
		\cons \str^{\kop} \cons \strtwo^{\kop} & \tokambnw^6\\
		
		\succl  R \counter'
		& \overbrace{\esub\var{\fix^{\kop}}\cons \esub {k'} 
		{k^{\kop}}\cons\esub {\wstrl'} {s^{\kop}}\cons\esub 
		{\wstrr'} 
		{r^{\kop}}\cons\esub {i'} {\enc i^{\kop}}}^{\env_2}\cons\esub {n'} 
		{n^{\kop}}
		& \stempty & \tokamsea^2\\
		\succl & \stempty & (R,\env_2)\cons n^{\kop} & \to^*\\
		R \defeq \la{\counter''}\wstrr' R_0^{\state_l,\elem_h} 
		R_1^{\state_l,\elem_h} 
		R_{\elemblank}^{\state_l,\elem_h}R_{\ems}^{\state_l,\elem_h} \var \cont' 
		\istr' \counter'' \wstrl'& \env_2 & m^{\kop}\defeq(n+1)^{\kop} &\tokambnw\\
		\wstrr' R_0^{\state_l,\elem_h} 
		R_1^{\state_l,\elem_h} 
		R_{\elemblank}^{\state_l,\elem_h}R_{\ems}^{\state_l,\elem_h} \var \cont' 
		\istr' \counter'' \wstrl' & \esub\var{\fix^{\kop}}\cons \esub {k'} 
		{k^{\kop}}\cons\esub {\wstrl'} {s^{\kop}}\cons\esub 
		{\wstrr'} 
		{r^{\kop}}\cons\esub {i'} {\enc i^{\kop}}\cons\esub {n''} 
		{m^{\kop}} & \stempty & \tokamseav^9\\
		\wstrr' & \esub {\wstrr'} {r^{\kop}} & 
		(R^{\state_l,\elem_h}_\texttt{x},\stempty)_{\texttt{x}\in\{0,1,\elemblank,\ems\}}
		 \cons \fix^{\kop} \cons 
		 \cont^{\kop} \cons \enc i^{\kop} \cons m^{\kop} \cons \str^{\kop}\\
	\end{array}\)}\end{sideways}
	\caption{The \SpKAM execution of the beginning of ``the heads move right''.}
	\label{fig:caseone} \end{figure} 
	Two cases.
	\begin{itemize}
		\item $\strtwo=\ems$. Define $\tm\defeq (\la d \lambda \wstrl'.\var\cont' 
		\tuple{\istr',\counter'\csep \wstrl',\cods{\elemblank},d \csep  
		\cods{\state_l}})\enc\ems$. The execution is in Figure~\ref{fig:casetwo}.
		\begin{figure}\begin{sideways}{\scriptsize\(
		\begin{array}{l|l|ll}
			\mathsf{Term}   & 
			\mathsf{Env} & \mathsf{Stack}\\
			\hhline{---} \rule{0pt}{2.6ex}
			\wstrr' & \esub {\wstrr'} {r^{\kop}} & 
			(R^{\state_l,\elem_h}_\texttt{x},\stempty)_{\texttt{x}\in\{0,1,\elemblank,\ems\}}
			\cons \fix^{\kop} \cons 
			\cont^{\kop} \cons \enc i^{\kop} \cons m^{\kop} \cons 
			\str^{\kop}&\to^5\\
			R^{\state_l,\elem_h}_\ems \defeq \la\var \la{\cont'} \la{\istr'} 
			\la{\counter'}
			\appendchar{\elem_h}\tm & \stempty & \fix^{\kop} \cons 
				\cont^{\kop} 
			\cons \enc i^{\kop} \cons m^{\kop} \cons \str^{\kop}&\tokambnw^4\\
			\appendchar{\elem_h}\tm & \overbrace{\esub\var{\fix^{\kop}}\cons \esub 
			{k'} 
			{k^{\kop}}\cons\esub {i'} {\enc i^{\kop}}\cons\esub {n'} 
			{m^{\kop}}}^{\env_2}& \str^{\kop} & \tokamseanv\\
			\appendchar{\elem_h} & \stempty & (\tm,\env_2)\cons \str^{\kop} & \to^*\\
			\tm \defeq (\la d \lambda \wstrl'.\var\cont' 
			\tuple{\istr',\counter'\csep \wstrl',\cods{\elemblank},d \csep  
			\cods{\state_l}})\enc\ems & \env_2 & \str_h^\kop\defeq 
			(\elem_h\cons\str)^{\kop} & \tokamsea\\
			\la d \lambda \wstrl'.\var\cont' 
			\tuple{\istr',\counter'\csep \wstrl',\cods{\elemblank},d \csep  
				\cods{\state_l}} & \env_2 & (\enc\ems,\stempty)\cons\str_h^\kop & 
				\tokambnw^2\\
			\var\cont' \tuple{\istr',\counter'\csep \wstrl',\cods{\elemblank},d 
			\csep  \cods{\state_l}} & \esub\var{\fix^{\kop}}\cons \esub 
			{k'} {k^{\kop}}\cons\overbrace{\esub {i'} {\enc i^{\kop}}\cons\esub 
			{n'} 
			{m^{\kop}} \cons \esub d {(\enc\ems,\stempty)}\cons 
			\esub{\wstrl'} {\str_h^\kop}}^{\env_3} & \stempty & \tokamsea^2\\
			\var & \esub\var{\fix^{\kop}} & k^{\kop} \cons 
			\overbrace{(\tuple{\istr',\counter'\csep \wstrl',\cods{\elemblank},d 
				\csep  \cods{\state_l}},\env_3)}^{\configtwo^\kop} & \tokamsub\\
			\var \var \vartwo
			&\esub\vartwo{(\transaux,\stempty)}\cons\esub\var{(\theta,\stempty)}
			&\cont^{\kop}\cons \configtwo^{\kop}
			& \tokamseav^{2}
			\\
			\var
			&\esub\var{(\theta,\stempty)}
			&(\theta,\stempty) \cons (\transaux,\stempty) \cons \cont^{\kop} \cons 
			\configtwo^{\kop} &\tokamsub
			\\
			\theta
			&\stempty
			&(\theta,\stempty) \cons (\transaux,\stempty) \cons \cont^{\kop} \cons 
			\configtwo^{\kop}
		\end{array}\)}\end{sideways}
	\caption{The \SpKAM execution of the sequel of ``the heads move right'', case $r=\varepsilon$.}
	\label{fig:casetwo} \end{figure} 
	 
	\item $\strtwo=\elem''\cons\strtwo'$. Define $\tm\defeq \lambda 
	\wstrl'.\var\cont' 
	\tuple{\istr',\counter' \csep \wstrl',\cods{\elem''},\wstrr' \csep 
		\cods{\state_l}}$. The execution is in Figure~\ref{fig:casethree}.
	\begin{figure}\begin{sideways}{\scriptsize\(
	\begin{array}{l|l|ll}
		\mathsf{Term}   & 
		\mathsf{Env} & \mathsf{Stack}\\
		\hhline{---} \rule{0pt}{2.6ex}
		\wstrr' & \esub {\wstrr'} {r^{\kop}} & 
		(R^{\state_l,\elem_h}_\texttt{x},\stempty)_{\texttt{x}\in\{0,1,\elemblank,\ems\}}
		\cons \fix^{\kop} \cons 
		\cont^{\kop} \cons \enc i^{\kop} \cons m^{\kop} \cons 
		\str^{\kop} & \tokamsub\\
		\la {\var_0} \la {\var_1} \la {\var_\elemblank} \la {\vartwo} 
		\var_{i_{\elem''}}\varthree & \esub\varthree{\strtwo'^\kop} & 
		(R^{\state_l,\elem_h}_\texttt{x},\stempty)_{\texttt{x}\in\{0,1,\elemblank,\ems\}}
		\cons \fix^{\kop} \cons \cont^{\kop} \cons \enc i^{\kop} \cons m^{\kop} 
		\cons \str^{\kop} & \to^6\\
		R^{\state_l,\elem_h}_{\elem''}\defeq \la{\wstrr'} \la\var \la{\cont'} 
		\la{\istr'} \la{\counter'}\appendchar{\elem_h}\tm & \stempty & 
		\strtwo'^\kop \cons 
		\fix^{\kop} \cons \cont^{\kop} \cons \enc i^{\kop} \cons m^{\kop} \cons 
		\str^{\kop}& \tokambnw^5\\
		\appendchar{\elem_h}\tm& \overbrace{\esub{\wstrr'}{\strtwo'^\kop}\cons 
		\esub\var{\fix^{\kop}}\cons 
		\esub {k'} {k^{\kop}}\cons\esub {i'} {\enc i^{\kop}}\cons\esub {n'} 
			{m^{\kop}}}^{\env_2} & \str^{\kop} & \tokamseanv\\
		\appendchar{\elem_h} & \stempty & (\tm,\env_2)\cons\str^{\kop} & \to^*\\
		\tm \defeq \lambda 
		\wstrl'.\var\cont' \tuple{\istr',\counter' \csep 
		\wstrl',\cods{\elem''},\wstrr' \csep \cods{\state_l}} & \env_2 & 
		\str_h^\kop\defeq(\elem_h\cons\str)^\kop & \tokambnw\\
		\var\cont' \tuple{\istr',\counter' \csep \wstrl',\cods{\elem''},\wstrr' 
		\csep \cods{\state_l}} & \esub\var{\fix^{\kop}}\cons \esub 
	{k'} {k^{\kop}}\cons\overbrace{\esub {i'} {\enc i^{\kop}}\cons\esub {n'} 
	{m^{\kop}} \cons \esub{\wstrr'}{\strtwo'^\kop}\cons \esub{\wstrl'} 
	{\str_h^\kop}}^{\env_3} & \stempty & \tokamsea^2\\
	\var & \esub\var{\fix^{\kop}} & \cont^\kop \cons 
	\overbrace{(\tuple{\istr',\counter' \csep \wstrl',\cods{\elem''},\wstrr' 
			\csep \cods{\state_l}},\env_3}^{\configtwo^\kop} & \tokamsub\\
		\var \var \vartwo
		&\esub\vartwo{(\transaux,\stempty)}\cons\esub\var{(\theta,\stempty)}
		&\cont^{\kop}\cons \configtwo^{\kop}& \tokamseav^{2}\\
		\var &\esub\var{(\theta,\stempty)}&(\theta,\stempty) \cons 
		(\transaux,\stempty) \cons \cont^{\kop} \cons \configtwo^{\kop} & \tokamsub\\
		\theta&\stempty &(\theta,\stempty) \cons (\transaux,\stempty) \cons 
		\cont^{\kop} \cons \configtwo^{\kop}
	\end{array}\)}\end{sideways}
\caption{The \SpKAM execution of the sequel of ``the heads move right'', case $r=a''\cons r'$.}
\label{fig:casethree} \end{figure} 
	\end{itemize}
	\item All the other cases are almost identical \emph{mutatis mutandis}.
	\end{itemize}
	About the space bound we observe that in the simulations \emph{all} the 
	pointers except for those related to the input part of the state, which are in \emph{fixed} number, are 
	pointers to the machine, and not to the input. Moreover, the space overhead 
	of the simulation of one step of the TM is constant, \ie non input 
	dependent.
\end{proof}

\begin{lem}\label{l:trans}
	If $\run:\config\to^n\configtwo$ and $\configtwo$ is final, then  
	$\spkamstate{\transs\,\cont}\stempty{\initconfigs^{\kop}} 
	\tospkam\spkamstate{\cont}\stempty{\config^{\kop}}$ in space 
	$\bigo{S_{\textrm{TM}}(\run)+\log(\size\istr)}$.
\end{lem}
\begin{proof}
	By a simple induction on $n$, using the two lemmata above, and knowing that 
	$S_{\textrm{TM}}(\run)=\max_{\config\in\run}\size\config$ (we have also to 
	consider that $\size\config=\size{\config^\kop}$, by 
	Lemma~\ref{l:config-size}).
\end{proof}

\end{document}